\documentclass[a4paper,fleqn]{cas-dc}

\usepackage[numbers]{natbib}
\usepackage{amsmath}
\usepackage{braket}
\usepackage{bm}
\usepackage{tikz}
\usetikzlibrary{arrows.meta}
\usepackage{csquotes}
\usepackage{cuted}
\usepackage{amsthm}

\def\tsc#1{\csdef{#1}{\textsc{\lowercase{#1}}\xspace}}
\tsc{WGM}
\tsc{QE}
\tsc{EP}
\tsc{PMS}
\tsc{BEC}
\tsc{DE}
\newcommand{\dd}{\mathrm{d}}
\newtheorem{proposition}{Proposition}
\newtheorem{corollary}{Corollary}
\begin{document}
\let\WriteBookmarks\relax
\def\floatpagepagefraction{1}
\def\textpagefraction{.001}
\shorttitle{Schr\"odinger's real-valued equation revisited}
\shortauthors{O. Passon and B. Rosenow}

\title[mode=title]{Schr\"odinger's real-valued equation revisited}

\author[1]{Oliver Passon}
\author[2]{Bernd Rosenow}

\affiliation[1]{
  organization={Bergische Universit\"at Wuppertal,
  Fakult\"at f\"ur Mathematik und Naturwissenschaften},
  addressline={Gau{\ss}str. 20},
  postcode={42117},
  city={Wuppertal},
  country={Germany}
}

\affiliation[2]{
  organization={Universit\"at Leipzig,
  Institut f\"ur Theoretische Physik},
  addressline={Br\"uderstra{\ss}e 16},
  postcode={04103},
  city={Leipzig},
  country={Germany}
}

 \begin{abstract}
The Schrödinger equation can be rewritten as a second-order equation for a single real-valued scalar field. Taken at face value, however, this one-field reduction obscures several local structures of quantum mechanics: the Born density and current, the role of momentum as a generator, and the usual derivation of the uncertainty relation. We argue that these difficulties do not show that real variables fail. They show that the reduced equation is not, by itself, a complete local representation of the theory. Its coherent real form is obtained by restoring the underlying Hamiltonian phase-space structure, where the missing component reappears as a canonical partner. In this real phase-space formulation, the standard structures reappear locally, and minimal magnetic coupling reveals an internal \(SO(2)\) gauge structure. Thus the symbol \(i\) can be removed, but the symplectic and complex structure it encodes cannot.
\end{abstract}
\maketitle
\section{Introduction}
Famously, the Schrödinger equation 
 \begin{equation}
 i\hbar\,\partial_t \psi(r,t) = \hat H\psi(r,t)
   \label{equ:se}
\end{equation}
contains the imaginary unit, and its solutions are, in general, complex-valued functions. The question of whether this is unavoidable or whether a real-valued formulation is also possible has been asked repeatedly since the early days. As is little known, Erwin Schrödinger had already proposed a real-valued Schrödinger equation in 1926 \cite{schroedinger26}, but he soon lost interest in this approach.

In 1932, Paul Ehrenfest published a charming paper titled {\em Einige die Quantenmechanik betreffende Erkundigungsfragen} (Some queries related to quantum mechanics) \cite{ehrenfest32} in which he asked, among other things, about the origin of complex numbers in quantum mechanics. A published response was given by Wolfgang Pauli \cite{pauli33}, and we will comment on his reply shortly.

Erwin Schrödinger's early approach to a real-valued formulation was revisited and further developed by Robert Chen in the late 1980s and early 1990s \cite{chen89,chen91}. Chen viewed his work as proof that complex and real-valued formulations are equivalent, and recently, his work has attracted renewed interest \cite{callender2023,makris2025}. 

We should point out that the terminology in this area is somewhat treacherous. The phrase ``real quantum mechanics'' can mean different things. In one sense, it denotes quantum theory formulated directly on real Hilbert spaces, with real state spaces, real observables, and the usual tensor-product rule for composite systems. This is the sense relevant to the recent no-go results of Renou et al. and Weilenmann \emph{et al.} ~\cite{renou2021,Weilenmann2025Partial}.

This is not the sense in which the real-valued equations discussed below are real. Schrödinger's early real-valued equation and Chen's later extension of the same idea do not replace the Hilbert-space postulates of quantum mechanics with a new real-Hilbert-space theory. Rather, they attempt to rewrite ordinary Schrödinger dynamics in terms of real functions.

Within real-variable approaches, one must also distinguish whether the role normally played by the imaginary unit is retained explicitly. Stueckelberg's formulation is a real-Hilbert-space formulation equipped with a distinguished operator $J$ satisfying $J^2=-\mathbb 1$ \cite{stueckelberg1960real}. This operator plays the role of multiplication by $i$ and keeps the analog of the complex phase available from the outset. By contrast, the single-component strategy tries to go further: it eliminates the imaginary part and describes the dynamics with a single real scalar field obeying a second-order equation. Schrödinger's early, soon-abandoned real-valued equation of 1926 is the historical origin of this second strategy, while Chen's work gives its later systematic development.

The present paper is concerned primarily with this last possibility, which we will call the \emph{reduced single-field formulation}: the second-order formulation expressed through a real-valued field and its time derivative, once the imaginary component has been removed. We first take this formulation seriously on its own terms and examine where it succeeds and where it becomes obscure. The problems encountered in the reduced description will then motivate the real phase-space reconstruction developed in the second part of the paper, where the missing component reappears as a canonical momentum. 

The phase-space reconstruction developed below interprets the role of \(i\) geometrically and connects with a well-known tradition of geometric formulations of quantum mechanics, in which the complex Hilbert space, or more precisely its projective space of rays, is regarded as a real manifold equipped with compatible symplectic, metric, and complex structures \cite{AshtekarSchilling1999,Heslot1985,Kibble1979,Strocchi1966}. That tradition was mainly concerned with the relation between classical and quantum mechanics and with possible generalizations of quantum theory, especially in connection with gravity. Here, the direction is different: we do not start from the full complex Hilbert space and geometrize it, but from the reduced Schrödinger-Chen single-field equation.

The aim is therefore not to propose a rival to ordinary complex quantum mechanics, nor to rediscover the real symplectic formulation already known from geometric quantum mechanics. We  use the Schrödinger-Chen reduction as a diagnostic tool: by first removing the imaginary part, we can see which structures the complex notation had been carrying and why the phase-plane structure must be restored.

To delimit the contribution explicitly: we claim no novelty for the lifted phase-space formulation itself, which is a special case of the geometric structure cited above. What we believe to be new are (i) the systematic diagnosis of the reduced single-field formulation, including the precise operator-theoretic sense in which density, current, and momentum become non-local in the Cauchy data $(u,\dot u)$ (Secs.~\ref{sec:prob} and \ref{sec:momentum}), together with its sharpening into the state-level Proposition~\ref{prop:statelevel}; (ii) the magnetic generalization of Chen's equation and its gyroscopic form (Sec.~\ref{sec:mag}); (iii) the correction of the Green-function argument of Ref.~\cite{chen91} (Sec.~\ref{sec:alt}); (iv) an explicitly real construction of Kramers degeneracy (Sec.~\ref{subsec:time_reversal_kramers_real}); and (v) the observation that the reconstruction map fails to factorize for composite systems, together with its bearing on the recent no-go debate (Sec.~\ref{sec:nogo_rqt_comparison}).

Only after this reconstruction is in place will we return, in Sec.~\ref{sec:nogo_rqt_comparison}, to the relation between the present real-variable reformulation, real Hilbert-space quantum theory, and the recent debate over network no-go results.

\section{The real-valued Schrödinger equation\label{sec:chen}}
The first time-dependent wave equation, which Schrödinger published in his fourth communication from 1926, was actually real-valued. In slightly modified notation and with $\psi=u+iv$, Eq.~(4) in \cite{schroedinger26} reads:
\begin{equation}
-\hbar^2 \,\partial^2_t u(r,t) =  \left(-\frac{\hbar^2}{2m}\nabla^2 +V(r)\right)^2u(r,t).\label{equ:schr_real}
\end{equation}
Note that this equation is second order in time and fourth order in the spatial coordinates (i.e., the Laplace operator $\nabla^2$ is squared). For some time, Schrödinger even regarded this equation as fundamental \cite{karam2020b}. Apparently, he abandoned this approach fairly quickly, in part because he was unable to extend it to non-conservative systems.

However, Robert Chen \cite{chen89,chen91} could demonstrate that the generalization for time-dependent potentials is possible. His starting point was to split the wave function into its real and imaginary parts $\psi=u+i v$. The Schrödinger equation \eqref{equ:se}  then implies two real-valued equations:
\begin{subequations}\label{grp}
\begin{align}
   \hbar\,\dot u &= \hat H v \label{equ:coup1}\\
  -\hbar\,\dot v &= \hat Hu. \label{equ:coup2}
\end{align}
\label{equ:coup12}
\end{subequations} 
Given the coupling, the imaginary part in Eq.~\eqref{equ:coup1} can be expressed in terms of the real part by applying the inverse operator $\hat{H}^{-1}$:
\begin{equation}
     v =\hbar\hat{H}^{-1} \,\dot u.\label{equ:idurchr}
\end{equation}
Now, Chen could show  that this specifies a unique relation if $u$ is the real part of a normalized solution of the complex Schrödinger equation. Hence, he claimed that for any system $\hat H$, the real part alone carries  the  full information on the state of a quantum system already \cite{chen89}. This claim requires a precision that matters throughout: it is the trajectory $t\mapsto u(\cdot,t)$, and equivalently, the Cauchy pair $(u,\dot u)$ at a single instant, that determines the state; a snapshot $u(\cdot,t_0)$ alone does not, since Eq.~\eqref{equ:idurchr} involves the time derivative.

The operator $\hat{H}^{-1}$ is to be understood on the subspace where it exists, and zero-modes need to be dealt with separately (the functional-analytic assumptions used throughout are collected in Appendix~\ref{app:hinv}). In general, $\hat{H}^{-1}$ is an integral operator. Its non-locality turns out to be the defining feature of the reduced single-field formulation and will be discussed in more detail in several places below (starting with Sec.~\ref{sec:alt}). 

It is now straightforward to derive the differential equation for $u$. If one takes $\partial_t$ on both sides of Eq.~\eqref{equ:coup1}, one gets:
\begin{eqnarray}
    \hbar\,\ddot u &=& \partial_t(\hat H\,v)\\
    &=& \hat{H} \,\dot v   + \dot V v.
\end{eqnarray}
Here $\partial_t(\hat Hv)=\hat H\dot v+\dot Vv$ because the kinetic part of $\hat H$ carries no explicit time dependence, i.e., $\partial_t\hat H=\dot V$. With Eqs.~\eqref{equ:coup2} and \eqref{equ:idurchr}, one arrives at:
\begin{eqnarray}
 -\hbar^2\,\ddot u   = \hat{H}^2 u - \dot V\hbar^2 \hat{H}^{-1} \dot u. \label{equ:chenfull}
\end{eqnarray}
This is Chen's generalization of the real-valued Schrödinger equation to the non-conservative case. For $\partial_t{\hat H}=0$, we recover the earlier real-valued version of the Schrödinger equation \eqref{equ:schr_real}:
\begin{equation}
    -\hbar^2\,\ddot u = \hat{H}^2u,\label{equ:rse}
\end{equation}
and its time-independent form is simply: 
\begin{equation}
    \hat{H}^2u=E^2u.\label{equ:rse-ti}
\end{equation}

The other key ingredient of quantum mechanics is the probability density $\rho=u^2+v^2$, which transforms into the conserved probability $\mathcal P =\int \dd^3r\,\rho$. With Eq.~\eqref{equ:idurchr}, it can be expressed as a function of $u$ and $\dot u$ only:
\begin{equation}
\mathcal P[u,\dot u]=\int_{\mathbb R^3} \dd^3r \left (u^2+\left(\hbar \hat{H}^{-1}\dot u\right)^2\right ).\label{equ:rhonaiv}
\end{equation}

Clearly, this reduced single-field formulation is mathematically much more demanding than the standard theory. In expanded form, Eq.~\eqref{equ:schr_real} turns into: 
\begin{equation}
\begin{aligned}
-\hbar^2 \frac{\partial^2 u}{\partial t^2}&=\Bigg[ \frac{\hbar^4}{4m^2}\nabla^4-\frac{\hbar^2}{m}V\nabla^2-\frac{\hbar^2}{m}\nabla V\cdot\nabla\\
&\quad-\frac{\hbar^2}{2m}\nabla^2 V+V^2\Bigg]u .
\end{aligned}
\label{eq:expanded_real_schrodinger}
\end{equation}
Terms like $\nabla V(r)$ or $\nabla u$ are introduced by the product rule for the Laplacian ($\nabla^2(fg)=f\nabla^2g+2\nabla f\nabla g+g\nabla^2f$). Notably, this equation also involves derivatives of the potential. 
 
Makris and Dargush~\cite{makris2025} have applied advanced numerical methods to solve this equation for several textbook systems, including the harmonic oscillator, the finite potential well, and the hydrogen atom. Their results agree with the standard complex theory, apart from the sign ambiguity in the eigenvalues that will be discussed in Sec.~\ref{ho}. Within this class of examples, the single-field equation, therefore, reproduces the usual spectral information.

This success should not be overinterpreted, however. The derivation of Eqs.~\eqref{equ:chenfull}--\eqref{equ:rse-ti} assumes that the Hamiltonian appearing in the Schrödinger equation is a real symmetric operator. This is the case for the standard Hamiltonian with a real scalar potential, but it is not true for all physically important systems. The simplest warning comes from magnetic coupling, where the vector potential enters through $-i\hbar\nabla-q\mathbf A$. The following subsection shows this limitation and explains why the naive single-field equation has to be modified before magnetic fields can be included.
 
\subsection{A first warning: magnetic fields\label{ssec:mag}}
The derivation above contains an important restriction. It assumes that the Hamiltonian $\hat H$ is a real symmetric operator. This is true for the usual Schrödinger Hamiltonian
\begin{equation}
\hat H=-\frac{\hbar^2}{2m}\nabla^2+V(\mathbf r)
\label{eq:real_scalar_hamiltonian}
\end{equation}
with a real scalar potential, but it is no longer true in the presence of a magnetic vector potential. For a charged particle, one has
\begin{equation}
\hat H_A=\frac{1}{2m}\bigl(-i\hbar\nabla-q\mathbf A\bigr)^2+V ,
\label{eq:magnetic_hamiltonian_intro}
\end{equation}
which contains terms that are explicitly imaginary in the position representation. For this reason, the direct analog of Eq.~\eqref{equ:rse},
\begin{equation}
-\hbar^2\ddot u=\hat H_A^2u ,
\label{eq:naive_real_second_order_magnetic}
\end{equation}
does not define a real second-order equation for $u$ alone.

This does not mean that magnetic fields cannot be treated in real variables, but only that the naive Chen form is too restrictive. In the magnetic case, one has to decompose the Hamiltonian as
\begin{equation}
\hat H_A=\hat H_R+i\hat H_I
\label{eq:magnetic_H_split_intro}
\end{equation}
and keep track of the way in which $\hat H_I$ mixes the real and imaginary parts of the wave function. The resulting equations are still real, but their structure is different from Eq.~\eqref{equ:rse}. Since this issue is best understood after the Hamiltonian reconstruction has been developed, we postpone the detailed discussion to Sec.~\ref{sec:mag}.

\subsection{The appearance of additional energies\label{ho}}
Since Eq.~\eqref{equ:rse-ti} is an eigenvalue equation for $E^2$, the resulting energies have a sign ambiguity $\pm |E|$. If this were the actual prediction of additional energy levels, the whole approach would be compromised. However, this concern is unfounded. Instead, the additional solutions represent the {\em time-reversed} evolution of the {\em same} state.

To recognize this, we first notice that Eq.~\eqref{equ:rse} depends only on $\hat{H}^2$ and is therefore invariant under the   reflection $t \to -t$. This equation does not distinguish a time direction, and for each solution $u(t)$, the reflected function $u(-t)$ is again a solution without any further transformation. Consequently, the reduced single-field formulation contains two frequency branches. However, the physical state is specified not by \(u\) alone but by \((u,\dot u)\). Eq.~\eqref{equ:idurchr} restores the first-order Schrödinger dynamics and {\em selects} a definite orientation of motion in this two-dimensional real description. In particular, changing the sign of the frequency, which corresponds to \(t\mapsto -t\),
implies:
\begin{equation}
\dot u \mapsto -\dot u \quad \Rightarrow \quad v \mapsto -v.
\end{equation}
Hence, the two branches are related by
\begin{equation}
    (u,v) \;\mapsto\; (u,-v),
\end{equation}
which corresponds to the complex conjugation of the wave function, $\psi\;\mapsto\; \psi^*$. This is precisely the time-reversal operation in standard quantum mechanics. Thus, the solution space only appears larger because it also contains the time-reversed solutions, and the {\em dynamical} equivalence between the reduced single-field formulation and ordinary quantum mechanics is complete.
 
The curious connection of the reduced single-field formulation to the problem of time-reversal invariance lies at the heart of Callender's analysis \cite{callender2023}. He is dealing with the controversial issue of why time-reversal in quantum mechanics is not achieved simply by replacing $t\mapsto -t$, but also by complex conjugation. According to Callender, the real-valued formulation illuminates this relationship because the imaginary part is linked to the velocity of the real part $\dot{u}$. He concludes \cite[p. 848]{callender2023}: 
\begin{quote}
Seen through the prism of the real formalism, we can understand why complex conjugation is connected to temporal reflection. Formulated in the real theory version, a temporal reflection essentially reverses the velocities.  
\end{quote}
We will revisit and extend this analysis in Sec.~\ref{subsec:time_reversal_kramers_real}.
\section{The probability density\label{sec:prob}}
 
The equivalence of the spectra is not sufficient to establish the equivalence of the real and complex formulations. Quantum mechanics also requires a positive probability density, and in the standard formulation, this density has the simple local form
\begin{equation}
  \rho = \psi^*\psi .\label{equ:born}
\end{equation}
The question is therefore not merely whether the complex Schrödinger equation can be rewritten as a real second-order equation, but whether the probabilistic structure of the theory survives this rewriting in a local and transparent form.

This issue has both a historical and a technical side. Historically, it was already at the center of the Ehrenfest-Pauli exchange on the need for complex numbers in quantum mechanics. Technically, the same question reappears in Chen's formulation, where the missing imaginary part is reconstructed from the real part by a generally non-local operation. We shall first recall the historical argument and then show how non-locality enters the probability density and current. For notational simplicity, this section and Sec.~\ref{sec:momentum} are formulated in one spatial dimension ($x\in\mathbb R$, derivative $\partial_x$); nothing depends on this choice, and later sections indicate the dimension explicitly whenever an example is three-dimensional.

\subsection{The Ehrenfest-Pauli exchange \label{ssec:early}}

In 1932, Paul Ehrenfest posed the question of why the imaginary unit enters the Schrödinger equation \cite{ehrenfest32}. Given that {\em one} complex function corresponds to {\em two} real-valued functions, he specifically asked for the motivation behind the need for two real-valued scalars in quantum mechanics. 

Pauli's reply to this question focused on the probability density \cite{pauli33}. In the standard complex formulation, the Born density Eq.~\eqref{equ:born} is positive, conserved, and defined locally from the wave function at a fixed time. Pauli's point was that this simplicity is not accidental. If one rewrites the theory in terms of a single real field, the equation becomes second order in time, and the information carried by the second real component has to reappear as additional initial data.

To illustrate this argument, Pauli outlined a real-valued formulation of the theory. This formulation differs, however, from the Schrödinger-Chen proposal discussed in Sec.~\ref{sec:chen}. Pauli did not split the wave function into real- and imaginary part, but  defined a single real-valued field $U$ by:
\begin{equation}
    \psi=\left(\hat H+i\hbar \partial_t\right)U\quad\text{and}\quad\psi^* = \left(\hat H-i\hbar \partial_t\right)U.\label{equ:paulirec1}
\end{equation}
This definition implies: 
\begin{equation}
\Re\psi = \hat H U \quad\text{and}\quad\Im\psi=\hbar\dot{U}.\label{equ:paulisU}
\end{equation}
Hence, Pauli's $U$ should not be confused with Chen's $u=\Re \psi$. The Schrödinger equation implies, however, that the quantity $U$ satisfies the equation 
\begin{equation}
 -\hbar^{2}\ddot U = \hat{H}^2 U\, ,\label{equ:paulirsg}
\end{equation}
which is structurally identical to Chen's Eq.~\eqref{equ:rse}. \footnote{A similar point was raised by David Bohm in his textbook from 1951 
\cite[Chap.~4]{Bohm51}. Bohm's textbook is noteworthy because it is among the few treatments that explicitly address the apparent need for complex numbers in quantum mechanics. Although he did not refer to Schrödinger's earlier real-valued equation or to the Ehrenfest-Pauli exchange, his remarks are relevant here because they concern the same attempt to assess whether the complex wave function can be replaced by a real one. In the notation used here, however, his discussion does not clearly distinguish Pauli's auxiliary field \(U\) from the real part \(u=\Re\psi\). Since they both obey the same Eq.~\eqref{equ:paulirsg} they are easy to confuse. For this reason, Bohm's conclusion about the inadequacy of a real single-field formulation is not compelling.}

Using Eqs.~\eqref{equ:paulirec1}, the Born density can be expressed in terms of the Pauli variables $(U,\dot U)$:
\begin{equation}
\rho=\psi^*\psi = (\hat H U)^2+\hbar^2(\dot U)^2.\label{equ:paulisrho}
\end{equation}
This expression makes Pauli's argument transparent. Ehrenfest explicitly asked
 why quantum mechanics requires either one complex wave function
or two real functions. Pauli's reply addresses this framing directly: the
density \(\rho\) in \eqref{equ:paulisrho} depends on the two real
scalars \(U\) and \(\dot U\), just as \(\psi\) depends on the real and imaginary components. Pauli thus regarded the issue as settled: if the probability
density is to be defined from only one scalar, complex numbers are
unavoidable. This should not be read as a rejection of real variables, but as the correct observation that a single real configuration field is not, by itself, a complete local representation of the quantum state. The argument also appeared in Pauli's 1933 handbook article \cite[p.~40]{pauli1933handbuch}. 

The two parametrizations by Pauli and Chen are related to one another by
\begin{equation}
  u = \hat H U ,
  \qquad
  v = \hbar \dot U .
  \label{eq:pauli-chen-relation}
\end{equation}
Hence, whenever the inverse exists in the relevant subspace,
\begin{equation}
  U = \hat H^{-1}u .
  \label{eq:pauli-chen-inverse}
\end{equation}
Pauli's formulation therefore does not eliminate the inverse Hamiltonian
altogether. It hides it in the relation between \(U\) and the physical real
part \(u\), whereas Chen's formulation displays it explicitly in the
reconstruction of the missing imaginary part. Substituting \eqref{eq:pauli-chen-relation} into
\eqref{equ:paulisrho}, one obtains
\begin{equation}
  \rho
  = u^2
  + \left(\hbar \hat H^{-1}\dot u\right)^2 ,
  \label{eq:chen-density-from-pauli}
\end{equation}
which is precisely the Born density in Chen's variables. Thus, Pauli's and
Chen's formulations agree on the probabilistic content, but they display
the imaginary part differently. Pauli's variables keep the density
local in \(U\) and \(\dot U\), while Chen's variables use the physical real
part \(u=\Re\psi\) and thereby introduce the non-local operator
\(\hat H^{-1}\).

Thus, Chen's formulation reconstructs the imaginary part dynamically, and the price is the non-locality to which we now turn. 

\subsection{The non-locality of $H^{-1}$ and the loss of a manifestly local probability current\label{sec:alt}}

We now turn to the non-locality problem, which arises from the appearance of the inverse operator $\hat H^{-1}$ in Eq.~\eqref{equ:idurchr}.\footnote{ If \(\hat H\) has zero modes, \(\hat H^{-1}\) is not defined on the full Hilbert space. Throughout the reduced single-field formulation,
\(\hat H^{-1}\) should therefore be understood either on the subspace orthogonal to the zero modes or as an appropriate Moore-Penrose pseudoinverse \cite{benisrael2003generalized}. The zero mode sector must then be specified separately.}
 
Throughout this section, non-locality is meant in a representational or operator-theoretic sense, not in the sense of superluminal influence or Bell-type correlations. A quantity will be called local if its value at a point can be obtained from the fields and a finite number of their spatial derivatives at that point. By contrast, an expression involving $\hat H^{-1}$ is generally non-local, because the inverse of a differential operator is typically an integral operator. More explicitly, one has
\begin{equation}
(\hat H^{-1}\dot u)(x)=\int_{\mathbb R}G(x,x')\,\dot u(x')\,\dd x' ,
\end{equation}
where $G(x,x')$ is the corresponding Green function. The value of $(\hat H^{-1}\dot u)(x)$ therefore depends, in general, on the values of $\dot u$ away from $x$.

This statement at the operator level must be kept distinct from the different issue of how long the non-local tails are for a given class of states. Since $\dot u$ is not an arbitrary source but is constrained by Schrödinger dynamics, special solutions may display much milder spatial behavior than the Green function alone would suggest. This state-dependent localization property does not remove the operator-level non-locality, but it can reduce its practical severity.

Chen was aware of this issue. He argued that, although $\hat H^{-1}$ initially appears as a non-local Green-function operator, its action on the time derivative of a wave function effectively reduces to a local differential operation because of the properties of the Dirac delta function and its derivatives~\cite{chen91}. However, this argument is not fully convincing and contains a simple error. For the free particle in one dimension, Chen used $G(x,x')=|x-x'|^{-1}$, whereas the distributional solution of $-G''=\delta$ grows linearly, $G\propto|x-x'|$ up to affine terms; on $L^2(\mathbb R)$, the free Hamiltonian has no bounded inverse at all, so the statement necessarily concerns this distributional kernel. His subsequent Eq.~(14) in Ref.~\cite{chen91} therefore does not follow.

Nevertheless, Chen's intuition was not entirely misguided. As we will show in Sec.~\ref{sec:non-loc-rev}, free wave packets can exhibit a kind of containment of the otherwise generic non-locality: the reconstructed imaginary part may remain well localized even though the reconstruction map itself is non-local. This, however, is a state-dependent feature of particular solutions, not a local representation of the operator $\hat H^{-1}$.

 We now demonstrate that the probability given by Eq.~\eqref{equ:rhonaiv} is {\em globally} conserved by the real-valued Schrödinger dynamics. Consider, for this purpose, the time derivative of Eq.~\eqref{equ:rhonaiv}:
\begin{equation}
\dot{\rho} = 2 u\dot{u} + 2\hbar^2 (\hat{H}^{-1} \dot{u}) (\hat{H}^{-1} \ddot{u}).
\end{equation}
Using the real-valued Schrödinger equation to substitute $\ddot{u}$ with $-\frac{1}{\hbar^2} \hat{H}^2  u$ and 
given that $\hat{H}^{-1} \hat{H}^2 = \hat H$, one obtains:
\begin{equation}
\dot{\rho} = 2 u\dot{u} - 2 (\hat{H}^{-1} \dot{u}) (\hat H  u).
\end{equation}
For the {\em global} conservation of probability, one must consider the integral: 
\begin{equation}
\frac{d}{dt} \int_{\mathbb R} \rho\,  \dd x = 2 \int_{\mathbb R}  u \,\dot{u}\, \dd x - 2 \int_{\mathbb R} (\hat{H}^{-1} \dot{u}) (\hat H  u)\, \dd x.\label{equ:giltk}
\end{equation}
Since $H$ is self-adjoint, one can shift it and  exploit $\hat H\hat{H}^{-1}={\mathbb{1}}$. As a result, the two integrals cancel each other out, and the desired result $\frac{d}{dt} \int \rho \, \dd x =0$ follows.  

In standard quantum mechanics, $\dot{\rho}(x,t)$ can also be written as the divergence of a  {\em local} current $j(x,t)$, which implies the local continuity equation $\dot{\rho}+\partial_x j=0$. However, this local conservation of probability fails in the reduced variables $(u,\dot{u})$. Formally, one may write:
\begin{equation}
    j=\frac{\hbar^2}{m}\left [ u\,\partial_x( \hat{H}^{-1}\dot{u})-(\hat{H}^{-1}\dot{u})\partial_x u \right].\label{equ:pcurrent-chen}
\end{equation}
In the operator-theoretic sense specified above, this current is non-local: its value at $x$ depends on the reconstructed field $(\hat H^{-1}\dot u)(x)$, which is not determined by the values of $u$, $\dot u$, and finitely many of their spatial derivatives at $x$. Thus, the reduced variables $(u,\dot u)$ do not provide a manifestly
local probability current.

This kind of non-locality is not, by itself, an empirical inconsistency. Inverse differential operators also occur in familiar reformulations of local theories; for example, when a constraint is solved explicitly by a Green function. Such examples show that a non-local expression need not signal a new non-local physical interaction. Nevertheless, the non-locality is not innocuous in the present context. One motivation for the single-field formulation was precisely to provide a description in terms of one real field. If the imaginary part, and with it the probability density and current, can be recovered only by applying $\hat H^{-1}$ to $\dot u$, then the reduced variables $(u,\dot u)$ do not display the local structure that the standard complex formulation makes manifest. The non-locality is therefore not a contradiction; rather, it is a symptom that something has been hidden by the reduction. In Sec.~\ref{sec:non-loc-rev} this operator-level diagnosis will be sharpened into a state-level one (Proposition~\ref{prop:statelevel}): two admissible sets of reduced Cauchy data can coincide, together with all their spatial derivatives, in a neighbourhood of a point and nevertheless assign different Born densities to that point. A closely related problem will reappear in the discussion of momentum and the uncertainty relation.

\section{Momentum and uncertainty\label{sec:momentum}}
One strategy for avoiding the imaginary unit in quantum mechanics is the transition to a two-component formalism, which treats the components $u$ and $v$ of the wave function $\psi=u+iv$ as independent variables (see also our remarks in the introduction). While, for example, the momentum operator $\hat{p}_k=-i\hbar\,{\partial q_k}$ has complex eigenfunctions $\psi_{p_k}=e^{ip_kq_k/\hbar}$, the two-component version can be written as: 
\begin{equation}
\begin{aligned}
\hat p_k
&=
-\hbar J\,\partial_{q_k}
=
\begin{pmatrix}
0 & \hbar\partial_{q_k} \\
-\hbar\partial_{q_k} & 0
\end{pmatrix},
\\
&\text{with:}\quad J
=
\begin{pmatrix}
0 & -1 \\
1 & 0
\end{pmatrix}.
\end{aligned}
\label{eq:real_momentum_operator}
\end{equation}
Its eigenfunctions are real-valued,
\begin{equation}
      \psi^{(2)} =   \left ( \begin{array}{r}
\cos(p_kq_k/\hbar)  \\
\sin(p_kq_k/\hbar)  
\end{array}\right ) =\left ( \begin{array}{l}
 u  \\
 v  
\end{array}\right ),
\end{equation}
and the components can be directly identified with the real and imaginary parts of the complex solution. The symplectic matrix $J$ turns the vector by $\frac{\pi}{2}$; hence, it causes a sign change if applied twice: $J^2=-\mathbb{1}$. This obviously mimics the effect of the imaginary unit $i$. If one defines the matrix version of the position operator $\hat{x}=x\mathbb{1}$, even the canonical commutator relation $[\hat{x},\hat{p}]=\hbar J$ follows. 

Such a substitute for the momentum operator is not available in the single-component formulation.\footnote{Neither Chen \cite{chen89,chen91} nor Makris and Dargush \cite{makris2025} deal with this issue.}

\subsection{Momentum in the reduced single-field formulation}
A possible workaround for the reduced single-field formulation is to generalize the expression that relates the expectation value of the momentum to the probability current in the standard formulation:
\begin{eqnarray}
\langle \hat{p} \rangle_{QM} &=&  \int_{\mathbb R} \psi^* (-i\hbar\,\partial_x)\psi\,\dd x\label{equ:pqm}\\
&=&m\int_{\mathbb R} j_{QM}\,\dd x, 
\end{eqnarray}
with:
\begin{equation}
j_{QM} = \frac{\hbar}{2 i m} ( \psi^* \partial_x\psi-\psi\,\partial_x\psi^*).
\end{equation}
If we instead use Eq.~\eqref{equ:pcurrent-chen}, we obtain an expression for the expectation value in the reduced single-field formulation (note that we write $\langle p\rangle$ instead of $\langle \hat{p}\rangle$):
\begin{equation}
    \langle p \rangle = \hbar^2\int_{\mathbb R}  \left[ u \,\partial_x (\hat{H}^{-1}\dot{u}) - (\hat{H}^{-1}\dot{u}) \partial_x u \right] \dd x.\label{equ:momentum}
\end{equation}
As in the discussion of probability current, the appearance of
$\hat H^{-1}$ means that Eq.~\eqref{equ:momentum} is non-local in the reduced variables
$(u,\dot u)$. This non-locality does not prevent one from assigning the
standard expectation value of momentum: Eq.~\eqref{equ:momentum} is just the ordinary
complex expression rewritten after eliminating the imaginary part.
However, it does show that the reduced variables do not provide a local
momentum density. Momentum is obtained only after the imaginary part
has been reconstructed by applying $\hat H^{-1}$ to $\dot u$.

There is also a more technical subtlety. Eq.~\eqref{equ:momentum} can be simplified by
 integration by parts, but the validity of this step depends on the
asymptotic behavior of $H^{-1}\dot u$. In the independent
variables $u$ and $v$, the relevant boundary condition is imposed
directly on both fields. In the reduced variables, however, the boundary
term contains the non-local expression
\begin{equation}
\Big[ u(x) \cdot (\hat{H}^{-1} \dot u)(x) \Big]_{-\infty}^{+\infty},
\end{equation}
and the expression $(\hat{H}^{-1} \dot u)(x)$ may be non-zero in the limit $|x|\to\infty$. Thus the usual simplification is justified only when this term vanishes. For sufficiently localized physical states this will normally be the case, but it is not a purely local statement about the Cauchy data $(u,\dot u)$. Keeping this caveat in mind, one obtains a compact expression that bears some resemblance to Eq.~\eqref{equ:pqm}:
\begin{equation}
   \langle p \rangle = -2\hbar^2\int_{\mathbb R}   \hat{H}^{-1}\dot{u} \,\partial_x u \, \dd x.\label{equ:pcompact}
\end{equation}
The central point is therefore not that momentum has disappeared, but that it is no longer represented in the reduced single-field formulation by an operator acting on $u$ alone. Its expectation value can still be computed, but only by combining the spatial dependence of $u$ with the time dependence needed to reconstruct the imaginary part. 

This marks a sharp departure from the way momentum is represented in the standard formulation. In ordinary quantum mechanics, the state $\psi\sim e^{ikx}$ already has a definite momentum $\hbar k$. The real-valued snapshot $u\sim \cos(kx)$, by contrast, does not determine a direction of propagation. From $u(x,t)$ at a single time, one cannot tell whether the wave is moving to the right, moving to the left, or forming part of a standing wave. In the reduced single-field formulation, momentum can therefore be assigned only when the second Cauchy datum is also known: one needs both the spatial configuration $u(x,t)$ and its time derivative $\dot u(x,t)$.

The simplest example would be the plane wave $u(x,t)\sim\cos(kx-\omega t)$. Applying Eq.~\eqref{equ:pcompact} yields $\langle p\rangle =\hbar k$. One would expect that for the plane wave, not only is the {\em expectation value} $\hbar k$ but also that the momentum has the {\em definite value} $\hbar k$. In the standard formulation, this is implied by the plane wave being an eigenstate.  In the reduced real formulation, one may calculate the standard deviation:
\begin{equation}
    \Delta p = \sqrt{\langle p^2 \rangle - \langle p \rangle^2}.
\end{equation}
$\langle p\rangle$ is already given by Eq.~\eqref{equ:pcompact}, and for the expectation value of $p^2$, we obtain:
\begin{equation}
\begin{aligned}
\langle p^2\rangle
&=
\hbar^2
\int_{\mathbb R}
\bigl[
(\partial_x u)^2+(\partial_x v)^2
\bigr]\,\dd x
\\
&=
\hbar^2
\int_{\mathbb R}
(\partial_x u)^2\,\dd x
+
\hbar^4
\int_{\mathbb R}
\bigl[
\partial_x(\hat H^{-1}\dot u)
\bigr]^2\,\dd x .
\end{aligned}
\label{equ:momev}
\end{equation}
For the plane wave, we get $\langle p^2\rangle =\hbar^2k^2$; hence $\Delta p=0$. In this sense, the reduced formulation  also treats plane waves as states with definite momentum, although the notion of an eigenstate is no longer available.  

A definite momentum while the location is completely undetermined is reminiscent of Heisenberg's uncertainty principle. In the next section, we will examine what happens to this cornerstone of quantum mechanics in the reduced single-field formulation. 

\subsection{The uncertainty principle in the reduced variables\label{ssec:hur}}

The difficulty with the uncertainty principle is not that the inequality
becomes false. Since the reduced real equation is dynamically equivalent to
the complex Schrödinger equation, all physical states still satisfy the
standard bound. The problem is rather that the usual derivations no longer
have a local expression in the variables \((u,\dot u)\).

This is already suggested by the preceding discussion of momentum. In a
two-component real formulation, the missing role of the imaginary unit is
played by a real operator \(J\) with \(J^2=-\mathbb{1}\). Momentum can then be written as a generator involving \(J\). In the reduced one-field formulation,
however, there is no local operation on \(u\) alone that plays this role.
The imaginary part has been eliminated and can only be recovered as
\begin{equation}
  v = \hbar \hat H^{-1}\dot u .
  \label{eq:v-reconstruction-uncertainty}
\end{equation}
Thus, the complex structure needed for standard kinematics is still
present, but only in a non-local form.

The point can be seen most clearly from the standard Cauchy-Schwarz proof.
For a normalized state \(\psi=u+iv\), with
\(\langle x\rangle=\langle p\rangle=0\), one introduces
\begin{equation}
  f=x\psi,
  \qquad
  g=\partial_x\psi .
  \label{eq:f-g-uncertainty}
\end{equation}
Then
\begin{equation}
  (\Delta x)^2=\|f\|^2,
  \qquad
  (\Delta p)^2=\hbar^2\|g\|^2 ,
  \label{eq:variances-f-g}
\end{equation}
and the Cauchy-Schwarz inequality gives
\begin{equation}
  (\Delta x)^2\,(\Delta p)^2
  \geq \hbar^2 |\langle f,g\rangle|^2 .
  \label{eq:cs-uncertainty}
\end{equation}
The lower bound follows from the real part of
\begin{equation}
  \langle f,g\rangle  =  \int x\,\psi^*\partial_x\psi\,dx .
  \label{eq:fg-inner-product}
\end{equation}
Indeed,
\begin{align}
  \Re\langle f,g\rangle
  &=  \int x\left(u\partial_x u+v\partial_x v\right)dx\nonumber \\
  &=  \frac{1}{2}\int x\,\partial_x(u^2+v^2)\,dx \nonumber \\
  &= -\frac{1}{2},
  \label{eq:real-part-fg}
\end{align}
where the last step uses normalization and the vanishing of the boundary
term. Hence, Eq.~\eqref{eq:cs-uncertainty} implies
\begin{equation}
  \Delta x\,\Delta p \geq \frac{\hbar}{2}.
  \label{eq:heisenberg-bound}
\end{equation}

In the reduced variables, this key local step is no longer available. After
eliminating \(v\), one has
\begin{equation}
  v(x,t)=\hbar(\hat H^{-1}\dot u)(x,t),
  \label{eq:v-nonlocal-uncertainty}
\end{equation}
so that
\begin{equation}
  v\,\partial_x v
  =
  \hbar^2
  (\hat H^{-1}\dot u)\,
  \partial_x(\hat H^{-1}\dot u).
  \label{eq:v-dxv-nonlocal}
\end{equation}
Formally, this can still be written as
\begin{equation}
  v\,\partial_x v
  =
  \frac{1}{2}\partial_x
  \left[
    \hbar^2(\hat H^{-1}\dot u)^2
  \right].
  \label{eq:v-dxv-formal-total-derivative}
\end{equation}
But the quantity inside the derivative is not a local functional of \(u(x,t)\) and its spatial derivatives. The integration-by-parts step has, therefore, lost the local meaning it had in the two-component description. The obstruction is the same as for the probability current and momentum: the reduced variables contain the missing imaginary part only through the non-local operator \(\hat H^{-1}\).

Thus, the uncertainty relation is not threatened, but its standard
kinematical origin is obscured. The failure is therefore not a physical failure of the uncertainty principle, but a diagnostic failure of the variables \((u,\dot u)\) as local state variables.

This is compatible with de Gosson's symplectic reading of the uncertainty principle \cite{degosson2009camel}: the relevant structure is not merely the Fourier relation between \(x\) and \(p\), but the underlying symplectic geometry. The reduced single-field formulation gives a concrete illustration of this point. Once the imaginary part is eliminated, the uncertainty relation remains valid, but the symplectic structure from which its standard kinematical derivation draws its force is no longer locally visible.

\section{From the reduced equation to real phase space}
\label{sec:real-phase-space-reconstruction}

The preceding sections have shown that the real second-order equation is not merely a technical rewriting of the Schrödinger equation. It forces a sharper question about the role of the imaginary unit. If the symbol \(i\) is removed from the notation, which structures must still be retained in order to recover probability, momentum, and the uncertainty relation in their usual form? The answer developed in this section is that the reduced equation must be lifted to a real Hamiltonian phase space. This lift remains entirely real-valued, but it restores the canonical partner and the symplectic structure hidden by the one-field reduction.

\subsection{The limits of the reduced description}
 
The central difficulty of the reduced single-field formulation can now be stated in one sentence: it is dynamically sufficient but not locally complete. Probability, momentum, and the uncertainty relation do not disappear. Global probability is conserved, momentum expectation values can be computed, and the usual uncertainty bound remains valid. Yet in the variables \((u,\dot u)\) these structures are no longer available as autonomous local structures.

The reason is always the same. The imaginary part has been eliminated and can be recovered only through the Hamiltonian-dependent map
\begin{equation}
v=\hbar\hat H^{-1}\dot u .
\end{equation}

Thus, the reduced variables hide the local structure that the standard complex formalism keeps explicit. The question is, therefore, not merely whether the symbol \(i\) can be removed from the Schrödinger equation, but what structure has to remain in order for the usual quantum kinematics to be recovered.

The answer developed below is that the second-order real equation must be understood as the configuration-space projection of a  Hamiltonian system. In the case of a time-independent potential with a real Hamiltonian, the canonical partner of \(u\) is
\begin{equation}
\pi=\hbar^2\hat H^{-1}\dot u .    
\end{equation}
The non-locality encountered above is then seen as the price of eliminating this canonical partner. This phase-space reconstruction remains entirely real-valued; what fails is only the stronger claim that one real configuration field and its ordinary time derivative support a Born rule that is simultaneously local and independent of the Hamiltonian; the precise dichotomy, established below for the Chen and Pauli parametrizations, is stated in Sec.~\ref{ssec:pauliagain}.

\subsection{The Hamiltonian lift}
 \label{sec:realH}

We start from the reduced real equation
\begin{equation}
\ddot u+\hbar^{-2}\hat H^2u=0 .
\label{eq:red}
\end{equation}
This equation fixes the trajectories of \(u\), but it does not by itself fix the symplectic structure. This is already familiar from the ordinary harmonic oscillator: the equation
\(\ddot q+\omega^2q=0\) determines the motion, but not whether the canonical momentum is \(p=\dot q\), \(p=m\dot q\), or, more generally, \(p=\alpha\dot q\). To reconstruct the quantum theory in real variables, one must therefore ask which Hamiltonian lift of the second-order equation carries the correct conserved Born norm and generator structure.

\subsection*{Step 1: Matching the second-order equation}
Assume that the dynamics of \(u\) arise from first-order Hamiltonian equations of the form
\begin{equation}
\dot u = \hat L\pi,
\qquad
\dot \pi = -\hat Ku,
\label{eq:ham1}
\end{equation}
where $\hat L$ and $\hat K$ are real symmetric (self-adjoint) operators to be determined.\footnote{Two remarks on the scope of this ansatz. First, Eq.~\eqref{eq:ham1} excludes ``gyroscopic'' diagonal blocks of the form $\dot u=\hat Au+\hat L\pi$, $\dot\pi=-\hat Ku+\hat B\pi$. For a real, time-reversal-invariant $\hat H$, this exclusion is justified by requiring the lift to be covariant under the time-reversal map $(u,\pi)\mapsto(u,-\pi)$ of Sec.~\ref{subsec:time_reversal_kramers_real}, which forbids the diagonal blocks. When time-reversal is broken, such blocks do occur and are physically required -- see the magnetic case of Sec.~\ref{sec:mag}. Second, self-adjointness of $\hat K$ and $\hat L$ is used in Step 2; without it, the norm condition yields $\hat K=\hbar^2\hat L^{T}$ instead of Eq.~\eqref{eq:Aeq}, and Eq.~\eqref{eq:KsqHsq} admits further solutions of the form $\hat K=W|\hat H|$ with $W$ orthogonal.} Eliminating $\pi$ gives
\begin{equation}
\ddot u = -\hat L\hat K\,u.
\end{equation}
To reproduce the reduced equation \eqref{eq:red}, one therefore needs:
\begin{equation}
\hat L\hat K=\hbar^{-2}\hat H^2.
\label{eq:AB}
\end{equation}
This condition alone does not yet uniquely determine the momentum. We need an additional criterion in order to obtain the desired result $\hat K=\hat H$ and $\hat L=\hbar^{-2}\hat H$.

\subsection*{Step 2: Matching the conserved norm}

The second requirement is that the lifted theory should carry the natural rotationally invariant norm on the local phase-space plane. Thus, one demands a conserved quantity of the form
\begin{equation}
\mathcal P[u,\pi] =\int_{\mathbb R^3} \dd^3r\left(u^2+\frac{\pi^2}{\hbar^2}\right).
\label{eq:P}
\end{equation}
Differentiating with respect to time and inserting \eqref{eq:ham1} yields
\begin{align}
\dot{\mathcal P}&=2\langle u,\dot u\rangle+\frac{2}{\hbar^2}\langle \pi,\dot\pi\rangle\nonumber\\
&=2\langle u,\hat L\pi\rangle-\frac{2}{\hbar^2}\langle \pi,\hat Ku\rangle .
\end{align}
For $\dot{\mathcal P}$ to vanish for arbitrary states, one must have
\begin{equation}
\hat L=\hbar^{-2}\hat K.
\label{eq:Aeq}
\end{equation}
Combining \eqref{eq:AB} and \eqref{eq:Aeq} gives
\begin{equation}
\hat K^2=\hat H^2.
\label{eq:KsqHsq}
\end{equation}
We choose the physically natural branch
\begin{equation}
\hat K=\hat H,
\qquad
\hat L=\hbar^{-2}\hat H.
\end{equation}
The alternative choice $\hat K=-\hat H$, $\hat L=-\hbar^{-2}\hat H$ is absorbed by the trivial redefinition $\pi\to -\pi$.

A qualification is in order, however, because Eq.~\eqref{eq:KsqHsq} constrains only the square of $\hat K$. The redefinition $\pi\to-\pi$ absorbs the \emph{global} sign flip, but a self-adjoint solution may still choose the sign independently on each spectral subspace of $\hat H$: every $\hat K=\varepsilon(\hat H)\,\hat H$ with $\varepsilon(\hat H)^2=\mathbb 1$ satisfies Eqs.~\eqref{eq:AB} and \eqref{eq:Aeq}. In particular, $\hat K=|\hat H|$ conserves the norm \eqref{eq:P} and reproduces the reduced equation \eqref{eq:red}, yet it lifts to $i\hbar\dot\psi=|\hat H|\psi$ in the sense of Sec.~\ref{ssec:equivalence}, which is empirically inequivalent to the Schrödinger evolution whenever the spectrum of $\hat H$ contains both signs. For a superposition of a bound state ($E_1<0$) and a scattering state ($E_2>0$), the local density beats at $(|E_1|+E_2)/\hbar$ under the Schrödinger dynamics but at $\bigl||E_1|-E_2\bigr|/\hbar$ under the $|\hat H|$ lift -- a measurable difference built on identical reduced trajectory data, since both lifts reproduce the same second-order equation for $u$. For sign-definite $\hat H$, the lift is therefore unique up to $\pi\to-\pi$; for indefinite spectra, the norm criterion fixes it only up to these subspace-wise signs, and the physical branch is selected by requiring the lifted first-order system to reproduce the original spectral sign of $\hat H$, and not merely $\hat H^2$ -- equivalently, by matching the first-order dynamics \eqref{equ:coup12} (cf.\ Secs.~\ref{ho} and \ref{subsec:real_zeeman_splitting}, and Appendix~\ref{app:hinv}). With this understanding, the canonical equations are fixed to be
\begin{subequations}
\label{eq:canonical_equations}
\begin{align}
\dot u &= \hbar^{-2}\hat H\,\pi,\label{eq:u_dot_pi}\\
\dot \pi &= -\hat H\,u.\label{eq:pi_dot_u}
\end{align}
\end{subequations}
Solving the first equation for $\pi$ immediately gives
\begin{equation}
\pi=\hbar^2\hat H^{-1}\dot u.
\label{eq:pi_canonical}
\end{equation}
This is the desired conjugate momentum. If one introduces the dimensionless companion field $v := \hbar^{-1}\pi$, then Eqs.~\eqref{eq:canonical_equations} take the form of the coupled equations~\eqref{equ:coup12}, which were the starting point of Chen's discussion. This familiar pair of first-order equations is, therefore, not an external postulate. It is the Hamiltonian completion selected by Eq.~\eqref{eq:red}, together with the requirement that the norm becomes an exactly conserved quadratic form. We note that the lift extends verbatim to time-dependent potentials: with $\hat H(t)$ in Eqs.~\eqref{eq:canonical_equations}, eliminating $\pi$ reproduces Chen's Eq.~\eqref{equ:chenfull}, including the $\dot V$ term. This reinforces the reading that the first-order canonical system, not the reduced second-order equation, is the fundamental object.

\subsection{Mode-wise interpretation and quadratures}
 \label{ssec:heuristic}
The key probability assumption~\eqref{eq:P} and the Eq.~\eqref{eq:pi_canonical} for the conjugate momentum allow for an  intuitive motivation when the reduced real equation is investigated mode-by-mode. For each energy mode, Eq.~\eqref{eq:red} implies:
\begin{equation}
\ddot u_k+\omega_k^2 u_k=0,
\end{equation}
which is simply the equation of a harmonic oscillator. Its general solution may be written as
\begin{equation}
u_k(t)=A_k\cos(\omega_k t+\alpha_k),
\end{equation}
so that
\begin{equation}
\dot u_k(t)=-\,\omega_k A_k\sin(\omega_k t+\alpha_k).
\end{equation}
Dividing by \(\omega_k\) gives
\begin{equation}
\frac{\dot u_k(t)}{\omega_k}=-\,A_k\sin(\omega_k t+\alpha_k).
\end{equation}
This is the point at which the language of quadratures becomes useful. For each normal mode of the second-order equation, \(u_k\) supplies one oscillating component, while \(\dot u_k/\omega_k\) supplies the component shifted by \(\pi/2\). We will call this phase-shifted component the second quadrature of the mode.\footnote{In the following, we will use {\em second quadrature}, {\em companion quadrature}, and {\em canonical partner} interchangeably, depending on whether the mode picture, the reconstruction, or the phase-space role is in focus.} The key point is that
\begin{equation}
u_k^2+
\left(
\frac{\dot u_k}{\omega_k}
\right)^2
=
A_k^2
\end{equation}
is time independent, whereas \(u_k^2\) alone oscillates in time. Since in an energy eigenbasis
\begin{equation}
\hat H\phi_k=E_k\phi_k,
\qquad
\omega_k=\frac{E_k}{\hbar},
\end{equation}
one may equally write
\begin{equation}
\frac{\dot u_k}{\omega_k}=\frac{\hbar}{E_k}\dot u_k,
\end{equation}

which is just the \(k\)-th mode component of the operator relation $v=\hbar \hat H^{-1}\dot u$. Thus $v$ is the field-theoretic version of the second quadrature. Equivalently, the canonical momentum $\pi=\hbar v$ is the corresponding phase-space partner of $u$. (In writing $\omega_k=E_k/\hbar$ we have taken $E_k>0$; for spectra containing both signs, the oscillator frequency of each mode is $|E_k|/\hbar$, and the sign of $E_k$ is carried by the branch selection of Sec.~\ref{sec:realH} and Appendix~\ref{app:hinv}.)

The central point is that \(\pi\neq \dot u\). The reduced variable \(u\) is not the coordinate of a mechanical oscillator with unit mass. Instead, the effective \enquote{mass matrix} of the reduced theory is \(\hat H^{-1}\). This is why the canonical momentum \(\pi\) contains the inverse Hamiltonian.\footnote{ As noted in Sec.~\ref{sec:alt}, possible zero modes of \(\hat H\) must be excluded or treated separately before using \(\hat H^{-1}\); equivalently, one may work with a Moore-Penrose pseudoinverse \cite{benisrael2003generalized}. From the Hamiltonian point of view, this caveat prevents a spurious solution \(u_0(t)=a+bt\) for \(\hat H u_0=0\), which would otherwise appear in the reduced second-order equation.} In this way, the reduced second-order equation is seen to be the configuration-space projection of a first-order Hamiltonian system with the correct Born norm built in from the start.

The mode-wise discussion gives an intuitive interpretation of \(\pi/\hbar\) as the second quadrature associated with \(u\). The next step is to state the corresponding equivalence with the standard complex Schrödinger equation explicitly.

\subsection{Equivalence with the complex Schrödinger equation\label{ssec:equivalence}}
The equivalence established in this subsection is a special case of the standard geometric formulation of quantum mechanics. In that tradition, Schrödinger evolution is Hamiltonian flow on a real symplectic or Kähler state space, and the Poisson bracket of expectation-value functions reproduces the commutator bracket \cite{AshtekarSchilling1999,Strocchi1966}. The purpose of the present argument is therefore not to rederive that equivalence in general. It is to identify the particular canonical partner selected by the Schrödinger-Chen reduction. The missing component is not postulated as an independent second real field; it is fixed by the Hamiltonian lift of the
reduced equation. 

Let \(\mathcal K\) be the real Hilbert space of real-valued wave functions and let \(\hat H\) be a real self-adjoint Hamiltonian on a
common invariant domain. Let \(\mathcal K_{\mathbb C}\) denote its
complexification, and let \(\hat H_{\mathbb C}\) be the complex-linear
extension of \(\hat H\).

\begin{proposition}
The real Hamiltonian system
\begin{subequations}
\label{eq:real-hamilton-system-u-pi}
\begin{align}
\dot u &= \hbar^{-2}\hat H\,\pi,\label{eq:u_dot_pi2}\\
\dot \pi &= -\hat H\,u.\label{eq:pi_dot_u2}
\end{align}
\end{subequations}
is equivalent to the complex Schrödinger equation
\begin{equation}
  i\hbar\dot\psi=\hat H_{\mathbb C}\psi
  \label{eq:complex-schroedinger-equivalence}
\end{equation}
under the real-linear identification
\begin{equation}
  \psi = u+i\frac{\pi}{\hbar}.
  \label{eq:complex-coordinate-from-phase-space}
\end{equation}
Moreover, the conserved phase-space norm is precisely the usual complex Hilbert-space norm:
\begin{equation}
  \|\psi\|^2_{\mathcal K_{\mathbb C}}  =  \|u\|^2_{\mathcal K}  +  \frac{1}{\hbar^2}\|\pi\|^2_{\mathcal K}.
  \label{eq:norm-equivalence}
\end{equation}
\end{proposition}

\begin{proof}
The map
\begin{equation}
  T:\mathcal K\oplus\mathcal K\to \mathcal K_{\mathbb C},   \qquad
  T(u,\pi)=u+i\frac{\pi}{\hbar}
\end{equation}
is a real-linear bijection, with its inverse given by
\begin{equation}
  u=\Re\psi,  \qquad \pi=\hbar\,\Im\psi .
\end{equation}
Assume first that \((u,\pi)\) satisfy the Eqs.~\eqref{eq:real-hamilton-system-u-pi}.
Then
\begin{align}
  i\hbar\dot\psi
  &=  i\hbar  \left(   \dot u+i\frac{\dot\pi}{\hbar}  \right)  \nonumber \\
  &=  i\hbar\dot u-\dot\pi  \nonumber \\
  &=  i\hbar  \left(   \hbar^{-2}\hat H\pi  \right)  +\hat H u  \nonumber \\
  &=  \hat H u+i\hat H\frac{\pi}{\hbar}  \nonumber \\
  &=  \hat H_{\mathbb C}  \left(    u+i\frac{\pi}{\hbar}  \right)  =  \hat H_{\mathbb C}\psi .
\end{align}
Thus \(\psi\) satisfies the complex Schrödinger equation.

Conversely, suppose that \(\psi\) satisfies
\eqref{eq:complex-schroedinger-equivalence}, and write
\begin{equation}
  \psi=u+iv,
  \qquad
  \pi=\hbar v .
\end{equation}
Then
\begin{equation}
  i\hbar(\dot u+i\dot v)   =  \hat H u+i\hat H v .
\end{equation}
Equating real and imaginary parts gives
\begin{subequations}
 \begin{align}
  -\hbar\dot v &= \hat H u ,\\
  \hbar\dot u &= \hat H v .
\end{align}
\end{subequations}
Using \(v=\pi/\hbar\), these two equations become precisely Eqs.~\eqref{eq:real-hamilton-system-u-pi}.
Finally,
\begin{equation}
  \|\psi\|^2  =  \|u+iv\|^2  =  \|u\|^2+\|v\|^2  =  \|u\|^2+\frac{1}{\hbar^2}\|\pi\|^2 ,
\end{equation}
which proves \eqref{eq:norm-equivalence}.
\end{proof}

 \begin{corollary}
If \(\hat H^{-1}\) exists on the relevant subspace, the real Hamiltonian
system above projects to the reduced second-order equation
\begin{equation}
  \ddot u+\hbar^{-2}\hat H^2u=0 ,
  \label{eq:reduced-second-order-corollary}
\end{equation}
with
\begin{equation}
  \pi=\hbar^2\hat H^{-1}\dot u .
  \label{eq:pi-from-reduced-corollary}
\end{equation}
\end{corollary}

\begin{proof}
Eliminating \(\pi\) from the set \eqref{eq:real-hamilton-system-u-pi} gives
\begin{equation}
  \ddot u
  =
  \hbar^{-2}\hat H\dot\pi
  =
  -\hbar^{-2}\hat H^2u .
\end{equation}
Conversely, if \(u\) satisfies
\eqref{eq:reduced-second-order-corollary} and \(\hat H^{-1}\) is defined on
\(\dot u\), then defining \(\pi\) by
\eqref{eq:pi-from-reduced-corollary} gives
\begin{equation}
  \dot u=\hbar^{-2}\hat H\pi
\end{equation}
and
\begin{equation}
  \dot\pi
  =
  \hbar^2\hat H^{-1}\ddot u
  =
  -\hat H u ,
\end{equation}
so the lifted Hamiltonian equations are recovered.
\end{proof}

\subsection{Geometric interpretation of the complex structure}
The preceding proposition establishes the equivalence between the lifted real Hamiltonian system and the standard complex Schrödinger equation. The geometric meaning of this equivalence is simple. Once the canonical pair \((u,\pi)\) has been restored, the dimensionless pair \((u,\pi/\hbar)\) forms a real two-dimensional phase plane. On this plane, one may introduce a compatible complex structure \(J\), i.e., a real linear
map satisfying
\begin{equation}
  J^2=-\mathbb 1 .
\end{equation}
With the sign convention corresponding to
\begin{equation}
  \psi=u+i\frac{\pi}{\hbar},
\end{equation}
the action of \(J\) is a rotation by \(\pi/2\). In the associated complex
coordinate, this same operation is represented by multiplication with the
imaginary unit. Equivalently, complex multiplication is reconstructed from the real structure by
\begin{equation}
  (a+ib)X = aX+bJX .
\end{equation}
Figure~\ref{fig:phase_plane} summarizes this local geometry: free evolution rotates each mode's pair $(u,\pi/\hbar)$ clockwise at its frequency (Sec.~\ref{ssec:heuristic}), time reversal reflects the pair in the $u$-axis (Sec.~\ref{subsec:time_reversal_kramers_real}), and a gauge transformation rotates it by a position-dependent angle (Sec.~\ref{sec:mag}).

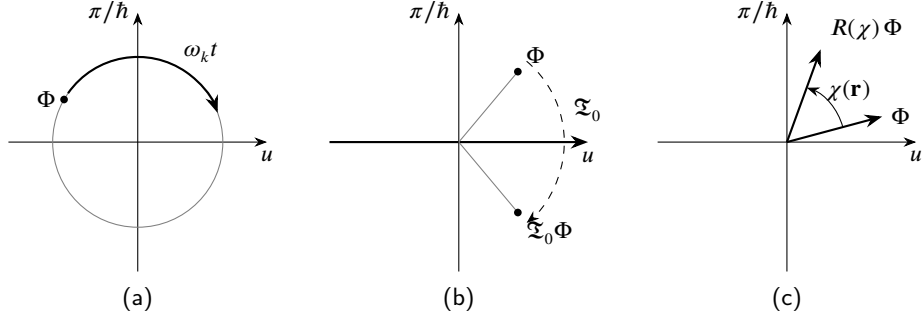
\begin{figure*}[t]
\centering
\begin{tikzpicture}[>=Stealth,scale=0.9]
  \draw[->] (-1.9,0) -- (1.9,0) node[below] {$u$};
  \draw[->] (0,-1.9) -- (0,1.9) node[left] {$\pi/\hbar$};
  \draw[gray] (0,0) circle (1.25);
  \fill (150:1.25) circle (1.6pt) node[left] {$\Phi$};
  \draw[->,thick] (145:1.25) arc (145:20:1.25);
  \node at (55:1.62) {$\omega_k t$};
  \node at (0,-2.3) {(a)};
\end{tikzpicture}\hspace{1.6em}
\begin{tikzpicture}[>=Stealth,scale=0.9]
  \draw[->,line width=0.8pt] (-1.9,0) -- (1.9,0) node[below] {$u$};
  \draw[->] (0,-1.9) -- (0,1.9) node[left] {$\pi/\hbar$};
  \draw[gray] (0,0) -- (50:1.35);
  \draw[gray] (0,0) -- (-50:1.35);
  \fill (50:1.35) circle (1.6pt) node[above right] {$\Phi$};
  \fill (-50:1.35) circle (1.6pt) node[below right] {$\mathfrak T_0\Phi$};
  \draw[->,dashed] (50:1.55) arc (50:-50:1.55);
  \node at (14:1.95) {$\mathfrak T_0$};
  \node at (0,-2.3) {(b)};
\end{tikzpicture}\hspace{1.6em}
\begin{tikzpicture}[>=Stealth,scale=0.9]
  \draw[->] (-1.9,0) -- (1.9,0) node[below] {$u$};
  \draw[->] (0,-1.9) -- (0,1.9) node[left] {$\pi/\hbar$};
  \draw[->,thick] (0,0) -- (15:1.45) node[right] {$\Phi$};
  \draw[->,thick] (0,0) -- (70:1.45) node[above right] {$R(\chi)\,\Phi$};
  \draw[->] (15:0.85) arc (15:70:0.85);
  \node at (40:1.18) {$\chi(\mathbf r)$};
  \node at (0,-2.3) {(c)};
\end{tikzpicture}
\caption{The local real phase plane $(u,\pi/\hbar)$ at a fixed point. (a)~Free evolution rotates every energy mode clockwise at its frequency $\omega_k=E_k/\hbar$, drawn for $E_k>0$ (Sec.~\ref{ssec:heuristic}); the Born density is the squared radius $u^2+\pi^2/\hbar^2$, the conserved $A_k^2$ of each mode. (b)~Spinless time reversal $\mathfrak T_0:(u,\pi)\mapsto(u,-\pi)$ is the reflection in the $u$-axis -- the real form of complex conjugation (Sec.~\ref{subsec:time_reversal_kramers_real}). (c)~A gauge transformation rotates the pair by the position-dependent angle $\chi(\mathbf r)$, Eq.~\eqref{eq:magnetic_rigid_SO2}, while the vector potential shifts as $A_i\mapsto A_i+(\hbar/q)\,\partial_i\chi$, Eq.~\eqref{eq:gauge_transform_real}.}
\label{fig:phase_plane}
\end{figure*}

Thus, after the phase-space pair has been identified, the transition to a complex wave function is mathematically straightforward and may even seem trivial. This, however, is not the substantive step in the present argument. We do not postulate a two-component real wave function from the start, as this would merely rewrite \(\psi=u+iv\). The starting point is the reduced second-order equation for the single real field \(u\). The second component is then fixed by the Hamiltonian reconstruction as the canonical momentum
\begin{equation}
  \pi=\hbar^2\hat H^{-1}\dot u ,
\end{equation}
or, equivalently, \(v=\pi/\hbar\). The reconstruction therefore does not
double variables by hand; it identifies which real pair carries the local
quantum kinematics.\footnote{We thank Martin Schwingenheuer for emphasising
this distinction.}

The examples discussed in the following sections should be read in this
light. They are not intended as independent proofs of the equivalence just
established. Their purpose is diagnostic: they show how familiar quantum structures look in the reduced single-field description and in the real phase-space formulation, and they identify which parts of the usual complex formalism are carried by the second Cauchy datum, the canonical momentum, or the internal \(SO(2)\) geometry.

\subsection{Relation to Pauli's real-valued equation\label{ssec:pauliagain}}

In fact, the real phase-space formulation proposed here can be expressed equally well in Pauli's variables, which were introduced in Sec.~\ref{ssec:early}. Instead of taking the configuration-type variable $u=\Re\psi$, Pauli introduced a real auxiliary field $U$ such that $u=\Re\psi=HU$ and $v=\Im\psi=\hbar \dot U$. In these variables, the natural conjugate momentum is simply
\begin{equation}
\Pi_U=\hbar^2 \dot U ,
\end{equation}
and the second-order equation $-\hbar^2 \ddot U = H^2 U$ can be obtained from the local Hamiltonian
\begin{equation}
\mathcal H = \frac{1}{2\hbar^2}\Pi_U^2+\frac12 (HU)^2 .
\end{equation}
The conserved Born density then takes the form
\begin{equation}
\rho=(HU)^2+\frac{\Pi_U^2}{\hbar^2}.
\end{equation}
This is the Pauli-variable version of the density associated with the
phase-space norm \(\mathcal P[u,\pi]\) introduced in Eq.~\eqref{eq:P}. In the canonical
variables \((u,\pi)\), its integrand is
\begin{equation}
  \rho=u^2+\frac{\pi^2}{\hbar^2}.
\end{equation}

Thus, Pauli's variables provide a different canonical parametrization of the same underlying real phase-space structure. They avoid the explicit appearance of \(\hat H^{-1}\) in the density, but only because the inverse Hamiltonian is hidden in the relation \(U=\hat H^{-1}u\) between Pauli's auxiliary field and Chen's physical real part.

It is worth being precise about what does and does not fail here, because the two single-field descriptions fail in \emph{different} ways. In Pauli's variables, the density above, the current $\bm j=(\hbar^2/m)\bigl[(\hat HU)\nabla\dot U-\dot U\,\nabla(\hat HU)\bigr]$ (for the standard kinetic-plus-potential Hamiltonian), and the momentum functional are all \emph{local} in $(U,\dot U)$ in the operator-theoretic sense of Sec.~\ref{sec:alt}: they involve the fields and finitely many of their spatial derivatives at each point. Taken by itself, the Pauli description is therefore a single-field account that is complete and local in this operator-theoretic sense -- though the operational meaning of $U$ itself, unlike that of $u=\Re\psi$, is fixed only through the Hamiltonian-dependent dictionary $u=\hat HU$, $v=\hbar\dot U$. What it sacrifices instead is the \emph{kinematical autonomy} of the probability rule: the map from the field data to the Born density contains the Hamiltonian explicitly, so the probabilistic reading of the state is entangled with the dynamics. Concretely, for composite systems the joint density contains the total $\hat H_{AB}$, which does not decompose into subsystem expressions (Sec.~\ref{subsec:nonfactorization}), and for time-dependent $\hat H(t)$ the state-to-probability dictionary itself becomes time-dependent. Chen's variables make the complementary trade: the field $u=\Re\psi$ retains its Hamiltonian-independent meaning, but density, current, and momentum then become non-local through $\hat H^{-1}$, and remain Hamiltonian-dependent as well.

The accurate lesson is therefore a dichotomy rather than a blanket impossibility, and its scope should be stated precisely: it is established here for the two natural single-field parametrizations of $\psi$ -- Chen's physical real part and Pauli's auxiliary field -- and a no-go theorem covering every conceivable single-field variable would require a specification of the admissible variables and reconstruction maps, which is not attempted here. In a single-field description of either type, the Born rule can be kept local (Pauli), or the field itself can be kept kinematically autonomous (Chen's $u=\Re\psi$), but the rule cannot be made simultaneously local and independent of the dynamics. Only the canonical pair achieves both at once: $\rho=u^2+\pi^2/\hbar^2$ is local, contains no reference to $\hat H$, and composes across subsystems. 

As a historical aside, one may ask whether this phase-space reading of the real second-order equation could have been recognized earlier. Had Schrödinger's equation, or Pauli's later parametrization, been understood as a reduced form of a real Hamiltonian system, the imaginary unit might have appeared from the beginning as a compact notation for an underlying symplectic structure. This counterfactual point is suggestive, but the
argument above does not depend on it.

\subsection{Phase-space action and reduced non-local Lagrangian}
The Hamiltonian equations already establish the lifted real phase-space
formulation. The action principle is useful for two further reasons. First,
it makes explicit that \(\pi\) is the canonical momentum conjugate to \(u\).
Second, it shows why any action written only in terms of \(u\) and \(\dot u\)
must be non-local: eliminating \(\pi\) necessarily introduces
\(\hat H^{-1}\). The same action will also be used in the next subsection to
derive the local continuity equation from the internal \(SO(2)\) symmetry.

The canonical system \eqref{eq:canonical_equations} follows from the phase-space action
\begin{equation}
  S[u,\pi]  =  \int dt  \left[    \langle \pi|\dot u\rangle    -    \mathcal H[u,\pi]  \right],
\label{eq:S_phase}
\end{equation}
with
\begin{equation}
  \mathcal H[u,\pi]
  =
  \frac12\langle u|\hat H|u\rangle
  +
  \frac{1}{2\hbar^2}\langle \pi|\hat H|\pi\rangle .
\label{eq:H_phase_space}
\end{equation}
Under the identification $\psi=u+i\pi/\hbar$, this functional is one half of the usual complex expectation value, $\mathcal H=\tfrac12\braket{\psi|\hat H|\psi}$; this convention propagates to the overall factor $\hbar/2$ in Eq.~\eqref{eq:so2_lagrangian_variation} and to the corresponding factor $\tfrac12$ in the magnetic variation of Sec.~\ref{sec:mag}. Variation with respect to \(u\) and \(\pi\) gives the Hamiltonian equations.
The stationarity condition with respect to \(\pi\) is
\begin{equation}
  \pi=\hbar^2\hat H^{-1}\dot u .
\end{equation}
Substituting this relation back into the phase-space action gives the
reduced action
\begin{equation}
  S_{\rm red}[u]
  =
  \frac12\int dt
  \left[
    \hbar^2\langle \dot u|\hat H^{-1}|\dot u\rangle
    -
    \langle u|\hat H|u\rangle
  \right].
\end{equation}
Thus, the non-locality of the reduced action 
is the direct result of eliminating the canonical momentum.

 \subsection{Internal $SO(2)$ symmetry and the continuity equation\label{ssec:so2_continuity}}
 In standard quantum mechanics, the transformation 
 \begin{equation}
 \psi\mapsto \psi^\prime =e^{i\theta}\psi
 \end{equation}
 leaves the physical state unchanged. This global $U(1)$ symmetry implies, via Noether's theorem, a conserved current whose time component is $\rho = |\psi|^2$. This quantity is subsequently interpreted as the probability density. In our formulation, this symmetry translates into a global  $SO(2)$ symmetry in the local real phase-space. This correspondence is mathematically exact: the groups \(U(1)\) and \(SO(2)\) are isomorphic as one-dimensional compact Lie groups. What differs is not the abstract symmetry itself, but only its representation -- as multiplication by a complex phase in the standard formalism and as a real rotation in the two-dimensional quadrature plane in the real phase-space formulation.

For the Hamiltonian 
\begin{equation}
\hat H = -\frac{\hbar^2}{2m}\nabla^2 + V(\bm r),
\label{eq:H_standard}
\end{equation}
Eq.~\eqref{eq:S_phase} may be written, up to a total time derivative, in the symmetrized local form
\begin{equation}
S_{\rm sym}[u,\pi]=\int_{{\mathbb{R}}^3} \dd t\,\dd^3r\,\mathcal L_{\rm sym},
\label{eq:S_sym}
\end{equation}
with
\begin{equation}
\begin{aligned}
\mathcal L_{\rm sym}
&=
\frac12(\pi\dot u-u\dot\pi)
-\frac{\hbar^2}{4m}(\nabla u)^2
-\frac{1}{4m}(\nabla\pi)^2
\\
&\quad
-\frac12 V(\mathbf r)
\left(
u^2+\frac{\pi^2}{\hbar^2}
\right).
\end{aligned}
\label{eq:L_sym_local}
\end{equation}
The action is invariant under global rotations in the internal
plane \((u,\pi/\hbar)\),
\begin{align}
  u_\theta
  &=
  u\cos\theta-\frac{\pi}{\hbar}\sin\theta,
  \label{eq:so2_rotation_u}
  \\
  \pi_\theta
  &=
  \pi\cos\theta+\hbar u\sin\theta .
  \label{eq:so2_rotation_pi}
\end{align}
Equivalently, with \(v=\pi/\hbar\) and
\begin{equation}    
  \Phi =
  \begin{pmatrix}
    u\\ v
  \end{pmatrix},
\end{equation}
these transformations are rigid rotations of the internal real plane
\(\Phi\mapsto R(\theta)\Phi\). The density
\(\rho=u^2+v^2\) is the Euclidean norm in this plane. Thus, the phase
symmetry is not an additional structure imposed on the real formulation; it
is the rotation symmetry of the canonical phase-space pair itself.
For infinitesimal $\theta$,
\begin{equation}
\delta u = -\theta\,\frac{\pi}{\hbar},\qquad\delta\pi = \theta\,\hbar u.\label{eq:delta_u_pi}
\end{equation}
Just as in complex theory, Noether’s theorem establishes the connection to the probability density. 
Let the rotation angle vary slowly in space and time,
\(\theta\to\theta(\mathbf r,t)\). A direct calculation yields
\begin{equation}
  \delta\mathcal L_{\rm sym}
  =
  -\frac{\hbar}{2}\rho\,\dot\theta
  -
  \frac{\hbar}{2}\mathbf j\cdot\nabla\theta.
  \label{eq:so2_lagrangian_variation}
\end{equation}
In Noether language, \(\rho\) and \(\mathbf j\) are therefore, up to the overall conventional factor \(\hbar/2\), the charge density and spatial current associated with the internal \(SO(2)\) symmetry:
\begin{subequations}
\label{eq:rhoj_real}
\begin{align}
\rho(\bm r,t)
&=
 u(\bm r,t)^2 + \frac{\pi(\bm r,t)^2}{\hbar^2},
\label{eq:prob_density_local}
\\
\bm j(\bm r,t)
&=
\frac{1}{m}\Big(u\nabla\pi-\pi\nabla u\Big).
\label{eq:prob_current_local}
\end{align}
\end{subequations}
In complex notation, with $\psi=u+i\pi/\hbar$, Eqs.~\eqref{eq:rhoj_real} are the familiar $\rho=|\psi|^2$ and $\bm j=(\hbar/m)\,\mathrm{Im}(\psi^*\nabla\psi)$. Noether's theorem, therefore, yields the continuity equation
\begin{equation}
\partial_t\rho + \nabla\cdot\bm j = 0.
\label{eq:continuity_real}
\end{equation}
The Born density is thus local in the canonical variables $(u,\pi)$. The same local continuity equation can also be verified directly from
Eqs.~\eqref{eq:u_dot_pi}--\eqref{eq:pi_dot_u}:
\begin{align}
\partial_t\rho
&= 2u\dot u + \frac{2}{\hbar^2}\pi\dot\pi
\nonumber\\
&= \frac{2}{\hbar^2}\Big(u\hat H\pi-\pi\hat H u\Big)
\nonumber\\
&= -\frac{1}{m}\nabla\cdot\Big(u\nabla\pi-\pi\nabla u\Big).
\end{align}
However, the Noether derivation is more informative: it shows that the Born density and current are tied to the internal \(SO(2)\) symmetry of the real phase-space plane.

One further consequence deserves explicit mention because it concerns a convention that is invisible in complex notation: the additive constant of the energy. In ordinary quantum mechanics, the shift $\hat H\mapsto\hat H+c\,\mathbb 1$ multiplies every solution by the global phase $e^{-ict/\hbar}$ and changes nothing observable. In the lifted real formulation, the same shift acts as a rigid, uniformly time-dependent rotation of every local plane, $u\mapsto\cos(ct/\hbar)\,u+\sin(ct/\hbar)\,v$ and $v\mapsto\cos(ct/\hbar)\,v-\sin(ct/\hbar)\,u$, under which the quadratic invariants of this subsection -- the density, the current, and the generator pairings -- are unchanged: the energy-zero shift is a manifest symmetry. In the reduced single-field description, by contrast, nothing of this is visible as a symmetry. The shift changes the field itself, since $u=\Re\psi$ mixes with the eliminated quadrature; it changes the second-order equation, $\hat H^2\mapsto(\hat H+c)^2$; it changes the reconstruction map, $\hat H^{-1}\mapsto(\hat H+c)^{-1}$; and it can create zero modes by moving a spectral point to the origin (Appendix~\ref{app:hinv}). The energy-zero convention, physically empty in the standard theory, thereby becomes entangled with the very definition of the reduced variables -- one more expression of the fact that the one-field reduction is not an autonomous local representation of the physical ray.

\section{Free wave packets in real phase space}
\label{sec:gaussian_green_reconstruction}
The equivalence established above guarantees that the standard free
Gaussian packet can be represented in real variables. 
This example shows how the spreading packet is encoded by the
second-order real equation, how the second Cauchy datum fixes the internal
quadrature of the Fourier modes, and how the canonical momentum reconstructed
through \(H^{-1}\) can remain localized despite the formal non-locality of
the Green-function kernel.

\subsection{Mode-by-mode solution of the free real equation}
For the free particle,
\begin{equation}
 \hat H=-\frac{\hbar^2}{2m}\nabla^2,
 \qquad
 \hat H^{\,2}=\frac{\hbar^4}{4m^2}\nabla^4,
\end{equation}
and the real second-order equation becomes
\begin{equation}
 \ddot u+\hbar^{-2}\hat H^{\,2}u=0,
 \label{eq:real_second_order_general}
\end{equation}
that is,
\begin{equation}
 \ddot u+\frac{\hbar^2}{4m^2}\nabla^4 u=0.
 \label{eq:real_second_order_free}
\end{equation}

With the Fourier conventions
\begin{eqnarray}
 u(\bm r,t)=\frac{1}{(2\pi)^{3/2}}\int_{\mathbb R^3}\tilde u(\bm k,t)e^{i\bm k\cdot\bm r}\,\dd^3k,
 \nonumber \\
 \tilde u(\bm k,t)=\frac{1}{(2\pi)^{3/2}}\int_{\mathbb R^3}u(\bm r,t)e^{-i\bm k\cdot\bm r}\,\dd^3r,
\end{eqnarray}
Eq.~\eqref{eq:real_second_order_free} reduces mode by mode to
\begin{equation}
 \partial_t^2 \tilde u(\bm k,t)+\omega_k^2\tilde u(\bm k,t)=0,
 \qquad
 \omega_k=\frac{\hbar k^2}{2m}.
 \label{eq:free_mode_oscillator}
\end{equation}
Hence the exact real evolution law is
\begin{equation}
 \tilde u(\bm k,t)
 =\tilde u(\bm k,0)\cos(\omega_k t)
 +\frac{\partial_t\tilde u(\bm k,0)}{\omega_k}\sin(\omega_k t).
 \label{eq:free_mode_solution_general}
\end{equation}
This is the essential point: the real theory propagates each Fourier mode as an ordinary oscillator, and the second initial datum $\partial_t\tilde u(\bm k,0)$ fixes the internal orientation of that oscillator.

\subsection{Exact Gaussian packet, canonical momentum, and Born density}

Choose the normalized initial profile
\begin{equation}
 u(\bm r,0)=u_0(\bm r)=\frac{1}{(\pi\sigma_0^2)^{3/4}}\exp\!\left(-\frac{r^2}{2\sigma_0^2}\right),
 \quad r=|\bm r|,
 \label{eq:gaussian_initial_u}
\end{equation}
and take
\begin{equation}
 \dot u(\bm r,0)=0.
 \label{eq:gaussian_initial_udot}
\end{equation}
The choice \eqref{eq:gaussian_initial_udot} selects the even-in-time member of the packet family: it fixes the internal $SO(2)$ orientation so that the initial canonical momentum vanishes. Equivalently, in complex notation, Eq.~\eqref{eq:gaussian_initial_udot} states that the initial wave function $\psi(\cdot,0)=u_0$ is real: the even-in-time member of the packet family is precisely the packet with real initial data. In Fourier space, $\tilde u(\bm k,0)=(\sigma_0^2/\pi)^{3/4}\exp(-\sigma_0^2k^2/2)$, and Eq.~\eqref{eq:free_mode_solution_general} gives $\tilde u(\bm k,t)=\tilde u(\bm k,0)\cos(\omega_k t)$, while the canonical momentum $\pi=\hbar^2\hat H^{-1}\dot u$ of Eq.~\eqref{eq:pi_canonical} becomes
\begin{equation}
 \tilde\pi(\bm k,t)=\frac{\hbar^2}{E_k}\,\partial_t\tilde u(\bm k,t)
 =-\hbar\,\tilde u(\bm k,0)\sin(\omega_k t),
 \label{eq:gaussian_pi_kspace}
\end{equation}
with $E_k=\hbar\omega_k$. The two inverse transforms are elementary and yield the exact real pair
\begin{equation}
 u(\bm r,t)=A(\bm r,t)\cos\Phi(\bm r,t),
 \qquad
 \frac{\pi(\bm r,t)}{\hbar}=A(\bm r,t)\sin\Phi(\bm r,t),
 \label{eq:u_v_gauss}
\end{equation}
with the common envelope
\begin{eqnarray}
 A(\bm r,t) & = & \frac{(1+\tau^2)^{-3/4}}{(\pi\sigma_0^2)^{3/4}}\,
 \exp\!\left(-\frac{r^2}{2\sigma_t^2}\right),
 \quad
 \tau:=\frac{\hbar t}{m\sigma_0^2},
 \nonumber\\
 \sigma_t& = &\sigma_0\sqrt{1+\tau^2},
 \label{eq:gauss_width}
\end{eqnarray}
and the common phase
\begin{equation}
 \Phi(\bm r,t)=\frac{\tau r^2}{2\sigma_t^2}-\frac{3}{2}\arctan\tau .
 \label{eq:phasePhi}
\end{equation}
Hence the local Born density is
\begin{equation}
\rho(\bm r,t)
=u^2+\frac{\pi^2}{\hbar^2}
=A^2
=\frac{\exp\!\left(-r^2/\sigma_t^2\right)}{(\pi\sigma_t^2)^{3/2}} .
 \label{eq:gaussian_density_real}
\end{equation}
Thus the free Gaussian packet 
is obtained entirely within the real phase-space description, without first introducing a two-component wavefunction. The packet disperses through the widening of $\sigma_t$, while the real pair $(u,\pi/\hbar)$ oscillates through the common phase $\Phi$. The role of the second initial condition is transparent: it fixes the quadrature needed to convert the oscillatory real field into a positive density.

\subsection{Green-function reconstruction and state-dependent containment\label{sec:non-loc-rev}}

We now return to the non-locality issue discussed in Sec.~\ref{sec:alt}. Chen had hoped to show that the action of \(\hat H^{-1}\) on \(\dot u\) effectively reduces to a local operation. As explained there, his particular Green's function argument does not work. The Gaussian packet nevertheless shows that the underlying intuition contains a grain of truth. The reconstruction map remains non-local as an operator, but for special Schrödinger states, its non-local tails can be strongly suppressed or absent.

For the free particle, the reconstruction $\pi=\hbar^{2}\hat H^{-1}\dot u$, Eq.~\eqref{eq:pi_canonical}, is equivalent to the Poisson equation
\begin{equation}
 -\nabla^2 \pi(\bm r,t)=2m\,\dot u(\bm r,t).
 \label{eq:poisson_for_pi}
\end{equation}
With the boundary condition $\pi(\bm r,t)\to 0$ as $|\bm r|\to\infty$, the unique solution is
\begin{equation}
 \pi(\bm r,t)=2m\int_{\mathbb R^3}\frac{\dot u(\bm r',t)}{4\pi|\bm r-\bm r'|}\,\dd^3r'.
 \label{eq:v_green_3d}
\end{equation}
This is the free-space reconstruction formula for the canonical momentum.

The kernel in Eq.~\eqref{eq:v_green_3d} makes the operator-level non-locality diagnosed in Sec.~\ref{sec:alt} manifest. It also yields a sharper, state-level formulation, which preempts the objection that non-local \emph{representations} of local dynamics are commonplace:

\begin{proposition}[state-level non-locality of the Born density]
\label{prop:statelevel}
Let $\hat H=-\frac{\hbar^2}{2m}\nabla^2$ on $L^2(\mathbb R^3)$, and let $\rho=u^2+\pi^2/\hbar^2$ with $\pi=\hbar^2\hat H^{-1}\dot u$ be the Born density in the reduced variables. For every point $\bm r_0$ and every time $t_0$, there exist admissible Cauchy data $(u_1,\dot u_1)$ and $(u_2,\dot u_2)$, specified at $t_0$, which coincide on an open neighbourhood of $\bm r_0$  (and hence agree at $\bm r_0$ together with all their spatial derivatives )  while the associated Born densities differ at $(\bm r_0,t_0)$. In particular, $\rho(\bm r_0,t_0)$ is not a function of $u$, $\dot u$, and finitely many of their derivatives at $\bm r_0$.
\end{proposition}

\begin{proof}
Let $(u_1,\dot u_1)$ be admissible, for instance, the Gaussian data \eqref{eq:gaussian_initial_u}--\eqref{eq:gaussian_initial_udot}, where "admissible” is meant in the precise sense of Appendix A. Choose a smooth $g\ge 0$, $g\not\equiv 0$, supported in a closed ball $B$ that lies strictly above the plane $\{z=z_0\}$ and does not contain $\bm r_0$, and set
\begin{equation*}
 u_2=u_1,
 \qquad
 \dot u_2=\dot u_1+\lambda\,\partial_z g ,
 \qquad
 \lambda\neq 0 .
\end{equation*}
The two data sets coincide outside $B$, hence in a neighbourhood of $\bm r_0$. Because $\partial_z g$ is the derivative of a compactly supported function, its integral vanishes, and $\hat H^{-1}\partial_z g$ decays like a dipole field, $\sim|\bm r|^{-2}$, at infinity; it is therefore square integrable, and the perturbed data remain admissible, $\dot u_2\in\operatorname{Ran}\hat H$ (Appendix~\ref{app:hinv}). By Eq.~\eqref{eq:v_green_3d} and one integration by parts (all fields evaluated at $t_0$),
\begin{equation*}
 \pi_2(\bm r_0)-\pi_1(\bm r_0)
 =\frac{\lambda m}{2\pi}\int_{B}
 g(\bm r')\,\frac{z'-z_0}{|\bm r_0-\bm r'|^{3}}\,\dd^3r'
 =:\lambda\,\delta ,
\end{equation*}
with $\delta>0$, since the integrand is non-negative on $B$ and strictly positive where $g>0$. Hence
$\rho_2(\bm r_0,t_0)-\rho_1(\bm r_0,t_0)
=\bigl(2\lambda\delta\,\pi_1(\bm r_0)+\lambda^{2}\delta^{2}\bigr)/\hbar^{2}$,
which vanishes for at most one non-zero value of $\lambda$.
\end{proof}

The non-locality established by Proposition~\ref{prop:statelevel} concerns which data determine the density; it is logically independent of the question of how long the non-local tails of $\pi$ are for a given state. The remainder of this subsection separates the two issues: admissibility removes the monopole obstruction -- and with it the leading $1/r$ tail -- for \emph{every} admissible datum, and can strongly suppress the remaining tails for special states such as the Gaussian packet of Eqs.~\eqref{eq:gaussian_initial_u}--\eqref{eq:gaussian_initial_udot}. It does not, by itself, imply Gaussian localisation or the absence of all algebraic multipole tails: the dipolar perturbation used in the proof above is admissible and produces a $\sim r^{-2}$ tail in $\pi$.

At first sight, the Coulomb kernel in Eq.~\eqref{eq:v_green_3d} suggests an algebraic
tail decaying like the inverse of distance. For admissible data, however, the source is constrained by the Schrödinger evolution, and the leading obstruction is absent. The Cauchy data $(u,\dot u)$ descend from a normalisable Schrödinger solution precisely when $\dot u\in\operatorname{Ran}\hat H$, i.e., when $\dot u=\hat Hv/\hbar$ for a normalisable $v$; see Eq.~\eqref{equ:coup1}. For the free particle, the source of the Poisson equation \eqref{eq:poisson_for_pi} is then itself a total divergence, $2m\,\dot u=-\hbar\,\nabla^2v$, so that
\begin{equation}
 \int_{\mathbb R^3}\dot u(\bm r,t)\,\dd^3r
 =-
 \frac{\hbar}{2m}\lim_{R\to\infty}
 \oint_{S_R}\nabla v\cdot \dd\bm S
 =0.
 \label{eq:neutrality_pi_source}
\end{equation}
The monopole moment of the source therefore vanishes \emph{by construction},  
and for this reason the Green reconstruction produces no $1/r$ field. The decay beyond the monopole order is state-dependent: for the Gaussian data \eqref{eq:gaussian_initial_u}--\eqref{eq:gaussian_initial_udot}, the reconstruction inherits the Gaussian decay of the source, as the short-time cheque below makes this explicit; whereas the dipolar data of Proposition~\ref{prop:statelevel} retain their algebraic tail. The admissibility condition $\dot u\in\operatorname{Ran}\hat H$ is discussed together with the other standing assumptions in Appendix~\ref{app:hinv}.

A useful short-term cheque is obtained directly from the real equation. Since \eqref{eq:gaussian_initial_udot} implies $\dot u(\bm r,t)=\mathcal O(t)$,
\begin{equation}
 \dot u(\bm r,t)=-\frac{\hbar^2 t}{4m^2}\nabla^4 u_0(\bm r)+\mathcal O(t^3),
\end{equation}
and therefore
\begin{equation}
 \pi(\bm r,t)=\frac{\hbar^2 t}{2m}\nabla^2 u_0(\bm r)+\mathcal O(t^3).
 \label{eq:pi_smallt_local}
\end{equation}
For the Gaussian initial profile \eqref{eq:gaussian_initial_u}, this is a polynomial multiplied by the original Gaussian. The absence of algebraic tails is thus already visible at leading order.

This gives a controlled version of the intuition behind Chen's argument. The map \(\hat H^{-1}\) remains non-local as an operator, but for Schrödinger data with the appropriate global constraints, its action need not produce the leading $1/r$ tail suggested by the Green function alone. Given the equivalence of Sec.~\ref{ssec:equivalence}, the localisation of $\pi=\hbar\,\mathrm{Im}\,\psi$ for the Gaussian packet is, of course, immediate from the known complex solution; the non-trivial content of this subsection is that the same conclusion can be reached from the reduced data $(u,\dot u)$ alone, with the admissibility condition $\dot u\in\operatorname{Ran}\hat H$ doing the work.

\subsection{The harmonic oscillator and coherent packets in brief}
\label{ssec:ho_brief}

The one-dimensional harmonic oscillator can be treated in exactly the same mode-wise fashion, and we only record the outcome here; the construction is collected in Appendix~\ref{sec:ho_ladder_coherent_real}. The spatial spectral theory -- Hermite-functions and real ladder operators -- is unchanged by the real formulation; what changes is the description of time-dependent states. In particular, the coherent packet is not introduced as an eigenstate of a complex annihilation operator but appears as a \emph{phase-locked orbit of real normal modes}: a Poisson-weighted ladder of oscillator levels whose phases are locked linearly in the level index, whose Born density is a rigidly translating Gaussian following the classical orbit, and which saturates the Heisenberg bound $\Delta x\,\Delta p=\hbar/2$. Whether the packet moves to the right, to the left, or merely oscillates in place is fixed by the second Cauchy datum, i.e., by the canonical momentum.

\section{Heisenberg's uncertainty relation and momentum}
\label{sec:fisher_uncertainty_real}
The equivalence established in Sec.~\ref{ssec:equivalence} already guarantees that the usual Heisenberg uncertainty relation is recovered in the real phase-space formulation. This can be seen directly in the momentum moments, which become local functionals of the canonical variables
\((u,\pi)\):
\begin{align}
  \langle p\rangle
  &=  \int_{\mathbb R}  \bigl(    u\,\partial_x\pi-\pi\,\partial_x u  \bigr)\,\dd x ,  \label{eq:p_mean}  \\
  \langle p^2\rangle
  &=  \hbar^2  \int_{\mathbb R}  \left[    (\partial_x u)^2    +    \frac{1}{\hbar^2}(\partial_x\pi)^2  \right]\dd x .  \label{eq:p2_form}
\end{align}
Once these local phase-space expressions are available, the standard Cauchy-Schwarz proof of the uncertainty relation goes through exactly as in Sec.~\ref{ssec:hur}. The Heisenberg bound is therefore not a new test of the real phase-space formulation; it is a reflection of the restored local symplectic geometry.

The purpose of the present section is, therefore, not to prove the bound once more, but to identify where the momentum operator has gone in the real variables. This point matters because the reduced single-field formulation has no local momentum operator acting on \(u\) alone. In Sec.~\ref{sec:momentum}, we saw that momentum expectation values can still be computed, but only after the eliminated imaginary part has been reconstructed through \(\hat H^{-1}\dot u\). In the real  phase-space formulation, the situation is different. The canonical pair \((u,\pi)\) carries the local symplectic structure, and spatial translations act naturally on this pair. Momentum is therefore recovered as the symplectic generator of translations, rather than as an operator on the single field \(u\).

\subsection{Translations and the missing factor of \(i\)}
\label{subsec:translations_second_order}

Assume first that the Hamiltonian is translation invariant in one spatial dimension, $[\hat H,\partial_x]=0$. The real second-order equation is then invariant under the translation group $(T(a)u)(x,t):=u(x-a,t)$, and since $T(a)$ commutes with $\hat H^{-1}$, the conserved norm $\mathcal P[u,\dot u]$ of Sec.~\ref{sec:realH} is translation invariant as well: the reduced formulation already carries a perfectly ordinary representation of spatial translations on its Cauchy data. The role of the second initial datum is transparent for a monochromatic mode $u_k(x,t)=A_k\cos(kx-\omega_k t+\alpha_k)$: a translation by $a$ shifts the phase, and using $\dot u_k=\omega_k A_k\sin(kx-\omega_k t+\alpha_k)$ one may rewrite the shifted mode as
\begin{equation}
  T(a)u_k  =  u_k\cos(ka)  +  \frac{\dot u_k}{\omega_k}\sin(ka).
  \label{eq:translation_mode_rotation}
\end{equation}
For each Fourier mode, a spatial translation therefore acts as a rotation in the oscillator plane $(u_k,\dot u_k/\omega_k)$,   the same pair that already emerged in Sec.~\ref{sec:realH} when the norm was inferred from the monochromatic-wave heuristic.  For this reason,  the momentum associated with translations must depend on both quadratures.

The infinitesimal generator of translations, $A=-\partial_x$, is thus available in the reduced formulation, but it is not yet the quantum-mechanical momentum operator: in the complex theory, $\hat p=-i\hbar\partial_x$, and the factor $i$ turns the generator of translations into a self-adjoint observable. In the reduced single-field formulation, there is no local operation on $u$ alone that can play this role. Once the canonical partner $\pi$ is retained, however, the pair $(u,\pi/\hbar)$ carries a local rotation $J$ with $J^2=-\mathbb 1$ -- the real counterpart of multiplication by $i$ -- and momentum is recovered as the generator of spatial translations combined with this internal rotation.

\subsection{Momentum functional in real phase space}
\label{subsec:translations_reduced_phase_space}

It is convenient to assemble the state into the two-component real field
\begin{equation}
  \Phi(x)
  :=
  \begin{pmatrix}
    u(x)\\[3pt]
    \pi(x)/\hbar
  \end{pmatrix},
  \;
  \|\Phi\|^2
  :=
  \int_{\mathbb R}\left(u^2+\frac{\pi^2}{\hbar^2}\right) \dd x.
  \label{eq:Phi_real}
\end{equation}
Translations act componentwise, $(T(a)\Phi)(x)=\Phi(x-a)$, and preserve not only the
Born norm \eqref{eq:Phi_real} but also the antisymmetric bilinear form
\begin{eqnarray}
  \Omega(\Phi_1,\Phi_2)
 & := &
  \int_{\mathbb R}
  \left(
    u_1\,\frac{\pi_2}{\hbar}
    -
    \frac{\pi_1}{\hbar}\,u_2
  \right) \dd x,
  \nonumber \\
  \Phi_a(x)&=&\binom{u_a(x)}{\pi_a(x)/\hbar}\ (a=1,2).
  \label{eq:symplectic_form}
\end{eqnarray}
The natural scalar associated with translations is then obtained by pairing the
translation generator with this antisymmetric form:
\begin{align}
\langle p\rangle
:&= -\hbar\,\Omega\bigl(\Phi,A\Phi\bigr)
= \hbar\,\Omega\bigl(\Phi,\partial_x\Phi\bigr)
\notag \\
&= \int_{\mathbb R}
\bigl(
u\,\partial_x\pi
-
\pi\,\partial_xu
\bigr)\,\dd x .
\label{eq:momentum_functional}
\end{align}
This is precisely the momentum functional already used in the uncertainty argument. The normalization is fixed by the monochromatic wave: for $u=A\cos(kx-\omega t+\alpha)$ and $\pi=\hbar A\sin(kx-\omega t+\alpha)$ one finds
$u\partial_x \pi-\pi \partial_x u = \hbar kA^2$, so a right-moving plane wave carries momentum $+\hbar k$ per unit norm (plane waves being non-normalizable, this is a statement about densities).

When the Hamiltonian is translation invariant, $\langle p\rangle$ is conserved. Using Eq.~\eqref{eq:canonical_equations} and integrating by parts,
\begin{align}
  \frac{d}{dt}\langle p\rangle
  &=
  2\int_{\mathbb R}
  \left(\dot u\,\partial_x \pi - \dot\pi\,\partial_x u\right) \dd x
  \nonumber\\
  &=
  2\int_{\mathbb R}
  \left[
    \frac{1}{\hbar^2}(\hat H\pi)\,\partial_x \pi
    +
    (\hat H u)\,\partial_x u
  \right] \dd x.
  \label{eq:p_conservation_step}
\end{align}
Each term vanishes separately because $\hat H$ is self-adjoint and commutes with $\partial_x$:
\begin{align}
  \braket{\hat H f | \partial_x f}
  &=
  \braket{f | \hat H\partial_x f}
   =
  \braket{f | \partial_x(\hat H f)}
  \nonumber \\
  &=
  -\braket{\partial_x f | \hat H f}
   =
  -\braket{\hat H f | \partial_x f}.
  \label{eq:Hf-dxf-vanishes}
\end{align}
Hence
\begin{equation}
  \frac{d}{dt}\langle p\rangle = 0.
  \label{eq:p_conserved}
\end{equation}
The first moment of momentum is therefore the Noether charge associated with spatial translations, but it becomes visible only after the second quadrature has been restored.

The conceptual conclusion is now clear. Spatial translations are already present in the real second-order equation. What is absent for a single real field is a non-trivial first-order momentum observable constructed from the symmetric inner product alone. The second moment \(p^2\) is available directly as a positive quadratic form.  However, as Eq.~\eqref{equ:momev} makes explicit, both quadratures contribute to it, \(\langle p^2\rangle=\hbar^2\int[(\partial_x u)^2+(\partial_x v)^2]\,\dd x\), so that \(\hbar^2\int(\partial_x u)^2\,\dd x\) alone is \emph{not} \(\langle p^2\rangle\) -- but the first moment of momentum requires the antisymmetric structure supplied by the second quadrature. Thus \(\langle p\rangle\) becomes local only after the canonical partner \(\pi\) has been restored. This is why the real phase-space formulation reproduces the usual kinematics without introducing an external complex unit; the role normally played by \(i\) is carried by the local symplectic structure.

\section{Magnetic field coupling}
\label{sec:mag}
 We now return to the magnetic case briefly flagged in Sec.~\ref{ssec:mag}. There we noted that the reduced single-field equation assumes that \(\hat H\) is real in the chosen representation, whereas minimal coupling makes the Hamiltonian complex in the position representation. The direct algebraic extension of Chen's reconstruction will be given first. We then reinterpret the same mixing geometrically as the gauging of the internal \(SO(2)\) symmetry identified in Sec.~\ref{ssec:so2_continuity}.


\subsection{The magnetic extension of Chen's reconstruction\label{ssec:magnetic_chen}}

Before turning to the gauge-theoretic interpretation, it is useful to spell out the direct extension of Chen's argument. In the absence of magnetic fields, the Hamiltonian is real in the position representation, and the split \(\psi=u+iv\) gives the coupled equation \eqref{equ:coup12}. Magnetic coupling precisely changes this step since a Hamiltonian like Eq.~\eqref{eq:magnetic_hamiltonian_intro} contains an imaginary contribution. Indeed, writing
\begin{equation}
  \hat H_A=\hat H_R+i\hat H_I,   
\end{equation}
and separating the real and imaginary parts of
\(i\hbar\dot\psi=\hat H_A\psi\), one obtains
\begin{subequations}\label{grp2}
\begin{align}
  \hbar\dot u &= \hat H_Rv+\hat H_Iu,
  \label{eq:magnetic_real_u} \\
  -\hbar\dot v &= \hat H_Ru-\hat H_Iv .
  \label{eq:magnetic_real_v}
\end{align}
\end{subequations} 
Thus, the magnetic field mixes the two real quadratures already at the
first-order level. Assuming \(\hat H_R\) is invertible on the subspace of interest,
Eq.~\eqref{eq:magnetic_real_u} may be solved for the second quadrature,
\begin{equation}
  v  =  \hat H_R^{-1}  \left(    \hbar\dot u-\hat H_Iu  \right).
   \label{eq:v_mag}
\end{equation}
For a time-independent magnetic Hamiltonian, this also gives a closed second-order equation for the single real field \(u\). Differentiating
Eq.~\eqref{eq:magnetic_real_u} and using
Eq.~\eqref{eq:magnetic_real_v}, one obtains
\begin{equation}
  \hbar^2 \ddot u  =  - \hat H_R^2 u
  + \hat H_R \hat H_I \hat H_R^{-1}
    \bigl( \hbar \dot u - \hat H_I u \bigr)
  + \hbar \hat H_I \dot u .
  \label{eq:magnetic_chen_second_order}
\end{equation}
This is the magnetic counterpart of Chen's reduced second-order equation. For \(\hat H_I=0\) and \(\hat H_R=\hat H\), it reduces to the familiar expression $\hbar^2\ddot u+\hat H^2u=0$.

Thus, magnetic coupling does not destroy the reduced single-field description, but it changes its structure: the eliminated quadrature now reappears through terms first order in time and through the ordered product \(\hat H_R\hat H_I\hat H_R^{-1}\).

In the real phase-space formulation, the same reconstruction is expressed by the canonical momentum
\begin{equation}
  \pi=\hbar v =
  \hbar\hat H_R^{-1}
  \left(
    \hbar\dot u-\hat H_Iu
  \right).
  \label{eq:pi_mag}
\end{equation}
This is the magnetic counterpart of \(\pi=\hbar^2\hat H^{-1}\dot u\). The vector potential does not merely add a
new potential term to the reduced equation; it shifts the velocity-like datum by the rotational term \(\hat H_Iu\). This is why Chen's original equation fails as a real single-field equation.

This  result has also a geometric interpretation: The operator \(\hat H_I\) is precisely the part of the magnetic Hamiltonian that rotates the internal \((u,v)\)-plane. Magnetic coupling is therefore most naturally understood in the real phase-space formulation as a gauging of the internal \(SO(2)\) symmetry identified above.

\subsection{Local $SO(2)$ rotations and covariant derivatives}
 
The algebraic mixing displayed in the previous subsection has a simple
geometric interpretation. In the real phase-space formulation, the local
internal plane is described by
\begin{equation}
  \Phi(\mathbf r,t)
  =
  \begin{pmatrix}
    u(\mathbf r,t)\\
    v(\mathbf r,t)
  \end{pmatrix},
  \qquad
  v:=\frac{\pi}{\hbar}.
\end{equation}
The Born density is the Euclidean norm in this plane,
\begin{equation}
  \rho(\mathbf r,t)=\Phi^T\Phi=u^2+v^2 .
\end{equation}
In the absence of a vector potential, this internal plane is invariant under
rigid rotations
\begin{equation}
  \Phi \mapsto R(\chi)\Phi,
  \qquad
  R(\chi)=
  \begin{pmatrix}
    \cos\chi & -\sin\chi \\
    \sin\chi & \cos\chi
  \end{pmatrix}.
  \label{eq:magnetic_rigid_SO2}
\end{equation}
This is the real form of the global phase symmetry discussed in
Sec.~\ref{ssec:so2_continuity}. Magnetic coupling is obtained by promoting this rigid rotation to a local one, \(\chi=\chi(\mathbf r)\) (Fig.~\ref{fig:phase_plane}(c)). Under such local rotations, ordinary derivatives of \(u\) and \(v\) no longer transform covariantly. Differentiating Eq.~\eqref{eq:magnetic_rigid_SO2} gives
\begin{align}
  \partial_i u'
  &= \cos\chi\,\partial_i u - \sin\chi\,\partial_i v
     - (\partial_i\chi)\,v',
  \\
  \partial_i v'
  &= \cos\chi\,\partial_i v + \sin\chi\,\partial_i u
     + (\partial_i\chi)\,u'.
\end{align}
The extra terms proportional to $\partial_i\chi$ are precisely the obstruction to local invariance. To compensate for them, we introduce a vector field $\mathbf A(\mathbf r)$ and define
\begin{equation}
  D_i u := \partial_i u + \frac{q}{\hbar}A_i v,
  \qquad
  D_i v := \partial_i v - \frac{q}{\hbar}A_i u.
  \label{eq:covariant_derivatives_real}
\end{equation}
If simultaneously
\begin{equation}
  A_i' = A_i + \frac{\hbar}{q}\,\partial_i\chi,
  \label{eq:gauge_transform_real}
\end{equation}
then one checks directly that
\begin{eqnarray}
  D_i u'  &= & \cos\chi\,D_i u - \sin\chi\,D_i v,
  \nonumber \\
  D_i v'  & = & \cos\chi\,D_i v + \sin\chi\,D_i u.
  \label{eq:covariant_transformation_real}
\end{eqnarray}
The magnetic vector potential, therefore, appears as the connection required to compare neighboring internal phase-space planes when the $SO(2)$ angle is allowed to vary from point to point.
This construction is entirely real, and   the antisymmetric mixing of $u$ and $v$ in Eq.~\eqref{eq:covariant_derivatives_real} is dictated by the symplectic structure of phase-space. 

One restriction of scope should be stated explicitly. Throughout this section, the rotation angle is static, $\chi=\chi(\mathbf r)$, and the vector potential time-independent: we treat spatial gauge covariance in the magnetostatic setting. For a time-dependent angle $\chi(\mathbf r,t)$, the covariant completion also involves the time direction, and the scalar potential shifts as $V\mapsto V-\hbar\,\partial_t\chi$; its spatially constant, linear-in-time special case $\chi=-ct/\hbar$ is precisely the additive energy shift discussed at the end of Sec.~\ref{ssec:so2_continuity}.

\subsection{Gauged Hamiltonian density and current}
The preceding subsection introduced the vector potential as the connection for local rotations in the internal \((u,v)\)-plane. We now show that this geometric construction reproduces the usual Hamiltonian and probability current of minimal coupling. The point is not to repeat the Chen-style reconstruction of Sec.~\ref{ssec:magnetic_chen}, but to identify the local Hamiltonian structure whose reduction leads to it.

The local, gauge-covariant Hamiltonian functional $\mathcal H_A[u,v]$ is now
\begin{equation}
\begin{aligned}
  \mathcal H_A
  &=
  \frac12
  \int_{\mathbb R^3} \dd^3 r\,
  \biggl[
    \frac{\hbar^2}{2m}
    \sum_i
    \bigl(
      (D_i u)^2 + (D_i v)^2
    \bigr)
    \\
    &\qquad\qquad
    +\, V(\mathbf r)\bigl(u^2+v^2\bigr)
  \biggr] .
\end{aligned}
\label{eq:Hmag}
\end{equation}
Because of Eq.~\eqref{eq:covariant_transformation_real}, this Hamiltonian is invariant
under the local $SO(2)$ transformation. Expanding Eq.~\eqref{eq:Hmag} in ordinary derivatives gives
\begin{equation}
\begin{aligned}
\mathcal H_A[u,v]
&= \frac12\int_{\mathbb R^3} \dd^3r\,\Bigg[ \frac{\hbar^2}{2m} \bigl((\nabla u)^2+(\nabla v)^2\bigr)\\
&\quad - \frac{q\hbar}{m}\,\mathbf A\cdot \bigl(u\nabla v-v\nabla u\bigr) \\
&\quad + \left(\frac{q^2A^2}{2m}+V(\mathbf r) \right)(u^2+v^2) \Bigg].
\end{aligned}
\label{eq:Hmag_expanded_real}
\end{equation}
Two remarks are worth making. First, the term linear in $\mathbf A$ is the usual minimal-coupling term, proportional to the canonical momentum density already identified in the translation analysis:
\begin{equation}
  \mathbf g_{\rm can}
  := \hbar\,(u\nabla v - v\nabla u)
  = u\nabla\pi - \pi\nabla u.
  \label{eq:gcan_mag}
\end{equation}
Thus, the vector potential does indeed couple to momentum density.  However, local $SO(2)$ invariance does not permit one to stop at this linear term alone; it forces the full covariant completion, including the diamagnetic term proportional to $A^2\rho$.

Second, Eq.~\eqref{eq:Hmag_expanded_real} may be rewritten in operator form as
\begin{equation}
  \mathcal H_A[u,v]
  = \frac12\langle u|\hat H_R|u\rangle
    + \frac12\langle v|\hat H_R|v\rangle
    - \langle u|\hat H_I|v\rangle,
  \label{eq:Hmag_split_real}
\end{equation}
with
\begin{align}
  \hat H_R &=
  -\frac{\hbar^2}{2m}\nabla^2
  + \frac{q^2A^2}{2m}
  + V(\mathbf r),
  \\
  \hat H_I &=
  \frac{\hbar q}{2m}\Big(\mathbf A\cdot\nabla + \nabla\cdot\mathbf A\Big).
  \label{eq:HI_def}
\end{align}
Here, the second operator is to be understood in the symmetrized sense,
\begin{equation}
  \Big(\mathbf A\cdot\nabla + \nabla\cdot\mathbf A\Big)f
  := \mathbf A\cdot\nabla f + \nabla\cdot(\mathbf A f),
\end{equation}
so that $\hat H_I$ is skew-adjoint with the usual boundary conditions.

The gauge-covariant probability current is obtained from the response of the
Hamiltonian to the external vector potential. Varying \(\mathcal H_A\) with
respect to \(\mathbf A\), while keeping \(u\) and \(v\) fixed, gives
\begin{equation}
  \delta \mathcal H_A
  = -\frac{q}{2}\int \dd^3r\,\mathbf j\cdot \delta\mathbf A,
\end{equation}
with
\begin{eqnarray}
  \mathbf j
 & = & \frac{\hbar}{m}\bigl(u\nabla v - v\nabla u\bigr)
    - \frac{q}{m}\mathbf A\,(u^2+v^2)\nonumber \\
 &  = & \frac{1}{m}\bigl(u\nabla\pi - \pi\nabla u\bigr)
    - \frac{q}{m}\mathbf A\,\rho.
  \label{eq:j_magnetic_real}
\end{eqnarray}
The first term is the canonical current; the second is the gauge correction that turns
it into the kinematic one.  This is the real-variable version of minimal coupling. 

\subsection{Gyroscopic dynamics, reduced action, and non-local norm}

The magnetic Chen equation \eqref{eq:magnetic_chen_second_order} derived in Sec.~\ref{ssec:magnetic_chen} can be written more compactly as
\begin{equation}
  \ddot u-\hat\Gamma\dot u+\hat\Omega^2u=0,
  \label{eq:gyro_form}
 \end{equation}
with
\begin{equation}
  \hat\Gamma
  =
  \frac{1}{\hbar}
  \left(
    \hat H_I+\hat H_R\hat H_I\hat H_R^{-1}
  \right),
  \label{eq:magnetic_gamma}  
\end{equation}
and
\begin{equation}
  \hat\Omega^2
  =
  \frac{1}{\hbar^2}
  \left(
    \hat H_R^2
    +
    \hat H_R\hat H_I\hat H_R^{-1}\hat H_I
  \right).
  \label{eq:magnetic_omega}
\end{equation}
The first-order time-derivative term $\hat\Gamma$ is the distinctive new ingredient.  In the mechanical analogy, it is gyroscopic: it rotates the local phase-space plane but does not, by itself, contribute to the conserved norm. More precisely, $\hat\Gamma$ is in general not skew-adjoint in the ordinary $L^2$ inner product when $\hat H_R$ and $\hat H_I$ fail to commute; its gyroscopic character holds with respect to the reduced mass metric $\langle f|\hat H_R^{-1}|g\rangle$ that also governs the reduced action \eqref{eq:mag_reduced_action}: using $\hat H_I^T=-\hat H_I$, one checks directly that $\hat\Gamma^{T}\hat H_R^{-1}=-\hat H_R^{-1}\hat\Gamma$.

The same elimination can be performed at the action level. Eliminating \(v\), or equivalently \(\pi=\hbar v\), from the first-order real phase-space action yields 
\begin{equation}
\begin{aligned}
  S_A[u]
  &=
  \frac12\int_{\mathbb R}\dd t\,
  \Bigg[
    \hbar^2
    \big\langle
      \dot u-\hbar^{-1}\hat H_Iu
      \big|
      \hat H_R^{-1}
      \big|
      \dot u-\hbar^{-1}\hat H_Iu
    \big\rangle
    \\
  &\hspace{3.5cm}
    -
    \langle u|\hat H_R|u\rangle
  \Bigg].
\end{aligned}
\label{eq:mag_reduced_action}
\end{equation}
whose canonical momentum is precisely Eq.~\eqref{eq:pi_mag}. Again, the non-local operator $\hat H_R^{-1}$ appears when the local second quadrature is eliminated in favor of the reduced variables $(u,\dot u)$.

The same pattern appears in the norm. In the local real phase-space
variables one has
\begin{equation}
  \mathcal P_{\rm mag}
  =
  \int_{\mathbb R^3} d^3r
  \left(
    u^2+\frac{\pi^2}{\hbar^2}
  \right).
\end{equation}
In the reduced variables, however, the same norm becomes a non-local functional. Using Eq.~\eqref{eq:v_mag}, one obtains
\begin{equation}
  \mathcal P_{\rm mag}  =  \int_{\mathbb R^3} d^3r  \left[   u^2+    \left(      \hat H_R^{-1}      \left(        \hbar\dot u-\hat H_Iu      \right)    \right)^2  \right].
\end{equation}
Thus, the magnetic field does not force us back to complex numbers. It shows, rather, that the local object being gauged is the real phase-space plane; non-locality appears only after the second quadrature has again been eliminated. In Appendix \ref{sec:landau_real} we use the formalism developed here to discuss Landau levels in real veriables.

\section{Spin-$\tfrac12$ in the real-variable formulation}
\label{sec:spin_half_real}

Spin provides a useful further test case for the real phase-space description. Unlike the examples discussed so far, the relevant twofold structure is not produced by spatial motion or by a magnetic gauge choice, but is already present as an internal degree of freedom. In the standard formulation this structure is encoded by a two-component complex spinor and the Pauli matrices. In the real phase-space formulation, the same information is represented by two real channels together with their canonical partners.

The point of the following discussion is not to rederive spin quantum mechanics from scratch. Rather, it is to show where the familiar Bloch-vector geometry and the non-commuting spin components are located in real variables. In particular, spin makes especially clear that the coordinate \(u\) alone
does not contain all operationally relevant information: the companion quadrature carried by \(\pi\) is needed as soon as one compares spin measurements along different axes.

\subsection{Two real channels and the internal phase space}
\label{subsec:spin_realification}

The magnetic  problem discussed above  and the Landau   level solution of App. \ref{sec:landau_real}   involve a single real field.  A spin-$\tfrac12$ particle adds an internal two-state degree of freedom.  To exhibit its structure most clearly, we suppress the spatial coordinate and discuss a single internal state; for a spatial wave packet, the same formulas hold after integrating over space.

Choose a reference axis, denoted $z$, and resolve the reduced real state into two channels, $u(t)=(u_1(t),u_2(t))^T$. When the internal Hamiltonian is diagonal in this basis -- for example, for a Zeeman field parallel to $z$ -- each channel satisfies the same real second-order oscillator equation. Exactly as in the scalar case, the intensity of a channel is therefore not given by $u_a^2$ alone, but by the invariant quadratic combination of the coordinate with its second datum. It is therefore convenient to package the latter into the canonical momentum $\pi(t)=(\pi_1(t),\pi_2(t))^T$ and to use the real phase-space vector
\begin{equation}
  \Phi:=\begin{pmatrix}u\\ \pi/\hbar\end{pmatrix}\in\mathbb R^4.
  \label{eq:Phi_spin_def}
\end{equation}
Normalized pure states satisfy
\begin{equation}
  \|\Phi\|^2 = u\cdot u + \pi\cdot\pi/\hbar^2 =1.
  \label{eq:spin_norm_real}
\end{equation}
In the reference basis, the two outcome probabilities are simply the channel intensities.
\begin{equation}
\begin{gathered}
p_{\uparrow_z}
= u_1^2+\frac{\pi_1^2}{\hbar^2},
\qquad
p_{\downarrow_z}
= u_2^2+\frac{\pi_2^2}{\hbar^2},
\\[2pt]
\text{with:}\quad p_{\uparrow_z}+p_{\downarrow_z}=1 .
\end{gathered}
\label{eq:spin_z_probabilities}
\end{equation}
Thus, a state polarized along $+\hat z$ is represented by a concentration of the norm in the first channel, while a state polarized along $-\hat z$ is represented by a concentration in the second.


\subsection{Real spin observables and the Bloch vector}
\label{subsec:real_pauli}

A measurement along another axis must be represented by a real symmetric bilinear form on the same internal phase-space.  Introduce the real $2\times 2$ matrices

\begin{equation}
\begin{gathered}
X =
\begin{pmatrix}0&1\\[2pt]1&0\end{pmatrix},
\qquad
Y =
\begin{pmatrix}0&-1\\[2pt]1&0\end{pmatrix},
\\
Z =
\begin{pmatrix}1&0\\[2pt]0&-1\end{pmatrix}.
\end{gathered}
 \end{equation}
and define the three real symmetric $4\times 4$ matrices

\begin{equation}
\begin{gathered}
\Sigma_x =
\begin{pmatrix}
X & 0\\[2pt]
0 & X
\end{pmatrix},
\qquad
\Sigma_y =
\begin{pmatrix}
0 & -Y\\[2pt]
Y & 0
\end{pmatrix},
\\
\Sigma_z =
\begin{pmatrix}
Z & 0\\[2pt]
0 & Z
\end{pmatrix}.
\end{gathered}
\label{eq:Sigma_matrices}
\end{equation}
They satisfy
\begin{equation}
  \Sigma_i^T=\Sigma_i,\quad
  \Sigma_i^2=I_4,\quad
  \Sigma_i\Sigma_j+\Sigma_j\Sigma_i=2\delta_{ij}I_4.
  \label{eq:Sigma_anticommutation}
\end{equation}
For any unit vector $\mathbf n\in\mathbb R^3$, define
\begin{equation}
  \Sigma_{\mathbf n}:=n_x\Sigma_x+n_y\Sigma_y+n_z\Sigma_z,
  \qquad
  \hat S_{\mathbf n}\ \widehat{=}\ \frac{\hbar}{2}\Sigma_{\mathbf n}.
  \label{eq:Sigma_n}
\end{equation}
The spin expectation values are then the three real quadratic forms
\begin{equation}
  s_i(\Phi):=\Phi^T\Sigma_i\Phi,
  \qquad
  \mathbf s(\Phi):=(s_x,s_y,s_z).
  \label{eq:spin_expectation_def}
\end{equation}
Explicitly,
\begin{eqnarray}
  s_x&=&2\left(u_1u_2+\frac{\pi_1\pi_2}{\hbar^2}\right),\nonumber\\
  s_y&=&\frac{2}{\hbar}\left(u_1\pi_2-u_2\pi_1\right),\nonumber\\
  s_z&=&\left(u_1^2+\frac{\pi_1^2}{\hbar^2}\right)-\left(u_2^2+\frac{\pi_2^2}{\hbar^2}\right).
  \label{eq:spin_expectations_real}
\end{eqnarray}
These formulas make the geometry transparent.  The component $s_z$ is the difference of the two channel intensities \eqref{eq:spin_z_probabilities}.  The component $s_x$ measures the in-phase correlation between the two channels.  The component $s_y$ is a genuine phase-space quantity: it is proportional to the oriented area $u_1\pi_2-u_2\pi_1$ in the internal $(u,\pi)$-plane. Thus, the second initial datum is already indispensable for spin along the $y$-axis.

For a normalised pure state, one finds
\begin{equation}
  s_x^2+s_y^2+s_z^2=1.
  \label{eq:bloch_sphere_real}
\end{equation}
Hence, the usual Bloch sphere appears here as a sphere of quadratic forms on the real phase-space; for instance, $\Phi_{+z}=(1,0,0,0)^T$, $\Phi_{+x}=(1,1,0,0)^T/\sqrt2$, and $\Phi_{+y}=(1,0,0,1)^T/\sqrt2$ carry the Bloch vectors $\mathbf s=(0,0,1)$, $(1,0,0)$, and $(0,1,0)$, respectively.

This makes explicit what the second quadrature does. A state polarised along $+\hat x$ has equal channel intensities together with a positive in-phase correlation $u_1u_2+\pi_1\pi_2/\hbar^2$; a state polarised along $+\hat y$ has the same channel intensities but is distinguished by the orientated phase-space area $u_1\pi_2-u_2\pi_1$. The coordinate $u$ thus suffices for one preferred measurement basis; the second datum, encoded in $\pi$, becomes unavoidable as soon as non-commuting spin components are compared.

\subsection{Spin rotations as orthogonal transformations (generators)}
\label{subsec:spin_rotations_real}
Spin rotations act on the real state space as norm-preserving linear maps.  Introduce the real antisymmetric generators
\begin{eqnarray}
  K_x&:=&\begin{pmatrix}0&X\\[2pt]-X&0\end{pmatrix},\quad
  K_y:=\begin{pmatrix}Y&0\\[2pt]0&Y\end{pmatrix},\\
  \text{and,}&& K_z:=\begin{pmatrix}0&Z\\[2pt]-Z&0\end{pmatrix}.
  \nonumber \label{eq:K_generators}
\end{eqnarray}
They satisfy
\begin{equation}
\begin{gathered}
K_i^T=-K_i,
\\
[K_i,K_j]=2\varepsilon_{ijk}K_k,
\qquad
[K_i,\Sigma_j]=2\varepsilon_{ijk}\Sigma_k .
\end{gathered}
\label{eq:K_su2}
\end{equation}
A rotation through an angle $\theta$ about the axis $\mathbf n$ is therefore represented by 
\begin{equation}
    O(\mathbf n,\theta):=\exp\!\bigl(\tfrac{\theta}{2}K_{\mathbf n}\bigr)\in SO(4),
\end{equation}
 with 
 \begin{equation}
     K_{\mathbf n}:=n_xK_x+n_yK_y+n_zK_z.
 \end{equation}
It acts on the real state as $\Phi\mapsto O(\mathbf n,\theta)\Phi$ and transforms the observables covariantly: 
\begin{equation}
O(\mathbf n,\theta)\,\Sigma_{\mathbf m}\,O(\mathbf n,\theta)^T=\Sigma_{R(\mathbf n,\theta)\mathbf m},
\end{equation}
where $R(\mathbf n,\theta)\in SO(3)$ is the ordinary spatial rotation. The physical rotation of the spin vector is thus encoded entirely in orthogonal rotations of the real four-dimensional phase-space.

For later use, we record one further property. In the ordered basis $(u_1,u_2,\pi_1/\hbar,\pi_2/\hbar)$, multiplication by $i$ under the identification $\psi_a=u_a+i\pi_a/\hbar$ is represented by the real block matrix
\begin{equation}
  J_4:=\begin{pmatrix}0&-I_2\\ I_2&0\end{pmatrix},
  \qquad
  J_4^2=-I_4 ,
  \label{eq:J4_def}
\end{equation}
the four-dimensional form of the complex structure of Sec.~\ref{sec:real-phase-space-reconstruction}. All three generators commute with it, $[K_i,J_4]=0$: rotations are complex-linear (unitary) maps. The time-reversal map of Sec.~\ref{subsec:time_reversal_kramers_real} will instead \emph{anti}commute with $J_4$, and this distinction between complex-linear and complex-antilinear orthogonal maps is precisely what drives Kramers degeneracy.

Projective spin measurements are described by the real symmetric projectors $P^{(\mathbf n)}_{\pm}=\tfrac12\bigl(I_4\pm\Sigma_{\mathbf n}\bigr)$; the Born rule $p(\pm|\mathbf n)=\tfrac12\bigl(1\pm\mathbf n\cdot\mathbf s\bigr)$, the L\"uders update, and the $\cos^2(\theta/2)$ transition law for sequential Stern--Gerlach measurements all follow from the real algebra \eqref{eq:Sigma_anticommutation} alone. Since this computation is standard, it is deferred to Appendix~\ref{app:seqspin}.
The interplay of spin and orbital angular momentum is discussed in Appendix~\ref{app:zeeman}.

\section{Time-reversal symmetry and Kramers degeneracy in the real phase-space formulation} \label{subsec:time_reversal_kramers_real}

The discussion in Sec.~\ref{ho}, following Callender \cite{callender2023}, already shows how the reduced real description explains the complex conjugation in spinless time-reversal: temporal reflection reverses the second Cauchy datum and hence the reconstructed quadrature. The present section uses this spinless case as the baseline for the spin-\(\frac12\) situation. In the standard formulation, this is the familiar distinction between time-reversal operations which square to \(+\mathbb 1\) and those which square to \(-\mathbb 1\), the latter leading to Kramers degeneracy \cite{Roberts2012,SigwarthMiniatura2022}. In the real phase-space formulation, the same distinction appears as a difference between two real orthogonal time-reversal maps.

\subsection{Spinless time-reversal}
 
For spinless systems, the present discussion reduces to the point already encountered in Sec.~\ref{ho}. In the reduced single-field formulation, time-reversal acts by reversing the second Cauchy
datum,
\begin{equation}
  (u,\dot u)\mapsto (u,-\dot u).
\end{equation}
Equivalently, in the real phase-space formulation, we denote the spinless
time-reversal map by
\begin{equation}
  \mathfrak T_0:(u,\pi)\mapsto (u,-\pi).
\end{equation}
Since \(v=\pi/\hbar\), this is precisely the real form of complex
conjugation,
\begin{equation}
  \psi=u+iv\mapsto u-iv=\psi^* .
\end{equation}
Thus
\begin{equation}
  \mathfrak T_0^2=+\mathbb 1 .
\end{equation}
This restates Callender's explanatory point in canonical language. In the pure Chen formulation, complex conjugation accompanies time-reversal because temporal reflection reverses the second Cauchy datum: the reconstructed imaginary part \(v=\hbar\hat H^{-1}\dot u\) changes sign when \(\dot u\) changes sign. The real phase-space formulation sharpens this observation. What is reversed is not merely a velocity-like auxiliary datum, but the canonical partner \(\pi\). Thus, complex conjugation is the complex notation for the real phase-space reflection \((u,\pi)\mapsto(u,-\pi)\) (Fig.~\ref{fig:phase_plane}(b)). 

The spinless case supplies the \(+\mathbb 1\) baseline: time-reversal reverses the orientation of motion in the real phase-space plane, but it does not, by itself, imply a degenerate partner. The latter requires the additional spin-\(\frac12\) structure discussed next.
 
\subsection{Spin-\(\frac12\) time-reversal}
For a single spin-$\tfrac12$ degree of freedom, the real state at a fixed position is
\begin{equation}
  \Phi:=
  \begin{pmatrix}
    u_1\\[2pt]
    u_2\\[2pt]
    \pi_1/\hbar\\[2pt]
    \pi_2/\hbar
  \end{pmatrix}
  \in\mathbb R^4.
\end{equation}
time-reversal must preserve the norm, reverse the sign of the canonical motion, and invert the three spin components defined in
Eq.~\eqref{eq:spin_expectations_real}. The required map is
\begin{equation}
  \mathfrak T_{1/2}:
  \begin{pmatrix}
    u_1\\[2pt]
    u_2\\[2pt]
    \pi_1/\hbar\\[2pt]
    \pi_2/\hbar
  \end{pmatrix}
  \mapsto
  \begin{pmatrix}
    -u_2\\[2pt]
    \phantom{-}u_1\\[2pt]
    \phantom{-}\pi_2/\hbar\\[2pt]
    -\pi_1/\hbar
  \end{pmatrix}.
  \label{eq:tau_real_TR}
\end{equation}
Equivalently,
\begin{equation}
\begin{gathered}
\mathfrak T_{1/2}\Phi=\tau\Phi,
\\[2pt]
\tau:=
\begin{pmatrix}
Y & 0\\[2pt]
0 & -Y
\end{pmatrix},
\qquad
Y:=
\begin{pmatrix}
0 & -1\\[2pt]
1 & \phantom{-}0
\end{pmatrix}.
\end{gathered}
\label{eq:spin_half_time_reversal_real}
\end{equation}
where $Y$ is the antisymmetric channel map already encountered in the real spin algebra. It is fixed entirely by the symplectic geometry of the two-dimensional channel plane: $Y$ rotates that plane by a quarter-turn while preserving its oriented area. time-reversal for spin-$\tfrac12$ combines this quarter-turn in the coordinate channels with the opposite quarter-turn in the conjugate channels.

A direct substitution into Eq.~\eqref{eq:spin_expectations_real} gives
\begin{equation}
  \mathbf s(\mathfrak T_{1/2}\Phi)=-\mathbf s(\Phi).
\end{equation}
Moreover,
\begin{equation}
  \tau^T=-\tau,
  \qquad
  \tau^T\tau=I_4,
  \qquad
  \tau^2=-I_4.
  \label{eq:tau_properties}
\end{equation}
The first two relations state that $\tau$ is orthogonal and antisymmetric; the last shows that applying time-reversal twice returns the negative of the state. Since all probabilities and expectation values are quadratic in $\Phi$, this overall sign has no observable effect, and Eq.~\eqref{eq:tau_properties} reproduces the familiar half-integer-spin rule.

In terms of the reduced single-field data, if $u=(u_1,u_2)^T$ and $\pi=(\pi_1,\pi_2)^T$, then we have
\begin{equation}
  \mathfrak T_{1/2}:(u,\dot u) \mapsto  (Y u,-Y\dot u),
\end{equation}
or, equivalently,
\begin{equation}
  \mathfrak T_{1/2}:(u,\pi) \mapsto  (Y u,-Y\pi).
  \label{eq:T_spinhalf_reduced}
\end{equation}

\subsection{Kramers degeneracy in real form}
Let $\mathcal H$ denote the real Hamiltonian operator acting on the four-component spin phase-space. Time-reversal invariance means that
\begin{equation}
  [\tau,\mathcal H]=0.
  \label{eq:tau_H_commute}
\end{equation}
Suppose now that $\Phi$ is an eigenstate of $\mathcal H$ with eigenvalue $E$:
\begin{equation}
  \mathcal H\Phi=E\Phi.
\end{equation}
Then, by Eq.~\eqref{eq:tau_H_commute},
\begin{equation}
  \mathcal H(\tau\Phi)=\tau\,\mathcal H\Phi=E\,\tau\Phi,
\end{equation}
so $\tau\Phi$ is another eigenstate with the same energy. Moreover,
\begin{equation}
  \Phi^T(\tau\Phi)=0,
  \label{eq:orthogonality_kramers}
\end{equation}
because $\tau^T=-\tau$ implies
\begin{equation}
    \Phi^T\tau\Phi
=(\Phi^T\tau\Phi)^T
=\Phi^T\tau^T\Phi
=-\Phi^T\tau\Phi.
\end{equation}
Hence, every non-zero eigenstate has a real-orthogonal partner of the same energy.

Real orthogonality alone, however, does not yet establish a genuine degeneracy, and the argument must be completed. In the real description, every state $\Phi$ is accompanied by $J_4\Phi$, with $J_4$ being the complex structure of Eq.~\eqref{eq:J4_def}; $J_4\Phi$ is likewise real-orthogonal to $\Phi$ (as $J_4^T=-J_4$) and yet represents the \emph{same} physical state, the complex ray of $\Phi$ being the real two-plane $\operatorname{span}\{\Phi,J_4\Phi\}$. Moreover, any $\mathcal H$ that descends from a complex Hamiltonian commutes with $J_4$, so real eigenspaces automatically have even real dimension -- a fact with no physical content. What completes the argument is the property that distinguishes $\tau$ from the rotation generators $K_i$ of Sec.~\ref{subsec:spin_rotations_real}: $\tau$ \emph{anticommutes} with the complex structure,
\begin{equation}
  \tau J_4=-J_4\tau ,
  \label{eq:tau_antilinear}
\end{equation}
as one verifies directly from Eqs.~\eqref{eq:spin_half_time_reversal_real} and \eqref{eq:J4_def}. This is the real form of antiunitarity. Two short arguments then close the gap. First, suppose $\tau\Phi$ lies in the ray of $\Phi$, i.e., $\tau\Phi=(\alpha+\beta J_4)\Phi$ with $\alpha,\beta\in\mathbb R$. Applying $\tau$ once more and using Eq.~\eqref{eq:tau_antilinear} gives $\tau^2\Phi=(\alpha^2+\beta^2)\Phi$, contradicting $\tau^2=-I_4$; hence $\tau\Phi$ is a genuinely distinct state, not a disguised phase rotation of $\Phi$. Second, $J_4\tau$ is antisymmetric by Eq.~\eqref{eq:tau_antilinear} together with $J_4^T=-J_4$ and $\tau^T=-\tau$, so $\Phi^T(J_4\tau\Phi)=0$; combining this with $\Phi^TJ_4\Phi=0$, with Eq.~\eqref{eq:orthogonality_kramers}, and with the orthogonality of $\tau$ and $J_4$, one cheques that the four vectors
\begin{equation}
  \Phi,\qquad J_4\Phi,\qquad \tau\Phi,\qquad J_4\tau\Phi
  \label{eq:kramers_quartet}
\end{equation}
are mutually orthogonal eigenvectors of $\mathcal H$ with the same eigenvalue. Every level of a time-reversal-invariant $\mathcal H$ therefore has a real multiplicity of at least four, i.e., a complex multiplicity of at least two. This, and not the even real dimension by itself, is Kramers degeneracy in the real formulation. We note in passing that the pair $\{J_4,\tau\}$, with $J_4^2=\tau^2=-I_4$ and $J_4\tau=-\tau J_4$, equips the real space with a quaternionic structure; its appearance for time-reversal-invariant half-integer spin is the real-variable form of Dyson's threefold way \cite{dyson1962threefold}.

Equivalently, in the reduced phase-space, if $(u,\pi)$ is a stationary mode, then
\begin{equation}
  (u,\pi)
  \qquad\text{and}\qquad
  (Y u,-Y\pi)
  \label{eq:kramers_pair_real}
\end{equation}
form a Kramers pair. In particular, no non-zero state can be invariant under $\mathfrak T_{1/2}$, because $\tau^2=-I_4$ excludes real eigenvectors of $\tau$ with eigenvalue $+1$.

The comparison with the spinless case is instructive. For spinless particles, time-reversal is an involution on the reduced data and merely reverses the direction of motion in the \((u,\pi)\) plane. For spin-\(\frac12\), the corresponding real map squares to \(-1\) and therefore acts as an orthogonal quarter-turn on the internal four-dimensional phase-space. Kramers degeneracy is the spectral consequence of this extra structure.

This also clarifies the role of the magnetic examples discussed above. External magnetic fields break time-reversal symmetry: under time-reversal, \(\mathbf B\mapsto-\mathbf B\). Consequently, the Zeeman and Landau Hamiltonians do not commute with \(\mathfrak T_{1/2}\), and the associated magnetic splittings may be viewed as the lifting of Kramers doublets. This
is the real-formulation version of the standard statement that magnetic fields break time-reversal symmetry and remove Kramers degeneracy.

\section{Recent no-go results and other uses of real structures in quantum theory} \label{sec:nogo_rqt_comparison}

We close by distinguishing the present real-variable reconstruction from two nearby but different uses of real structures in quantum theory. The first is operational real Hilbert-space quantum theory, which is the target of recent network no-go theorems. The second is the Majorana-Dirac real formulation of relativistic spinor theory described in App. \ref{sec:majorana_dirac_comparison}. These comparisons are independent but complementary; both help delimit what the present construction does and what it does not claim.

Recent network results have argued that real Hilbert-space quantum theory can be experimentally distinguished from standard complex quantum theory once one considers several independent sources \cite{renou2021,Weilenmann2025Partial}. These predictions have also been tested experimentally. The interpretation of this result, however, has
recently been challenged by Hoffreumon and Woods \cite{HoffreumonWoods2026}. They argue that the claimed falsification of real quantum theory depends on how source independence is formulated: if independence is defined operationally rather than by a product-state condition, the no-go conclusion does not follow. This debate is ongoing and the Refs.~\cite{bruss2025,hw2025,HoffreumonWoods2026,MaioliCuradoGazeau2026} are, at the time of writing, preprints. We take no position on its outcome here; what matters for the present paper is only the structural distinction drawn below.
 
This debate is relevant here for a substantive reason. It shows that the empirical content of a real Hilbert-space theory is not fixed merely by replacing complex matrices with real ones. One must also specify which real structures are allowed to carry the role normally played by complex composition. The present construction belongs to a different category: it is not an autonomous real Hilbert-space theory with independently postulated composition and source-independence rules, but a real-variable reconstruction of ordinary Schrödinger quantum mechanics.

\subsection{Real Hilbert-space quantum theory and source independence}\label{subsec:network_assumptions}

In the sense relevant to Refs.~\cite{renou2021,Weilenmann2025Partial}, real quantum theory is not merely ordinary quantum mechanics written in real coordinates. It is a distinct operational framework in which local systems are assigned real Hilbert spaces and composite systems are formed by the standard real tensor product. The assumptions may be summarised as follows:
\begin{enumerate}
\item \textbf{Local real state spaces.} States, effects, and reversible transformations are represented on real Hilbert spaces by the corresponding real operators \cite{renou2021,stueckelberg1960real,Weilenmann2025Partial}.
\item \textbf{Standard composition.} Independent systems $A$ and $B$ are combined through $\mathcal H_A\otimes_{\mathbb R}\mathcal H_B$, with product states as the basic representatives of independent preparations. 
\item \textbf{Product-state source independence.} Independent sources are represented
by states that factorize across the corresponding tensor factors, up to the classical shared randomness admitted in the model.
\end{enumerate}
Under product-state versions of these assumptions, Renou et al. \cite{renou2021} showed that certain network correlations achievable in complex quantum theory are impossible in real quantum theory. Weilenmann, Gisin, and Sekatski \cite{Weilenmann2025Partial} sharpened this conclusion by showing that even partial source independence suffices to rule out real quantum theory experimentally.

Hoffreumon and Woods \cite{HoffreumonWoods2026} have challenged the operational status of the product-state assumption. Their point is not that the locality of measurements is suspect; locality has a direct operational meaning. The subtle issue is source independence. If independence is defined operationally, as the absence of observable cross-source correlations under local measurements, then real quantum theory contains non-product states that are nevertheless operationally independent. On this basis, they argue that every finite network correlation compatible with complex quantum theory also has a real quantum theory representation, provided that source independence is imposed operationally.  

\subsection{The present real-variable reformulation}
\label{subsec:present_real_reformulation}

The present construction starts from a different kinematic point. It is not a real Hilbert-space theory with independently postulated real state spaces and the standard real tensor-product rule. It begins with ordinary complex Schrödinger dynamics and rewrites it in real variables. There are two levels to this rewriting. 

At the lifted level, one keeps the canonical pair \((u,\pi)\), or equivalently $\psi = u + i\frac{\pi}{\hbar}$. This is just standard complex quantum mechanics in real canonical coordinates; no operational predictions are changed. At the reduced level, one eliminates the second quadrature and works with the Cauchy data \((u,\dot u)\). The missing canonical partner is then recovered by the Hamiltonian-dependent map  $\pi = \hbar^2 \hat H^{-1}\dot u$.    

Thus, the reduced single-field description is not an autonomous real Hilbert-space theory of the kind considered in the network debate  \cite{HoffreumonWoods2026,renou2021,Weilenmann2025Partial}. It is a reduced description of standard quantum mechanics, whose probability and metric structures are inherited from the lifted phase-space formulation. The difference becomes important for composites because the reduced reconstruction map is global and does not respect the local tensor-product structure assumed in the operational no-go framework.
 
\subsection{Composite systems and the non-factorizing reconstruction map} \label{subsec:nonfactorization}
The distinction becomes sharp for composite systems. Consider two non-interacting subsystems $A$ and $B$ with
\begin{equation}
  \hat H_{AB} = \hat H_A\otimes I_B + I_A\otimes \hat H_B.
  \label{eq:H_sum}
\end{equation}
In the reduced single-field variables, the canonical momentum of the composite is reconstructed by
\begin{equation}
  \pi_{AB}=\hbar^2\hat H_{AB}^{-1}\dot u_{AB}.
  \label{eq:piAB}
\end{equation}
Although $\hat H_{AB}$ is a sum of local operators, its inverse is not, in general, a sum or product of local inverses. In an eigenbasis,
\begin{equation}
  \hat H_A\ket{a}=E_a\ket{a},\qquad \hat H_B\ket{b}=F_b\ket{b},
\end{equation}
so that
\begin{equation}
  \hat H_{AB}\ket{ab}=(E_a+F_b)\ket{ab},
\end{equation}
and therefore
\begin{equation}
  \hat H_{AB}^{-1}\ket{ab}=\frac{1}{E_a+F_b}\ket{ab}.
  \label{eq:inverse_spectrum}
\end{equation}
This should be compared with, for example, $\hat H_A^{-1}\otimes I + I\otimes \hat H_B^{-1}$, which would act as $(E_a^{-1}+F_b^{-1})\ket{ab}$ and is therefore inequivalent to \eqref{eq:inverse_spectrum} in general. The reduced reconstruction map for the composite system is therefore intrinsically global, even in the absence of any interaction term.

The same point may be written without reference to a spectral basis. For positive $\hat H_A$ and $\hat H_B$,
\begin{equation}
  \hat H_{AB}^{-1} = \int_{0}^{\infty} e^{-t\hat H_{AB}}\,\dd t
  = \int_{0}^{\infty} \bigl(e^{-t\hat H_A}\otimes e^{-t\hat H_B}\bigr)\,\dd t.
  \label{eq:resolvent_integral}
\end{equation}
At each fixed $t$, the semigroup factorises, but the integral over $t$ does not reduce to a local tensor-product expression for $\hat H_{AB}^{-1}$.

\subsection{Consequences for the no-go theorems}
\label{subsec:network_consequences}
The comparison with the network debate can now be stated more precisely. If real Hilbert-space quantum theory is supplemented by product-state versions of source independence, the network results of Renou {\em et al.} and Weilenmann {\em et al.} yield separations from standard complex quantum theory. If, however, source independence is imposed only operationally, Hoffreumon and Woods argue that these separations disappear for finite network and sequential protocols compatible with complex quantum theory. In either case, the discussion concerns an autonomous real Hilbert-space theory with its own composition and independence postulates.

The reduced single-field data $(u,\dot u)$ used here does not, by itself, carry such a local tensor-product structure. The quantity that restores Born probabilities, namely $\pi=\hbar^2\hat H^{-1}\dot u$, depends on the inverse of the \emph{global} Hamiltonian. For composites, this inverse does not factorise, even when the Hamiltonian itself is non-interacting.

It is, therefore, useful to distinguish two levels within the present real-variable reformulation.
\begin{enumerate}
\item If one keeps the full canonical pair $(u,\pi)$ for each subsystem, then one has simply rewritten standard complex quantum mechanics in real coordinates, and all network predictions are exactly those of the usual theory.
\item If one passes to the reduced single-field description $(u,\dot u)$, then the price of eliminating the second quadrature is precisely the non-factorising reconstruction map 
  $\pi = \hbar^2\hat H^{-1}\dot u$. The resulting reduced variables are not the kinematic objects assumed in the operational no-go framework.
\end{enumerate}
This is the sense in which the network debate does not directly bear on the present construction. The product-state no-go results rule out a particular operational real Hilbert-space theory; the Hoffreumon-Woods analysis argues that a different operational reading of source independence restores empirical equivalence with complex quantum theory. The present formulation is neither of these. With the full canonical pair retained, it is standard complex quantum mechanics written in real phase-space coordinates. With only the reduced field retained, the missing quadrature is recovered through a Hamiltonian-dependent, generally non-local reconstruction map.

This conclusion is consistent with the observation of McKague, Mosca, and Gisin \cite{McKague2009} that simulating complex quantum theory within a real Hilbert-space description requires an additional global resource: a ``reference frame of complexness.'' It is also consistent with more recent real embeddings of complex quantum theory in which the composition rule for composite systems is modified rather than taken to be the standard real Kronecker product \cite{bruss2025,hw2025}. Recent Kähler-space realifications make the same com{\-}po{\-}sition-rule point in especially explicit form: they retain the complex structure \(J\) on the real space and replace the standard real tensor product \(\otimes_{\mathbb R}\) by a \(J\)-compatible product \(\otimes_K\), so that realification commutes with the usual complex tensor product \cite{MaioliCuradoGazeau2026}.

The lifted formulation of the present paper is no exception to this composition-rule point. The real phase space of a composite system is the realification of $\mathbb C^{N_A}\otimes_{\mathbb C}\mathbb C^{N_B}$ and has real dimension $2N_AN_B$, whereas the real tensor product of the subsystem phase spaces has dimension $4N_AN_B$. Retaining the canonical pair $(u,\pi)$ for each subsystem therefore does not, by itself, supply the composite phase space; the pairing must be taken compatibly with the complex structure $J$, which is precisely the modified product $\otimes_K$ just described. In the reduced variables, the same fact resurfaces as the non-factorising reconstruction map of Sec.~\ref{subsec:nonfactorization}.

Hoffreumon and Woods' later analysis of the network no-go debate adds another option: if source independence is understood operationally, the missing structure can be carried by non-tomographic real correlations that remain invisible to local measurements \cite{HoffreumonWoods2026}. In the present construction, by contrast, the corresponding structure reappears as the real phase-space pair and, after reduction, as the non-factorising reconstruction map.

The common lesson is that the structural role of the complex phase cannot simply be deleted. It must be relocated: into a global reference frame, into a modified composition rule, into operationally invisible real correlations, or into the symplectic phase-space structure of the real Schrödinger reconstruction. In App. \ref{sec:majorana_dirac_comparison} we compare  operational real Hilbert-space quantum theory with the Majorana formulation of relativistic spin-\(\frac12\) theory \cite{Aste2010,Majorana1937,Pal2011}.

\section{Summary and conclusion}

We have revisited the old question of whether the imaginary unit is an essential ingredient of quantum mechanics. Schrödinger's early real-valued equation and Chen's later generalization show that the explicit complex form of the Schrödinger equation is not unavoidable: one can eliminate the imaginary part and obtain a second-order equation for a single real field. This observation is mathematically correct, but it is not yet the full story.

The reduced single-field equation carries the same dynamics, only at a price. The missing imaginary part has to be reconstructed through \(\hat H^{-1}\), and this reconstruction makes several familiar quantum structures opaque. Probability is globally conserved, but its local current is no longer manifest in the reduced variables. Momentum expectation values can be computed, but momentum is not represented by an operator acting on the single real field alone. Likewise, the uncertainty relation remains valid, but its usual kinematical proof has to appeal to structure no longer locally present in the reduced description. These difficulties do not amount to a new physical prediction; they show instead that the one-field variables are an awkward representation of the quantum state.

The constructive part of the paper identified what has been hidden. The second-order real equation has a natural Hamiltonian lift, and the canonical partner of \(u\) is
\begin{equation}
  \pi=\hbar^2\hat H^{-1}\dot u .
\end{equation}
Once the pair \((u,\pi)\) is used, the theory remains entirely real-valued while recovering the local Born density, the probability current, and the generator interpretation of momentum. More strongly, the lifted real phase-space formulation is strictly equivalent to the standard complex Schrödinger equation under the identification
\begin{equation}
  \psi=u+i\pi/\hbar.
\end{equation}
This is not merely the familiar split of an already-given complex wave function into real and imaginary parts. The route taken here starts with the reduced second-order equation and recovers the missing entry as the canonical momentum fixed by the Hamiltonian lift.

This equivalence should not be read as the rediscovery of a real symplectic formulation of quantum mechanics. That formulation is already part of the geometric understanding of the theory. Its role here is more specific. It provides the reference structure against which the reduced Schrödinger-Chen equation can be assessed. The point of the Hamiltonian lift is to show how this known phase-space structure is recovered when one starts from the one-field equation rather than from the full complex Hilbert space.

The later sections were included in the same diagnostic spirit. They do not ask again whether the lifted real phase-space formulation is equivalent to the complex theory. Instead, they use standard quantum structures as probes of the reduction itself. They ask where probability, generators, uncertainty, gauge covariance, spin, angular momentum, and time-reversal are stored once the symbol \(i\) has been removed. 

A uniform pattern emerged. At the level of the reduced equation, the evolution can be represented without displaying the full phase plane. At the level of local quantum structure, however, the omitted phase plane reappears. The Born density is its Euclidean norm; the probability current and momentum use its antisymmetric pairing; the uncertainty relation draws on its local symplectic geometry; magnetic coupling gauges its internal \(SO(2)\) rotations; spin and orbital angular momentum are encoded in real invariant planes; and time-reversal acts as a reflection or rotation of those planes. The examples, therefore, do not test whether the real phase-space formulation is equivalent to the complex theory. They show, case by case, which task the complex notation was performing. 

The same lesson clarifies the scope of the result and its relation to recent no-go theorems for operational real quantum theory. Those results concern real Hilbert-space theories with the standard real tensor-product rule. The construction studied here is different: it is either the usual theory written in real canonical variables or a reduced one-field description, whose missing quadrature is recovered by a Hamiltonian-dependent, generally non-local map.
 
The outcome is therefore neither a victory for a naive one-real-field ontology nor a defeat for real variables. What fails, for the parametrizations considered here, is the idea that a single real configuration field, together with its ordinary time derivative, can support a Born rule that is simultaneously local and independent of the dynamics: Chen's physical real part keeps the field kinematically autonomous at the price of a non-local rule, while Pauli's auxiliary field keeps the rule local at the price of entangling it with the Hamiltonian (Sec.~\ref{ssec:pauliagain}). What succeeds is the more precise claim: complex notation can be traded for real variables, provided the underlying phase-plane structure is retained. In this sense, the imaginary unit is not best understood as a mysterious extra ingredient. It is a compact notation for a real geometric structure. The symbol \(i\) may be removed from the equations, but the symplectic and complex structure for which it stands remains part of the theory.

\appendix
\section*{Appendices}

The appendices collect technical assumptions and worked-out examples that support the main argument without being needed for its basic line. Appendix~\ref{app:hinv} states the functional-analytic conventions used for inverse Hamiltonians, zero modes, admissible Cauchy data, and sign branches.  Appendices~\ref{sec:ho_ladder_coherent_real}--\ref{app:zeeman}  provide more detailed calculations for familiar quantum systems and symmetries: oscillator coherent packets, Landau levels, sequential spin measurements, and orbital angular momentum with Zeeman splitting. These examples illustrate how standard complex structures reappear in the real phase-space formulation, but they are not required for the central reconstruction developed in the main text. In addition, we compare withe the real approaches by Bohm and Majorana.

\section{Standing assumptions on $\hat H^{-1}$ and admissible Cauchy data}
\label{app:hinv}

The reconstruction of the missing quadrature uses the inverse Hamiltonian, $\hat H^{-1}$ in the electrostatic case and $\hat H_R^{-1}$ in the magnetic case. For ease of reference, we collect here the standing assumptions under which the formulas of the main text are to be understood.

\emph{(i) Self-adjointness.} $\hat H$ is assumed to be self-adjoint on a dense domain of $L^2$ and real symmetric in the position representation. In the magnetic case, the same is assumed for $\hat H_R$, while $\hat H_I$ is real skew-adjoint (Sec.~\ref{sec:mag}); operator products such as $\hat H_R\hat H_I\hat H_R^{-1}$ in Eqs.~\eqref{eq:magnetic_gamma} and \eqref{eq:magnetic_omega} are understood on their natural domains. Self-adjointness of the lift operators also enters the uniqueness discussion of Sec.~\ref{sec:realH}.

\emph{(ii) Zero modes.} If $\hat H$ has a non-trivial kernel, $\hat H^{-1}$ is defined on the orthogonal complement of $\ker\hat H$, or equivalently, it is replaced by the Moore-Penrose pseudo-inverse \cite{benisrael2003generalized}; the Zero mode sector must then be specified separately (see the footnote in Sec.~\ref{sec:alt}). This convention also excludes the spurious solution $u_0(t)=a+bt$ with $\hat Hu_0=0$ of the reduced second-order equation.

\emph{(iii) Zero in the continuous spectrum; admissible data.} For the free particle, $0$ lies in the continuous spectrum, so $\hat H^{-1}$ is unbounded and only densely defined. The reconstruction $\pi=\hbar^2\hat H^{-1}\dot u$ then requires $\dot u\in\operatorname{Ran}\hat H$. This is not an additional hypothesis but the statement that the Cauchy data descend from a Schr\"odinger solution: by Eq.~\eqref{equ:coup1}, $\dot u=\hat Hv/\hbar$ with normalisable $v$. Data with $\dot u\in\operatorname{Ran}\hat H$ (and $\dot u\in D(\hat H^{-1})$ where required) are called \emph{admissible}. For the free particle, admissibility makes the source of the Poisson equation a total divergence, so that its monopole moment vanishes identically, $\int\dot u\,\dd^3r=0$ (Sec.~\ref{sec:non-loc-rev}).

\emph{(iv) Indefinite spectra; sign branch.} If the spectrum of $\hat H$ contains both signs, the norm criterion of Sec.~\ref{sec:realH} fixes the lift only up to a sign choice on each spectral subspace, as discussed after Eq.~\eqref{eq:KsqHsq}; the physical branch is selected by matching the first-order dynamics \eqref{equ:coup12} -- equivalently, by the spectral sign of $\hat H$ itself, not merely of $\hat H^2$. The mode-frequency bookkeeping $\omega_{n\ell}=|E_{n\ell}|/\hbar$ used in Sec.~\ref{sec:zeeman_real} implements the same selection.

\emph{(v) Time dependence.} For time-dependent potentials, the canonical equations \eqref{eq:canonical_equations} hold with $\hat H(t)$; eliminating $\pi$ reproduces Chen's Eq.~\eqref{equ:chenfull}, including the $\dot V$ term, and the reconstruction map becomes time-dependent through $\hat H(t)^{-1}$.

\emph{(vi) Reading of the main statements.} The propositions and operator identities of the main text are rigorous, without further qualification, for finite-dimensional real Hilbert spaces and, more generally, for Hamiltonians with purely discrete spectra, on the invariant core spanned by finitely many eigenfunctions; there, all operator products appearing above are defined, and the manipulations are justified term by term. For Schrödinger operators with continuous spectra, the statements are to be read on the natural domains: the flow \eqref{eq:canonical_equations} on $D(\hat H)$, the Born norm \eqref{eq:P} for $u,\pi\in L^2$, the Hamiltonian functional \eqref{eq:H_phase_space} on the form domain $u,\pi\in D(|\hat H|^{1/2})$, and the reconstruction on the admissible data of item (iii); operator products such as $\hat H_R\hat H_I\hat H_R^{-1}$ are understood as in item (i).

\section{Ladder operators and coherent packets for the one-dimensional harmonic oscillator}
\label{sec:ho_ladder_coherent_real}

The equivalence established in Sec.~\ref{ssec:equivalence} guarantees that the harmonic oscillator can be described in the real phase-space variables. The point of this example is therefore not to test the equivalence but to show how one of the most familiar uses of complex notation appears in terms of the real-valued language. For the oscillator, the spatial ladder algebra remains essentially unchanged: the Hermite-functions are real, and the usual raising and lowering operators are real differential operators. What changes is the description of time-dependent states. In particular, coherent packets are not introduced here as eigenstates of a complex annihilation operator but as phase-locked orbits of real normal modes. This makes explicit which part of the usual complex formalism is carried by the second Cauchy datum, or equivalently by the canonical momentum. We consider the standard one-dimensional oscillator,
\begin{equation}
  \hat H  =  -\frac{\hbar^2}{2M}\partial_x^2  +\frac12 M\omega^2 x^2,  \qquad
  \ell:=\sqrt{\frac{\hbar}{M\omega}}.
  \label{eq:HO_Hamiltonian}
\end{equation}

\subsection{Spatial ladder algebra}
\label{subsec:HO_ladder}

The spatial spectral theory is the textbook one and is unchanged by the real formulation. The normalised Hermite-functions
\begin{equation}
  \varphi_n(x)
  =
  \frac{1}{\pi^{1/4}\sqrt{2^n n!\,\ell}}\,
  H_n\!\left(\frac{x}{\ell}\right)
  \exp\!\left(-\frac{x^2}{2\ell^2}\right)
  \label{eq:HO_hermite_functions}
\end{equation}
are real eigenfunctions, $\hat H\varphi_n=E_n\varphi_n$ with $E_n=\hbar\omega\left(n+\tfrac12\right)$, and the ladder operators $\hat a=\tfrac{1}{\sqrt2}\bigl(x/\ell+\ell\,\partial_x\bigr)$ and $\hat a^\dagger=\tfrac{1}{\sqrt2}\bigl(x/\ell-\ell\,\partial_x\bigr)$ -- with $[\hat a,\hat a^\dagger]=1$, $\hat H=\hbar\omega\bigl(\hat a^\dagger\hat a+\tfrac12\bigr)$, $\hat a\,\varphi_n=\sqrt n\,\varphi_{n-1}$, and $\hat a^\dagger\varphi_n=\sqrt{n+1}\,\varphi_{n+1}$ -- are real differential operators. Thus, the spatial spectral theory carries over unchanged; what changes is the kinematics of the time-dependent state.

\subsection{Real mode dynamics}

Expand the real field in the Hermite basis,
\begin{equation}
  u(x,t)=\sum_{n=0}^{\infty}q_n(t)\,\varphi_n(x).
  \label{eq:HO_q_expansion}
\end{equation}
The equation $\ddot u + \hbar^{-2}\hat H^2 u = 0$ then reduces, mode-by-mode, to
\begin{equation}
  \ddot q_n+\omega^2\left(n+\frac12\right)^2 q_n=0.
  \label{eq:HO_qn_oscillator}
\end{equation}
Hence
\begin{equation}
  q_n(t)=A_n\cos\!\left[\omega\left(n+\frac12\right)t\right]
       +B_n\sin\!\left[\omega\left(n+\frac12\right)t\right].
  \label{eq:HO_qn_solution}
\end{equation}
The coefficients $(A_n,B_n)$ are simply the mode-representation of the second-order initial data $(u,\dot u)$.

The compatible canonical momentum may be expanded as
\begin{equation}
  \pi(x,t)=\hbar\sum_{n=0}^{\infty}p_n(t)\,\varphi_n(x),
  \label{eq:HO_pi_expansion}
\end{equation}
with
\begin{equation}
  p_n(t)=B_n\cos\!\left[\omega\left(n+\frac12\right)t\right]
       -A_n\sin\!\left[\omega\left(n+\frac12\right)t\right].
  \label{eq:HO_pn_solution}
\end{equation}
Therefore
\begin{equation}
  \int_{\mathbb R}\left(u^2+\frac{\pi^2}{\hbar^2}\right) \dd x
  =\sum_{n=0}^{\infty}\left(q_n^2+p_n^2\right),
  \label{eq:HO_Born_norm_modes}
\end{equation}
so each oscillator level contributes the usual Euclidean norm in its own internal plane.
\subsection{Coherent packets as phase-locked real mode orbits}
\label{subsec:HO_coherent}

Among all solutions of Eq.~\eqref{eq:HO_qn_oscillator}, one family is distinguished by two simple physical properties: its probability density remains a Gaussian of fixed width, and its centre follows the classical oscillator orbit. In the real phase-space formulation, it is natural to define this family as the coherent packet.

Introduce three real parameters: $\lambda\ge 0$, an orbital angle $\vartheta$, and an internal phase $\theta_0$.
The coherent packet is the phase-locked mode superposition
\begin{equation}
\begin{aligned}
u(x,t)
&=
e^{-\lambda/2}
\sum_{n=0}^{\infty}
\frac{\lambda^{n/2}}{\sqrt{n!}}\,
\varphi_n(x)
\\
&\quad\times
\cos\left[
\omega\left(n+\frac12\right)t
-n\vartheta-\theta_0
\right].
\end{aligned}
  \label{eq:coherent_expansion_u}
\end{equation}
The corresponding canonical momentum is
\begin{equation}
\begin{aligned}
\frac{\pi(x,t)}{\hbar}
&=
-e^{-\lambda/2}
\sum_{n=0}^{\infty}
\frac{\lambda^{n/2}}{\sqrt{n!}}\,
\varphi_n(x)
\\
&\quad\times
\sin\left[
\omega\left(n+\frac12\right)t
-n\vartheta-\theta_0
\right].
\end{aligned}
  \label{eq:coherent_expansion_w}
\end{equation}
These formulas are already sufficient to display the real content of the construction: a coherent packet is a Poisson-weighted ladder of oscillator levels whose phases are locked linearly in $n$.

It is convenient to trade $(\lambda,\vartheta)$ for the initial packet centre $x_0$ and initial classical momentum $p_0$.
\begin{equation}
  x_0=\sqrt{2\lambda}\,\ell\cos\vartheta,
  \qquad
  p_0=\sqrt{2\lambda}\,\frac{\hbar}{\ell}\sin\vartheta.
  \label{eq:coherent_x0_p0}
\end{equation}
A standard Hermite generating-function resummation then gives the closed form
\begin{equation}
  u(x,t)  = \frac{1}{\pi^{\frac14}\sqrt\ell}\,
  \exp\!\left(-\frac{(x-x_c(t))^2}{2\ell^2}\right)
  \cos\Theta(x,t),
  \label{eq:coherent_u_closed}
\end{equation}
\begin{equation}
  \frac{\pi(x,t)}{\hbar}
  =
  \frac{1}{\pi^{\frac14}\sqrt\ell}\,
  \exp\!\left(-\frac{(x-x_c(t))^2}{2\ell^2}\right)
  \sin\Theta(x,t),
  \label{eq:coherent_w_closed}
\end{equation}
where the packet centre and classical momentum are
\begin{eqnarray}
  x_c(t)&=&x_0\cos\omega t+\frac{p_0}{M\omega}\sin\omega t,
  \nonumber\\
  p_c(t)&=&p_0\cos\omega t-M\omega x_0\sin\omega t,
  \label{eq:classical_orbit_coherent}
\end{eqnarray}
and the phase is
\begin{equation}
  \Theta(x,t)=\frac{p_c(t)x}{\hbar}-\frac{x_c(t)p_c(t)}{2\hbar}-\frac{\omega t}{2}+\theta_0.
  \label{eq:coherent_phase}
\end{equation}
The density is therefore
\begin{eqnarray}
  \rho(x,t)&=&u(x,t)^2+\frac{\pi(x,t)^2}{\hbar^2}\nonumber\\
  &=&\frac{1}{\sqrt\pi\,\ell}\exp\!\left(-\frac{(x-x_c(t))^2}{\ell^2}\right).
  \label{eq:coherent_density_real}
\end{eqnarray}
It retains the ground-state width for all times and simply follows the classical orbit \eqref{eq:classical_orbit_coherent}.

Read as a genuine second-order solution, the coherent packet is fixed by its Cauchy data $u(x,0)$ and $\dot u(x,0)$; setting $t=0$ in Eqs.~\eqref{eq:coherent_expansion_u} and \eqref{eq:coherent_expansion_w} gives their mode content, and in closed form
\begin{equation}
  u(x,0)
  =
  \frac{e^{\left(-\frac{(x-x_0)^2}{2\ell^2}\right)}}{\pi^{\frac14}\sqrt\ell}\,  
  \cos\!\left(\frac{p_0 x}{\hbar}-\frac{x_0p_0}{2\hbar}+\theta_0\right).
  \label{eq:coherent_initial_u_closed}
\end{equation}
The second datum $\dot u(x,0)$ is not an accessory. It fixes the internal orientation of each normal mode and thereby determines whether the packet moves to the right, to the left, or merely oscillates in place.

The same calculation also shows that the packet remains a minimum-uncertainty state. Because the density  \eqref{eq:coherent_density_real} is Gaussian with a fixed width, one has
\begin{equation}
  \Delta x=\frac{\ell}{\sqrt2},
  \qquad
  \Delta p=\frac{\hbar}{\sqrt2\,\ell},
  \qquad
  \Delta x\,\Delta p=\frac{\hbar}{2}.
  \label{eq:coherent_min_uncertainty}
\end{equation}
Thus, the coherent packet is, in the real language, a non-dispersing Gaussian orbit that saturates the Heisenberg bound.

The lesson of this appendix is not that the oscillator provides a new proof of the real-complex equivalence. Rather, it shows how the usual coherent state is redistributed in the real phase-space description. The ladder operators still organize the spatial spectrum, but the coherent packet itself is represented as a phase-locked orbit of real normal modes. Its motion is fixed by the second Cauchy datum or, equivalently, by the canonical momentum. In this way, the familiar complex coherent-state notation is replaced by an explicitly real account of the same phase-space geometry.

\section{Landau levels in real variables} \label{sec:landau_real}
The uniform-field problem is the simplest spectral application of the magnetic reconstruction developed in Sec.~\ref{sec:mag}. It illustrates how the reduced single-field description and the real phase-space formulation complement one another. The spectrum can be obtained from real separated modes, while the reconstructed companion quadrature makes the invariant two-planes, gauge rotations, and degeneracies transparent. For Landau levels, this is especially useful because each separated mode naturally selects a two-dimensional oscillator plane.

We consider a particle of mass $M$ and charge $q$ moving in the $(x,y)$-plane in a
uniform magnetic field $\mathbf B = B\,\hat{\mathbf z}$. It is convenient to introduce
\begin{eqnarray}
  \omega_c &:= &  \frac{|q|B}{M}>0,\qquad
  s:=\operatorname{sgn}(qB)\in\{\pm 1\},\nonumber \\
  \ell_B& := & \sqrt{\frac{\hbar}{|q|B}}.
  \label{eq:landau_notation}
\end{eqnarray}
The Hamiltonian is
\begin{equation}
  \hat H = \frac{1}{2M}\bigl(-i\hbar\nabla-q\mathbf A\bigr)^2.
  \label{eq:landau_H}
\end{equation}
In a divergence-free gauge ($\nabla\cdot\mathbf A=0$), the split
$\hat H=\hat H_R+i\hat H_I$ takes the form
\begin{equation}
  \hat H_R = -\frac{\hbar^2}{2M}\nabla^2 + \frac{q^2A^2}{2M},
  \quad
  \hat H_I = \frac{\hbar q}{M}\,\mathbf A\cdot\nabla,
  \label{eq:landau_HR_HI}
\end{equation}
so that the coupled real equations are
\begin{equation}
  \hbar \dot u = \hat H_R v + \hat H_I u,\qquad
  -\hbar \dot v = \hat H_R u - \hat H_I v.
  \label{eq:landau_uv}
\end{equation}
This is the system already discussed briefly in Sec.~\ref{ssec:mag}. We use Eq.~\eqref{eq:landau_uv} only as a local bookkeeping device. The actual reduced description remains the single real field $u$ together with its second datum $\dot u$, with the missing quadrature reconstructed from Eq.~\eqref{eq:v_mag} after the mode problem has been solved.

The explicit mode algebra in the two standard gauges is textbook material and is therefore deferred to Appendix~\ref{app:landau}. There, real separated modes of Eq.~\eqref{eq:landau_uv} yield the Landau spectrum $E_n=\hbar\omega_c\left(n+\tfrac12\right)$ with its guiding-center degeneracy in the Landau gauge, Eq.~\eqref{eq:landau_levels_landau_gauge}, and the counter-rotating branches of Eq.~\eqref{eq:circular_energies_sigma}, organized by the level index \eqref{eq:N_index}, in the circular (symmetric) gauge. The remainder of this section discusses what these results mean in the reduced single-field description and in the real phase-space formulation.

\subsection{Interpretation in the real phase-space formulation}
\label{subsec:landau_interpretation_real}
The two gauges highlight complementary aspects of the real description.

In the Landau gauge, a mode with fixed \(k\) behaves as a translated harmonic oscillator in \(x\). The real field \(u\) carries the cosine component, while its second datum \(\dot u\) identifies the sine companion in the same \(k\)-plane. The guiding-center label \(k\) accounts for the macroscopic degeneracy.

In the circular gauge, for every \(m>0\), the pair \(\cos(m\phi)\) and \(\sin(m\phi)\) spans a real invariant plane. The magnetic term acts gyroscopically on this plane and splits the two possible rotation senses. The Landau-level degeneracy is encoded here in the infinitely many triples \((n_r,m,\sigma)\) yielding the same \(N\) (in complex notation, $m_\ell=\sigma m$; see Appendix~\ref{app:landau}).

This use of two real components should not be mistaken for a return to a two-field postulate. The reduced equation remains a second-order equation for the single real field \(u\). To specify a solution one must, of course, give the Cauchy data \((u,\dot u)\). In separated modes, however, it is often
clearer to display the reconstructed companion \(v\), because the invariant planes and gauge rotations are then transparent. The field \(v\) is not introduced as an independent physical field; it is obtained from the reduced Cauchy data by the magnetic-field reconstruction map. Thus the use of \((u,v)\) in this section is a coordinate choice for displaying the reconstructed real dynamics, not a return to the complex wave-function formalism.

The Landau and circular gauges are related by
\begin{equation}
  \mathbf A_L=\mathbf A_C+\nabla\chi,
  \qquad
  \chi(x,y)=\frac{Bxy}{2}.
  \label{eq:landau_circular_gauge_function}
\end{equation}
In the explicit two-quadrature representation, the corresponding gauge transformation acts locally as an \(SO(2)\) rotation,
\begin{equation}
  \begin{pmatrix}
    u'\\ v'
  \end{pmatrix}
  =
  \begin{pmatrix}
    \cos(q\chi/\hbar) & -\sin(q\chi/\hbar)\\
    \sin(q\chi/\hbar) & \cos(q\chi/\hbar)
  \end{pmatrix}
  \begin{pmatrix}
    u\\ v
  \end{pmatrix}.
   \label{eq:gauge_rotation_real}  
\end{equation}
In the reduced variables \((u,\dot u)\), this gauge covariance is hidden because the second quadrature has already been eliminated. The spectra obtained in the Landau and circular gauges (Eqs.~\eqref{eq:landau_levels_landau_gauge} and
\eqref{eq:landau_levels_circular_gauge}) are, therefore, as expected, the same; the two gauges merely display the same real phase-space dynamics in different coordinates.

\subsection{Lowest Landau level and the preferred rotation sense in the $\cos/\sin$ planes} \label{subsec:lll_real}

The lowest Landau level (LLL) is especially transparent in the real phase-space formulation because it most clearly exhibits how the magnetic field selects one preferred sense of rotation in each real angular plane.

In the Landau gauge, the LLL is especially simple. From Eq.~\eqref{eq:landau_levels_landau_gauge}, the LLL in the Landau gauge is simply the oscillator ground state $n=0$ with an arbitrary guiding-center label $k$:
\begin{equation}
  E_0=\frac{\hbar\omega_c}{2}.
\end{equation}
A convenient real representative is
\begin{equation}
\begin{aligned}
u_{0,k}(x,y,t)
&=
\frac{1}{\pi^{1/4}\sqrt{\ell_B}}\,
\exp\!\left[-\frac{(x-x_k)^2}{2\ell_B^2}\right]
\\
&\qquad\times
\cos\!\left(ky-\frac{E_0t}{\hbar}\right),
\end{aligned}
\label{eq:LLL_landau_u}
\end{equation}
with $x_k=\hbar k/(qB)$. Its reconstructed companion is
\begin{equation}
\begin{aligned}
v_{0,k}(x,y,t)
&=
\frac{1}{\pi^{1/4}\sqrt{\ell_B}}\,
\exp\!\left[-\frac{(x-x_k)^2}{2\ell_B^2}\right]
\\
&\qquad\times
\sin\!\left(ky-\frac{E_0t}{\hbar}\right),
\end{aligned}
\label{eq:LLL_landau_v}
\end{equation}
Thus, the LLL is the harmonic-oscillator ground branch, together with the continuous guiding-center degeneracy parameterized by $k$.

In the circular gauge, the same degeneracy is reorganized into angular-momentum sectors.  For each $m>0$, the real functions
\begin{equation}
    C_m(\phi):=\cos(m\phi),\qquad S_m(\phi):=\sin(m\phi)
\end{equation}
span a two-dimensional invariant plane, and the azimuthal generator acts on this plane
as
\begin{equation}
  \partial_\phi
  \begin{pmatrix}
    C_m \\ S_m
  \end{pmatrix}
  = mJ
  \begin{pmatrix}
    C_m \\ S_m
  \end{pmatrix},
  \qquad
  J:=\begin{pmatrix}0&-1\\[2pt]1&0\end{pmatrix}.
  \label{eq:phi_rotation_plane}
\end{equation}
Hence, a general real mode in this plane may be written as
\begin{equation}
  u(r,\phi,t)=f(r)\Big(a(t)\,C_m(\phi)+b(t)\,S_m(\phi)\Big),
  \label{eq:general_m_plane}
\end{equation}
with coefficient vector
\begin{equation}
    U(t):=\begin{pmatrix}a(t)\\ b(t)\end{pmatrix}.
\end{equation}
The magnetic term acts gyroscopically on $U(t)$ through the matrix $mJ$.

As discussed around Eq.~\eqref{eq:circular_energies_sigma}, each
\(m\)-plane has two counter-rotating normal modes, labeled by
\(\sigma=\pm1\) (cf.\ Fig.~\ref{fig:zeeman_mplane}). With \(s=\operatorname{sgn}(qB)\), their energies are
given by Eq.~\eqref{eq:circular_energies_sigma}. The LLL is obtained by
\(n_r=0\) and \(\sigma=s\), for which the energy collapses to
%
\begin{equation}
  E_0=\frac{\hbar\omega_c}{2},
  \label{eq:lll_energy}
\end{equation}
independent of $m$. Thus, for every $m=0,1,2,\dots$, precisely one rotation sense in the real $\bigl(C_m,S_m\bigr)$ plane belongs to the LLL; the opposite sense lies at a higher energy.

For $m\ge 0$, the normalized radial LLL profile is obtained from Eq.~\eqref{eq:fock_darwin_radial} with $n_r=0$:
\begin{equation}
  f_m^{\mathrm{LLL}}(r)
  = \frac{1}{\sqrt{2\pi\,2^m m!}\,\ell_B}
    \left(\frac{r}{\ell_B}\right)^m
    \exp\!\left(-\frac{r^2}{4\ell_B^2}\right).
  \label{eq:lll_radial}
\end{equation}
A convenient real representative of the LLL branch is then
\begin{align}
  u_m^{\mathrm{LLL}}(r,\phi,t)
  &= f_m^{\mathrm{LLL}}(r)\,
     \cos\!\left(s\,m\phi-\frac{E_0 t}{\hbar}\right),
  \label{eq:lll_u}
  \\
  v_m^{\mathrm{LLL}}(r,\phi,t)
  &= f_m^{\mathrm{LLL}}(r)\,
     \sin\!\left(s\,m\phi-\frac{E_0 t}{\hbar}\right).
  \label{eq:lll_v}
\end{align}
For $m=0$, this reduces to the rotationally invariant Gaussian.
For $m>0$, the pair
\begin{equation}
f_m^{\mathrm{LLL}}(r)\cos(m\phi),\qquad
f_m^{\mathrm{LLL}}(r)\sin(m\phi)
\end{equation}
spans a real two-dimensional angular plane, and the magnetic field selects the branch whose sense of rotation matches the sign of $qB$.

The real interpretation of the LLL is now direct. In the usual complex notation, the LLL in symmetric gauge is characterized by keeping only one sign of the azimuthal phase factor. In the present real-variable formulation, the same statement becomes: in each real $\bigl(\cos(m\phi),\sin(m\phi)\bigr)$ plane, only one of the two counter-rotating normal modes belongs to the lowest Landau level. This is the real counterpart of the familiar holomorphic restriction.

Equivalently, the full LLL is the orthogonal direct sum of the one-dimensional $m=0$ sector and, for each $m>0$, one preferred rotation branch in the two-dimensional real $m$-plane:
\begin{equation}
  \mathcal H_{\mathrm{LLL}}^{(\mathrm{real})}
  =
  \mathcal H_{0,0}^{(\mathrm{real})}
  \oplus
  \bigoplus_{m=1}^{\infty}
  \mathcal H_{0,m,\sigma=s}^{(\mathrm{real})}.
  \label{eq:lll_direct_sum}
\end{equation}
The enormous Landau degeneracy is therefore visible in the real phase-space formulation as the fact that all these preferred-rotation modes, for all $m\ge0$, share the same energy $E_0=\hbar\omega_c/2$.

The two gauges are related by a reparametrization of the same degeneracy. The continuous guiding-center index $k$ in the Landau gauge and the discrete index $m$ in the circular gauge are merely two different ways of parameterizing the same degeneracy of the lowest Landau level. Gauge transformation reorganizes the family Eq.~\eqref{eq:LLL_landau_u}--\eqref{eq:LLL_landau_v} into the angular decomposition Eq.~\eqref{eq:lll_u}--\eqref{eq:lll_v}. In the explicit two-quadrature representation, this is just the local $SO(2)$ rotation \eqref{eq:gauge_rotation_real}; in the reduced single-component description, the same statement is hidden because the companion quadrature has already been eliminated.

\section{Landau-level mode algebra in the two gauges}
\label{app:landau}

This appendix records the explicit real separated-mode algebra behind the results quoted in the previous Appendix \ref{sec:landau_real}. The starting point is the coupled real system \eqref{eq:landau_uv}.

\subsection{Landau gauge}
\label{subsec:landau_gauge_real}
Choose the Landau gauge
\begin{equation}
  \mathbf A_L=(0,Bx,0).
  \label{eq:landau_gauge}
\end{equation}
Then
\begin{eqnarray}
  \hat H_R
  &= & -\frac{\hbar^2}{2M}\left(\partial_x^2+\partial_y^2\right)
    + \frac{q^2B^2x^2}{2M},
  \nonumber \\
  \hat H_I & = & \frac{\hbar qB}{M}\,x\,\partial_y.
  \label{eq:landau_gauge_HR_HI}
\end{eqnarray}

Because $\hat H_I\propto x\partial_y$, the real plane spanned by $\cos(ky)$ and $\sin(ky)$ is invariant for every fixed $k$. We therefore seek a stationary single-real-field mode of the form
\begin{equation}
  u(x,y,t)=f(x)\cos(ky-\Omega t),
  \label{eq:landau_real_ansatz}
\end{equation}
with $f(x)$ real. Its second initial datum is
\begin{equation}
    \dot u(x,y,t)=\Omega f(x)\sin(ky-\Omega t),
\end{equation}
so the Cauchy data already span the relevant two dimensional oscillator plane. If one temporarily writes the reconstructed quadrature in the same plane as
\begin{equation}
  v(x,y,t)=f(x)\sin(ky-\Omega t),
  \label{eq:v_landau_mode}
\end{equation}
then Eq.~\eqref{eq:landau_uv} closes on this ansatz and gives
\begin{equation}
\begin{aligned}
E f(x)
&=
\Bigg[
-\frac{\hbar^2}{2M}\frac{\dd^2}{\dd x^2}
+\frac{\hbar^2 k^2}{2M}
\\
&\qquad
+\frac{q^2B^2x^2}{2M}
-\frac{\hbar qBk}{M}\,x
\Bigg] f(x),
\\
E&:=\hbar\Omega .
\end{aligned}
\label{eq:landau_shifted_pre}
\end{equation}
Completing the square with the guiding-centre coordinate
\begin{equation}
  x_k := \frac{\hbar k}{qB},
  \label{eq:xk_landau}
\end{equation}
one obtains
\begin{equation}
  \left[
    -\frac{\hbar^2}{2M}\frac{d^2}{ \dd x^2}
    + \frac12 M\omega_c^2 (x-x_k)^2
  \right]f(x)
  = E\,f(x).
  \label{eq:landau_shifted_oscillator}
\end{equation}
Thus, each fixed-$k$ sector reduces to a one-dimensional harmonic oscillator centred
at $x_k$. The normalised eigenfunctions are
\begin{equation}
\begin{aligned}
f_n(x-x_k)
&=
\frac{1}{\pi^{1/4}\sqrt{2^n n!\,\ell_B}}\,
H_n\!\left(\frac{x-x_k}{\ell_B}\right)
\\
&\qquad\times
\exp\!\left[-\frac{(x-x_k)^2}{2\ell_B^2}\right].
\end{aligned}
\label{eq:landau_fn}
\end{equation}
and the spectrum is
\begin{equation}
  E_n = \hbar\omega_c\left(n+\frac12\right),\qquad n=0,1,2,\dots
  \label{eq:landau_levels_landau_gauge}
\end{equation}
independent of $k$. The continuous label $k$ therefore encodes the macroscopic Landau degeneracy through the guiding-center position $x_k$.

In the reduced single-component language, Eq.~\eqref{eq:landau_real_ansatz} together with its time derivative already contains the full mode information. The auxiliary field \eqref{eq:v_landau_mode} is merely a convenient local representative of the missing quadrature and may be reconstructed afterward from Eq.~\eqref{eq:v_mag}. For the eigenmode \eqref{eq:landau_real_ansatz}, this reconstruction returns exactly
Eq.~\eqref{eq:v_landau_mode}. A convenient explicit basis of real Landau modes is therefore:
\begin{equation}
  u_{n,k}(x,y,t)=f_n(x-x_k)\cos\!\left(ky-\frac{E_n t}{\hbar}\right).
  \label{eq:u_landau_mode}
\end{equation}

\subsection{Circular (symmetric) gauge}
\label{subsec:circular_gauge_real}

Choose the circular gauge
\begin{equation}
  \mathbf A_C = \frac{B}{2}(-y,x,0)
              = \frac{Br}{2}\,\hat{\boldsymbol\phi}.
  \label{eq:circular_gauge}
\end{equation}
In polar coordinates $(r,\phi)$, one has
\begin{eqnarray}
  \hat H_R
 & = &  -\frac{\hbar^2}{2M}
    \left(
      \partial_r^2 + \frac{1}{r}\partial_r + \frac{1}{r^2}\partial_\phi^2
    \right)
    + \frac{M\omega_c^2 r^2}{8},
  \nonumber \\
  \hat H_I &= & \frac{\hbar qB}{2M}\,\partial_\phi.
  \label{eq:circular_HR_HI}
\end{eqnarray}
Here, the natural invariant real planes are spanned by $\cos(m\phi)$ and $\sin(m\phi)$, with $m=0,1,2,\dots$. We therefore take
\begin{equation}
  u(r,\phi,t)=f(r)\cos(m\phi-\Omega t),
  \label{eq:circular_real_ansatz}
\end{equation}
so that
\begin{equation}
    \dot u(r,\phi,t)=\Omega f(r)\sin(m\phi-\Omega t).
\end{equation}
As in the Landau gauge, the second datum identifies the companion direction in the same
real two-plane. Writing that companion explicitly as
\begin{equation}
      v(r,\phi,t)=f(r)\sin(m\phi-\Omega t),
\end{equation}
and substituting into Eq.~\eqref{eq:landau_uv}, one finds the radial equation (with \(E:=\hbar\Omega\)):
\begin{equation}
\begin{aligned}
&\Bigg[
-\frac{\hbar^2}{2M}
\left(
\frac{\dd^2}{\dd r^2}
+\frac{1}{r}\frac{\dd}{\dd r}
-\frac{m^2}{r^2}
\right)
\\
&\qquad
+\frac{M\omega_c^2 r^2}{8}
-\frac{\hbar qB}{2M}\,m
\Bigg] f(r)
=
E f(r).
\end{aligned}
\label{eq:fock_darwin_real_branch}
\end{equation}
This is the radial Fock-Darwin equation at zero scalar confinement. For \(n_r=0,1,2,\dots\) its regular solutions are
\begin{equation}
\begin{aligned}
f_{n_r,m}(r)
&=
C_{n_r,m}
\left(\frac{r}{\ell_B}\right)^m
\exp\!\left(-\frac{r^2}{4\ell_B^2}\right)
\\
&\qquad\times
L_{n_r}^{m}\!\left(\frac{r^2}{2\ell_B^2}\right).
\end{aligned}
\label{eq:fock_darwin_radial}
\end{equation}
where $L_{n_r}^{m}$ are associated Laguerre polynomials.

For each $m>0$ there are two possible senses of rotation in the real $\bigl(\cos(m\phi),\sin(m\phi)\bigr)$ plane. It is therefore convenient to label the two normal-mode branches by $\sigma=\pm 1$, corresponding to the two counter-rotating motions in that plane. In complex notation, these two branches are the eigenfunctions $e^{\pm im\phi}$ of $L_z$: the labels $(m,\sigma)$ encode the usual magnetic quantum number $m_\ell=\sigma m$, with $\sigma$ redundant at $m=0$. For \(m=0,1,2,\dots\), the energies are (with the understanding that $m=0$ has only one independent angular function):
\begin{equation}
E_{n_r,m,\sigma}
=
\hbar\omega_c
\left(
n_r+\frac{m+1}{2}
-\frac{\sigma s m}{2}
\right).
\label{eq:circular_energies_sigma}
\end{equation}
Equivalently, defining the Landau-level index
\begin{equation}
  N := n_r + \frac{m-\sigma s m}{2}
  \in \{0,1,2,\dots\},
  \label{eq:N_index}
\end{equation}
the spectrum becomes
\begin{equation}
  E_N = \hbar\omega_c\left(N+\frac12\right),
  \label{eq:landau_levels_circular_gauge}
\end{equation}
exactly as in the Landau gauge. Thus, the symmetric-gauge quantum numbers $(n_r,m,\sigma)$ merely provide a different parametrisation of the same Landau levels.


\section{Sequential spin measurements in real variables}
\label{app:seqspin}

This appendix derives the projective-measurement and sequential-measurement statistics quoted at the end of Sec.~\ref{subsec:spin_rotations_real}, using only the real algebra \eqref{eq:Sigma_anticommutation}.

\label{subsec:sequential_spin_measurements}
For a measurement along the axis $\mathbf n$, the two projectors are
\begin{eqnarray}
  P^{(\mathbf n)}_{\pm}:=\frac12\left(I_4\pm\Sigma_{\mathbf n}\right), && (P^{(\mathbf n)}_{\pm})^T=P^{(\mathbf n)}_{\pm},\nonumber\\
    (P^{(\mathbf n)}_{\pm})^2&=&P^{(\mathbf n)}_{\pm}.
  \label{eq:spin_projectors}
\end{eqnarray}
For a normalised state $\Phi$, the Born rule is therefore,
\begin{equation}
  p(\pm|\mathbf n)=\Phi^T P^{(\mathbf n)}_{\pm}\Phi
  =\frac12\Bigl(1\pm \Phi^T\Sigma_{\mathbf n}\Phi\Bigr)
  =\frac12\bigl(1\pm \mathbf n\cdot\mathbf s\bigr).
  \label{eq:born_spin_real}
\end{equation}
Upon obtaining outcome $s\in\{+1,-1\}$, the post-measurement state is
\begin{equation}
  \Phi\longmapsto \Phi^{(\mathbf n)}_s
  :=\frac{P^{(\mathbf n)}_s\Phi}{\sqrt{\Phi^T P^{(\mathbf n)}_s\Phi}}.
  \label{eq:luders_update_real}
\end{equation}

Now consider two successive Stern-Gerlach measurements, first along $\mathbf n$ with outcome $s$, then along $\mathbf m$ with outcome $t$.  The conditional probability for the second result is
\begin{equation}
  p(t|\mathbf m;\,s,\mathbf n)
  =\frac{\Phi^T P^{(\mathbf n)}_s P^{(\mathbf m)}_t P^{(\mathbf n)}_s\Phi}
         {\Phi^T P^{(\mathbf n)}_s\Phi}.
  \label{eq:conditional_seq}
\end{equation}
Using only the algebra \eqref{eq:Sigma_anticommutation}, one finds
\begin{equation}
  P^{(\mathbf n)}_s P^{(\mathbf m)}_t P^{(\mathbf n)}_s
  =\frac12\bigl(1+s\,t\,\mathbf n\!\cdot\!\mathbf m\bigr)P^{(\mathbf n)}_s.
  \label{eq:PPP_identity}
\end{equation}
Hence
\begin{equation}
  p(t|\mathbf m;\,s,\mathbf n)=\frac12\bigl(1+s\,t\,\mathbf n\!\cdot\!\mathbf m\bigr).
  \label{eq:transition_probability}
\end{equation}
If $\theta$ is the angle between the two axes, this becomes
\begin{eqnarray}
  p(t=s|\mathbf m;\,s,\mathbf n)=\cos^2\!\frac{\theta}{2},\nonumber\\
  p(t=-s|\mathbf m;\,s,\mathbf n)=\sin^2\!\frac{\theta}{2}.
  \label{eq:cos2_sin2}
\end{eqnarray}
The familiar sequential-measurement law is thus recovered without ever leaving the real-variable formalism.

As a simple three-step example, measuring along $\mathbf n$, then $\mathbf m$, and finally again along $\mathbf n$ gives
\begin{equation}
  p(\text{final }=s\mid\text{initial }=s)
  =\frac12\bigl(1+(\mathbf n\!\cdot\!\mathbf m)^2\bigr),
  \label{eq:nmn_return_probability}
\end{equation}
which displays the disturbance produced by an intermediate measurement in a non-commuting basis.

\section{Orbital angular momentum and the Zeeman effect in the real phase-space formulation} \label{sec:zeeman_real}
 \label{app:zeeman}
 This appendix discusses orbital angular momentum as the spatial analogue of the
spin example treated in the main text. In the standard complex formulation, \(L_z\)-eigenstates are labeled by a magnetic quantum number \(m\), and the sign of \(m\) is encoded in the phase factor \(e^{im\phi}\). In a real single-field description, this phase factor is not available as a one-dimensional complex label. The question is therefore where the sign of orbital angular momentum is stored.

The answer is parallel to the momentum and spin cases. Rotational symmetry is already present in the reduced second-order equation, but a definite sense of angular motion is not contained in \(u\) alone. For \(m>0\), the real spherical harmonics \(\cos(m\phi)\) and \(\sin(m\phi)\) span an invariant real two-plane; the second Cauchy datum selects a sense of rotation within this plane. A weak magnetic field then resolves the two counter-rotating motions, giving the usual orbital Zeeman splitting.

\subsection{Rotations of the real second-order equation}
\label{subsec:rotations_real}
For a central potential $V(r)$ and a vanishing vector potential, the reduced real equation is
\begin{equation}
  \ddot u + \hbar^{-2}\hat H_0^{2}u = 0,
  \qquad
  \hat H_0 = -\frac{\hbar^2}{2M}\nabla^2 + V(r).
  \label{eq:H0_central_real}
\end{equation}
Because $\hat H_0$ commutes with spatial rotations, Eq.~\eqref{eq:H0_central_real} is invariant under $u(\mathbf r,t)\mapsto u(R^{-1}\mathbf r,t)$.  The infinitesimal generators acting on real scalar fields are the real differential operators
\begin{eqnarray}
  D_x &:=& y\partial_z-z\partial_y,\qquad
  D_y := z\partial_x-x\partial_z,\nonumber\\
  D_z &:=& x\partial_y-y\partial_x = \partial_\phi .
  \label{eq:D_generators}
\end{eqnarray}
They are skew-adjoint and satisfy
\begin{equation}
  [D_i,D_j] = -\varepsilon_{ijk}D_k,
  \label{eq:so3_D}
\end{equation}
and generate the ordinary rotation group on the Cauchy data $(u,\dot u)$.

At this level, nothing like a self-adjoint one-component $L_z$ has yet appeared.  The reduced real equation knows only that rotations act orthogonally on the space of initial data.  A definite sense of azimuthal motion is encoded not in $u$ alone but in the correlation between $u$ and the second datum $\dot u$.

What survives as a genuine real symmetric operator on the configuration space is the quadratic Casimir
\begin{equation}
  L^2 := -\hbar^2\bigl(D_x^2+D_y^2+D_z^2\bigr)
       = -\hbar^2\Delta_{S^2}.
  \label{eq:L2_real}
\end{equation}
Accordingly, the natural angular basis of the real second-order theory is organized by eigenfunctions of $L^2$, while the information ordinarily attached to the sign of $m$ is shifted into the dynamics inside degenerate real two-planes.

\subsection{Real spherical harmonics and the $m$-planes}
\label{subsec:real_spherical_harmonics}
For $m=0$, one may take the usual real spherical harmonic $Y_{\ell 0}=N_{\ell 0}\,P_\ell(\cos\theta)$; for $m=1,\dots,\ell$, a convenient real basis is
\begin{eqnarray}
  Y_{\ell m}^{c}&:= &N_{\ell m}\,P_\ell^m(\cos\theta)\cos(m\phi),
  \nonumber \\
  Y_{\ell m}^{s} & := & N_{\ell m}\,P_\ell^m(\cos\theta)\sin(m\phi),
  \label{eq:real_harmonics_def}
\end{eqnarray}
replacing the complex pair $e^{\pm im\phi}$. These are eigenfunctions of the real symmetric Casimir $L^2$ with the usual eigenvalue $\hbar^2\ell(\ell+1)$, and the azimuthal generator maps each pair into itself, $D_zY_{\ell m}^{c}=-m\,Y_{\ell m}^{s}$ and $D_zY_{\ell m}^{s}=m\,Y_{\ell m}^{c}$. Hence, for each $m>0$, the span $\mathcal H_{\ell m}^{(\mathrm{real})}:=\mathrm{span}\{Y_{\ell m}^{c},Y_{\ell m}^{s}\}$ is a two-dimensional invariant plane for rotations about the $z$-axis, on which, in the ordered basis $(Y_{\ell m}^{c},Y_{\ell m}^{s})^T$,
\begin{equation}
  D_z\big|_{\mathcal H_{\ell m}^{(\mathrm{real})}}
  = mJ,
  \quad
  J:=\begin{pmatrix}0&-1\\[2pt]1&0\end{pmatrix},
  \quad J^2=-{\mathbb 1}.
  \label{eq:Dz_m_plane}
\end{equation}
The label $m$ is thus retained geometrically: it is the angular velocity with which the pair rotates under a physical $z$-rotation.

Now fix an eigenstate of the central Hamiltonian $\hat H_0$ with radial profile $R_{n\ell}(r)$. For $m>0$, a general real solution within this orbital sector is
\begin{eqnarray}
  u(\mathbf r,t)
  & = & R_{n\ell}(r)\bigl[a(t)\,Y_{\ell m}^{c}(\theta,\phi)+b(t)\,Y_{\ell m}^{s}(\theta,\phi)\bigr],
  \nonumber \\
  U(t)& := & \binom{a(t)}{b(t)},
  \label{eq:real_orbital_mode_ansatz}
\end{eqnarray}
and projecting Eq.~\eqref{eq:H0_central_real} onto the invariant plane gives
\begin{equation}
  \ddot U + \omega_{n\ell}^{2}U = 0,
  \qquad
  \omega_{n\ell}:=\frac{|E_{n\ell}|}{\hbar}.
  \label{eq:orbital_plane_oscillator}
\end{equation}
The two distinguished normal motions are the counter-rotating solutions
\begin{equation}
  U_{\pm}(t)=R(\pm\omega_{n\ell}t)U_0,
  \quad
  R(\alpha):=\cos\alpha\,I_2+\sin\alpha\,J,
  \label{eq:orbital_plane_rotating_modes}
\end{equation}
with initial angular velocity $\dot U_{\pm}(0)=\pm\omega_{n\ell}J U_0$.

This is the key real-space picture.  The pair of states that would ordinarily be labeled by $\pm m$ is replaced by a single real $m$-plane, together with a choice of initial angular velocity in that plane.

\subsection{Orbital angular momentum on the real phase-space}
\label{subsec:zeeman_hamiltonian_real}

To attach an orbital angular momentum to a reduced real state, one must, exactly as in the translation problem, restore the canonical momentum
\begin{equation}
  \pi = \hbar^2\hat H_0^{-1}\dot u
  \label{eq:pi_orbital_section}
\end{equation}
on the subspace where $\hat H_0$ is invertible.  Introduce the real phase-space field
\begin{equation}
    \Phi:=\binom{u}{\pi/\hbar}.
\end{equation} 
The natural rotational functional is then
\begin{equation}
\begin{aligned}
\langle L_z\rangle
&:=
-\hbar\,\Omega\bigl(\Phi,D_z\Phi\bigr)
\\[-0.2em]
&=
\int_{\mathbb R^3}\dd^3r\,
\bigl(u\,D_z\pi-\pi\,D_z u\bigr).
\end{aligned}
\label{eq:Lz_functional_real}
\end{equation}
Whenever $[\hat H_0,D_z]=0$, this quantity is conserved by the real Hamiltonian flow. It is the direct rotational analogue of the momentum functional Eq.~\eqref{eq:momentum_functional}.

For the rotating normal modes, the result is transparent.  Let
\begin{equation}
  u_{\pm}(\mathbf r,t)
  = f_{n\ell m}(r,\theta)\cos\!\left(m\phi \mp \omega_{n\ell}t+\alpha\right),
  \label{eq:orbital_real_rotating_mode}
\end{equation}
with $f_{n\ell m}(r,\theta)=R_{n\ell}(r)N_{\ell m}P_\ell^m(\cos\theta)$ normalized in the usual way.  On the branch connected to positive $E_{n\ell}$, the canonical momentum is
\begin{equation}
  \pi_{\pm}(\mathbf r,t)
  = \pm\hbar\,f_{n\ell m}(r,\theta)\sin\!\left(m\phi \mp \omega_{n\ell}t+\alpha\right).
  \label{eq:orbital_real_rotating_pi}
\end{equation}
Substituting into Eq.~\eqref{eq:Lz_functional_real} gives
\begin{equation}
  \langle L_z\rangle_{\pm} = \pm \hbar m
  \label{eq:orbital_Lz_pm}
\end{equation}
per unit norm. Thus, the sign of the orbital angular momentum is encoded in the real phase-space formulation by the sense in which the solution rotates inside the $(\cos m\phi, \sin m\phi)$ plane. For $m=0$, the plane collapses to a line, and no such distinction remains.

\subsection{Real form of the Zeeman splitting}
\label{subsec:real_zeeman_splitting}

\begin{figure}[!htbp]
 \centering
\begin{tikzpicture}[>=Stealth,scale=0.78]
  \draw[->] (-1.8,0) -- (1.8,0) node[below] {$a$};
  \draw[->] (0,-1.8) -- (0,1.8) node[left] {$b$};
  \draw[->,thick] (-15:1.3) arc (-15:225:1.3);
  \node[font=\footnotesize] at (0,2.15) {$m_\ell{=}{+}m$:\ $\Omega_+$};
  \draw[->,thick] (200:0.8) arc (200:-40:0.8);
  \node[font=\footnotesize] at (0,-2.15) {$m_\ell{=}{-}m$:\ $\Omega_-$};
  \draw (3.2,0) -- (4.4,0);
  \node[above,font=\footnotesize] at (3.8,0) {$E_{n\ell}$};
  \draw (5.4,0.9) -- (6.6,0.9);
  \draw (5.4,-0.9) -- (6.6,-0.9);
  \draw[dashed,gray] (4.4,0) -- (5.4,0.9);
  \draw[dashed,gray] (4.4,0) -- (5.4,-0.9);
  \node[above,font=\footnotesize] at (6.0,0.9) {$E_{n\ell}+\hbar\omega_L m$};
  \node[below,font=\footnotesize] at (6.0,-0.9) {$E_{n\ell}-\hbar\omega_L m$};
  \node[font=\footnotesize] at (3.8,-2.15) {$B=0$};
  \node[font=\footnotesize] at (6.0,-2.15) {$B>0$};
\end{tikzpicture}
\caption{Real form of the orbital Zeeman effect in an $(n,\ell,m)$ plane with $m>0$, drawn for an electron ($\omega_L=eB/2M>0$). Left: the two counter-rotating normal modes in the coefficient plane $(a,b)$ of $(\cos m\phi,\sin m\phi)$, Eq.~\eqref{eq:general_m_plane}. With $u=\operatorname{Re}\psi$, counterclockwise motion corresponds to $m_\ell=+m$ and rotates at $\Omega_+=\omega_{n\ell}+\omega_L m$, clockwise motion to $m_\ell=-m$ at $\Omega_-=\omega_{n\ell}-\omega_L m$, Eq.~\eqref{eq:zeeman_frequencies}. Right: the field resolves the degenerate real $m$-doublet into the branches $E_{n\ell}\pm\hbar\omega_L m$, Eq.~\eqref{eq:orbital_zeeman_shift}; the $m=0$ sector is unsplit.}
\label{fig:zeeman_mplane}
\end{figure}
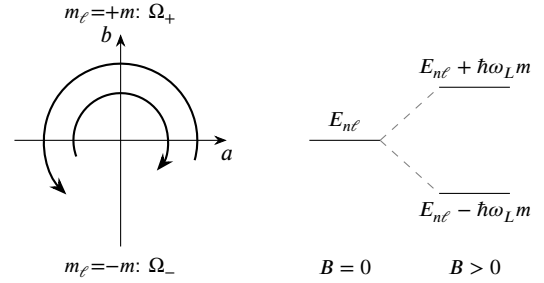

We now place the same central problem in a weak uniform magnetic field
$\mathbf B=B\hat{\mathbf z}$. In the symmetric gauge $\mathbf A=\tfrac12\,\mathbf B\times\mathbf r$ (the circular gauge of Appendix~\ref{app:landau}), the magnetic field enters the reduced real equation through the gyroscopic term derived in Sec.~\ref{sec:landau_real}.  To linear order in $B$, one may neglect the diamagnetic piece and keep only
\begin{equation}
  \hat H_I = -\hbar\omega_L D_z,
  \qquad
  \omega_L := -\frac{qB}{2M}.
  \label{eq:orbital_zeeman_term}
\end{equation}
Because both $\hat H_0$ and $D_z$ preserve the $(n,\ell,m)$ plane, the general magnetic second-order equation \eqref{eq:gyro_form} reduces on that plane to
\begin{equation}
  \ddot U + 2\omega_L m\,J\dot U
  + \bigl(\omega_{n\ell}^{2}-\omega_L^{2}m^{2}\bigr)U = 0,
  \label{eq:gyroscopic_m_doublet}
\end{equation}
again with $U=(a,b)^T$.  The magnetic field, therefore, does not assign a new real eigenvalue to a one-component $L_z$ operator.  It lifts the degeneracy by acting as a gyroscopic coupling on the pre-existing real $m$-plane.

Seeking uniformly rotating solutions $U(t)=R(\Omega t)\,U_0$, with $R(\alpha)$ as defined in Eq.~\eqref{eq:orbital_plane_rotating_modes}, gives $\bigl(\Omega+\omega_L m\bigr)^2=\omega_{n\ell}^2$; therefore, the two rotation senses acquire the frequencies
\begin{equation}
  \Omega_{\pm}=\omega_{n\ell}\pm\omega_L m.
  \label{eq:zeeman_frequencies}
\end{equation}
Restoring the sign of the unperturbed level by selecting the branch that reproduces the spectral sign of $\hat H_0$, as discussed in Sec.~\ref{ho} and after Eq.~\eqref{eq:KsqHsq}, the corresponding energy branches are:
\begin{equation}
  E_{n\ell m}^{(\pm)} = E_{n\ell}\pm \hbar\omega_L m.
  \label{eq:orbital_zeeman_shift}
\end{equation}
For an electron ($q=-e$, so that $\omega_L=eB/2M>0$ and $\hbar\omega_L=\mu_B B$), this is the usual linear orbital shift $\Delta E = \mu_B B\,m$.  The $m=0$ sector has no linear splitting, while every $m>0$ real doublet separates into the two counter-rotating branches that, in the standard bookkeeping, are associated with magnetic quantum numbers $\pm m$. Figure~\ref{fig:zeeman_mplane} illustrates both the counter-rotating $m$-plane modes and the resulting level scheme.

\subsection{Including spin-$\tfrac12$}
\label{subsec:spin_zeeman_real}

Once the internal two-state degree of freedom of Sec.~\ref{sec:spin_half_real} is included, the magnetic field acts in two distinct places: the orbital motion is split by Eq.~\eqref{eq:orbital_zeeman_shift}, while the internal spin sector acquires the real symmetric Zeeman term $\hat H_Z^{(\mathrm{spin})}\ \widehat{=}\ g_s\mu_B B\,\Sigma_z/2$, with $\Sigma_z$ as in Sec.~\ref{subsec:real_pauli}, which shifts the two spin channels by $g_s\mu_B B\,m_s$ and $m_s=\pm\tfrac12$. In the weak-field regime, one therefore obtains
\begin{equation}
  \Delta E = \hbar\omega_L m + g_s\mu_B B\,m_s ,
  \label{eq:orbital_plus_spin_shift}
\end{equation}
i.e., $\Delta E=\mu_B B\,(m+g_s m_s)$ for an electron. If spin-orbit coupling is present, one diagonalizes the corresponding finite-dimensional real invariant subspace of total angular momentum and again recovers the standard Land\'e form $\Delta E=\mu_B\,g_j\,m_j\,B$.

\subsection{What the \(m\)-planes encode}
\label{subsec:zeeman_interpretation}

This section displays, in a compact atomic setting, the same lesson already encountered for Landau levels and for spin.  Rotational symmetry is fully present in the reduced real equation, but its bookkeeping is different.  The one-dimensional $L_z$ eigenstates of the conventional language are replaced by real invariant $m$-planes. The second initial datum chooses a rotation sense inside such a plane, and that choice is exactly what carries the sign of the orbital angular momentum.

The Zeeman effect then becomes especially transparent.  A magnetic field along the $z$-axis does not create angular momentum out of nothing, nor does it force one to return to a complex one-component description.  It simply removes the degeneracy between the two counter-rotating motions already present in each real $m>0$ plane. Orbital and spin splitting are treated on the same footing: in both cases, the field resolves an underlying real two-dimensional rotational structure.

\section{Relation to the de Broglie-Bohm theory and the quantum potential}
A useful comparison is provided by the Madelung and de Broglie-Bohm rewriting of the Schrödinger equation. This is again a real-variable description, but of a different kind from the reduced single-field formulation studied above. The de Broglie-Bohm theory does not eliminate one quadrature in favor of a second-order equation for \(u\). Instead, it rewrites the two local quadratures in polar variables, namely the density \(\rho\) and phase \(S\).

This comparison is instructive because it displays a complementary trade-off. In the local real phase-space variables, the passage from \((u,\pi/\hbar)\) to \((\sqrt{\rho},S/\hbar)\) is just a change of coordinates on the same internal plane. If one insists on the reduced variables \((u,\dot u)\), however, the same polar quantities become non-local functionals because the missing quadrature must first be reconstructed by \(\pi=\hbar^2\hat H^{-1}\dot u\).

We therefore start from the local canonical variables \((u,\pi)\), and only
afterward reduce to the single real field \(u\). At each point in space, the pair $\Phi=(u,\pi/\hbar)^T$ forms a two-component real vector whose Euclidean norm is precisely the Born density, Eq.~\eqref{eq:prob_density_local}. Away from nodes, where $\rho(\bm r,t)>0$, one may therefore introduce polar coordinates on this internal plane by
\begin{subequations}
\label{eq:bohm_polar_from_upi}
\begin{align}
 u(\bm r,t)
 &=
 \sqrt{\rho(\bm r,t)}\,\cos\theta(\bm r,t),
 \\
 \frac{\pi(\bm r,t)}{\hbar}
 &=
 \sqrt{\rho(\bm r,t)}\,\sin\theta(\bm r,t),
 \\
 S(\bm r,t)&=\hbar\,\theta(\bm r,t).
\end{align}
\end{subequations}
Thus, the polar variables of Bohm's or Madelung's formulation \cite{Bohm1952,madelung27} are nothing but non-linear coordinates on the same local phase-space plane. The change of variables is, in fact, canonical up to normalisation, $\dd u\wedge\dd\pi=\tfrac12\,\dd\rho\wedge\dd S$, so that density and phase form a conjugate pair on the same phase space; the starting point of the exact-uncertainty and configuration-space-ensemble approaches to reconstructing the Schrödinger equation \cite{hall2002exact}. In particular, the real current from Eq.~\eqref{eq:prob_current_local} becomes
\begin{equation}
\bm j
=
\frac{1}{m}\Big(u\nabla\pi-\pi\nabla u\Big)
=
\rho\,\frac{\nabla S}{m},
\label{eq:bohm_current_from_upi}
\end{equation}
so the de Broglie-Bohm velocity field (``guidance equation'') is simply
\begin{equation}
\bm v_B = \frac{\bm j}{\rho} = \frac{\nabla S}{m}.
\label{eq:bohm_velocity_from_upi}
\end{equation}
At the level of the local first-order theory, therefore, the present real phase-space formulation and the de Broglie-Bohm one encode the same two quadratures, but in different coordinates: Cartesian coordinates $(u,\pi/\hbar)$ and polar coordinates $(\sqrt{\rho},S/\hbar)$.

Substituting the polar parametrisation \eqref{eq:bohm_polar_from_upi} into the Schr\"odinger dynamics for $\hat H=-\frac{\hbar^2}{2m}\nabla^2+V$ yields the two equations that were the starting point of Bohm's 1952 work \cite{Bohm1952}: the Madelung continuity equation $\partial_t\rho+\nabla\cdot(\rho\,\nabla S/m)=0$ and the quantum Hamilton-Jacobi equation $\partial_t S+(\nabla S)^2/(2m)+V+Q=0$, with the quantum potential $Q=-(\hbar^2/2m)\,\nabla^2\sqrt{\rho}/\sqrt{\rho}$. Supplemented with the guidance law \eqref{eq:bohm_velocity_from_upi}, these equations assign trajectories to particles \cite{Bohm1952}. In this form, the density $\rho$ and the velocity field $\bm v_B$ are local; the price is that the evolution in $(\rho,S)$ is non-linear and becomes singular on the nodal set $\rho=0$.

The contrast with the present real reduction appears when one insists on describing the dynamics by the single field $u$ and its second initial datum $\dot u$. The second quadrature is then no longer an independent local field but must be reconstructed through the non-local map \eqref{eq:pi_canonical}; consequently, $\rho=u^2+\pi^2/\hbar^2=u^2+(\hbar\hat H^{-1}\dot u)^2$ and, away from nodes,
\begin{equation}
S=\hbar\,\arctan\left(\frac{\pi/\hbar}{u}\right)=\hbar\,\arctan\left(\frac{\hbar\hat H^{-1}\dot u}{u}\right).
\label{eq:bohm_phase_from_udot}
\end{equation}
Hence, the reduced real description keeps the equation of motion linear but transfers part of the structure into the non-local relation between $u$ and its canonically conjugate quadrature.

The comparison may therefore be stated succinctly. The de Broglie-Bohm theory maintains the amplitude and phase locally at the cost of rewriting the wave dynamics in non-linear polar variables and, for many-body systems, makes the dependence on the full configuration-space wave function explicit. The present second-order reformulation proceeds in the opposite direction: it preserves a linear real wave equation for a single field, but the Born density and internal phase become non-local functionals when expressed solely in terms of $(u,\dot u)$. The standard complex Schrödinger equation avoids both complications only because it keeps both quadratures on equal footing from the outset.

This polar-coordinate comparison concerns another real rewriting of the same complex Schrödinger dynamics. We finally turn to two different uses of real structures that might otherwise be confused with the present construction: operational real Hilbert-space quantum theory and the Majorana-Dirac real formulation.

\section{The Majorana-Dirac real formulation}
 \label{sec:majorana_dirac_comparison}
The preceding subsections concerned operational real Hilbert-space quantum theory and the recent no-go results. A different comparison is provided by the Majorana formulation of relativistic spin-\(\frac12\) theory \cite{Aste2010,Majorana1937,Pal2011}. It is useful to distinguish this real formulation from the non-relativistic reconstruction developed here.

A Majorana fermion is defined by the Lorentz-covariant Majorana condition
\begin{equation}
  \Psi^c=\Psi ,
\end{equation}
where \(\Psi^c\) denotes the charge-conjugated spinor. In a Majorana representation, this condition can be expressed by choosing the spinor components to be real-valued. The fact that the components are real-valued, however, is a feature of that representation, not the invariant definition of a Majorana fermion.

In a Majorana representation, one may choose the Dirac matrices such that
\begin{equation}
  (\tilde\gamma^\mu)^*=-\tilde\gamma^\mu,
  \label{eq:Majorana_basis_gamma}
\end{equation}
so that $i\tilde\gamma^\mu$ has real-valued entries. The free Dirac equation
\begin{equation}
  (i\tilde\gamma^\mu\partial_\mu-m)\tilde\Psi=0
  \label{eq:Majorana_real_first_order}
\end{equation}
is then a real first-order linear system, and in this basis, the Majorana condition simply becomes
\begin{equation}
  \tilde\Psi^*=\tilde\Psi.
  \label{eq:Majorana_basis_reality}
\end{equation}

This should be contrasted with the present non-relativistic construction. Here one begins with the decomposition $\psi=u+iv$, which yields a real first-order Hamiltonian flow for the pair $(u,v)$. Only after eliminating one quadrature does one arrive at a single real equation
\begin{equation}
  \ddot u+\hbar^{-2}\hat H^{\,2}u=0,
\end{equation}
together with the reconstruction $v=\hbar \hat H^{-1}\dot u$.

Thus, the single-component real Schrödinger reformulation is second order in time and, at the level of the Born norm, generically non-local in the reduced variables $(u,\dot u)$. By contrast, the Majorana equation remains first order and local; what replaces the explicit imaginary unit is not a non-local reconstruction map but the real Clifford algebra carried by the matrices $\tilde\gamma^\mu$. This is the strategy developed systematically in the geometric-algebra tradition, where the unit imaginary of the Dirac and Pauli theories is replaced throughout by real bi-vectors with direct geometric meaning \cite{hestenes1967real}; like the present reconstruction, it removes the symbol $i$ by \emph{retaining} real structure rather than eliminating it.

In this sense, the Majorana construction is much closer to the \emph{two-quadrature} real-variable formulation than to the \emph{single-component real} reduction. Indeed, a Dirac field can be decomposed into two Majorana fields:
\begin{equation}
  \Psi_D=\frac{1}{\sqrt2}\bigl(\Psi_1+i\Psi_2\bigr),
  \qquad
  \Psi_1^c=\Psi_1,\quad \Psi_2^c=\Psi_2,
  \label{eq:Dirac_as_two_Majoranas}
\end{equation}
with inverse relations
\begin{equation}
  \Psi_1=\frac{1}{\sqrt2}\bigl(\Psi_D+\Psi_D^c\bigr),
  \qquad
  \Psi_2=\frac{1}{i\sqrt2}\bigl(\Psi_D-\Psi_D^c\bigr).
  \label{eq:Majorana_from_Dirac}
\end{equation}
Hence, a charged complex Dirac field is equivalent to a doubled real description, just as in the non-relativistic theory, the pair $(u,v)$ provides the natural real phase-space description of a complex wave function.
 
The comparison becomes particularly sharp in the presence of electromagnetism. For a charged Dirac field, one has
\begin{equation}
  \left(
    i\gamma^\mu(\partial_\mu+i e A_\mu)-m
  \right)\Psi=0 .
\end{equation}
The ordinary local \(U(1)\) gauge transformation
\begin{equation}
  \Psi(x)\mapsto e^{-ie\eta(x)}\,\Psi(x)
\end{equation}
does not preserve the Majorana condition \(\Psi^c=\Psi\), unless the electric charge \(e\) vanishes. Equivalently, a truly Majorana field is necessarily electrically neutral \cite{Aste2010,Pal2011}. A charged Dirac field can therefore be written in terms of two Majorana fields, but the electromagnetic coupling rotates these two real fields into one another rather than preserving either one separately. This parallels the non-relativistic situation discussed above: for Hamiltonians that are real in the position representation, one can reduce the dynamics to a single real field. However, once ordinary magnetic coupling is present, the strictly one-field real-variable formulation is obstructed, and one is driven back to a doubled real structure with gyroscopic couplings.

The comparison, therefore, suggests a precise conclusion. The Majorana formulation does not show that quantum theory can generically be reduced to a single real wave function. Rather, it shows that in special relativistic systems of neutral spin-\(\frac12\) particles, the explicit imaginary unit may be absorbed into a real matrix representation of the Clifford algebra. In the present non-relativistic setting, the corresponding structure is carried instead by the pair of real quadratures or, after reduction, by the non-local phase-space metric. The broader lesson is the same in both cases: the explicit symbol \(i\) can be hidden, but the algebraic and symplectic structure it encodes must reappear somewhere.

\section*{Acknoledgement}  We thank Christian H\"olbling for his collaboration during the early phase of this project, in particular for his part in the critical analysis of the reduced single-field formulation -- the treatment of magnetic coupling and the status of the probability density -- from which the questions addressed in this paper emerged.

\section*{Declaration of competing interest}
The authors declare that they have no known competing financial interests or personal relationships that could have appeared to influence the work reported in this paper.
\section*{Data availability}
No new data were generated or analysed in this study.
\section*{Funding}
This research did not receive any specific grant from funding agencies in the public, commercial, or not-for-profit sectors.
\section*{Author contributions}
Oliver Passon conceived and initiated the project. He carried out the historical work on Schrödinger, Ehrenfest, Pauli, and Bohm, analyzed Chen's real-valued formulation, and identified the difficulties of the reduced description that this paper addresses, namely the treatment of magnetic coupling, the status of the probability density and of momentum, and the sensitivity to the choice of the energy zero. He also developed the analysis of time-reversal symmetry and Kramers degeneracy in the real formulation. Bernd Rosenow clarified the operator-level and state-level non-locality of the reconstruction map, developed the Hamiltonian lift and the real phase-space formulation together with the gauging of the internal SO(2) symmetry, and worked out the applications and examples, from wave packets and Landau levels to spin and composite systems, together with their bearing on the recent no-go debate. The framing of the paper, the interpretation of the results, and the revisions were developed in close collaboration. Both authors contributed to the manuscript, reviewed the complete text, and approved the final version.
\section*{Declaration of generative AI and AI-assisted technologies in the writing process}
During the preparation of this work, the authors used ChatGPT 5.4 and Fable 5 to assist with performing calculations, drafting text, and generating LaTeX code for the illustrative diagrams. After using these tools, the authors rigorously reviewed and verified all calculations, compiled the code locally, and edited the text as needed. The authors take full responsibility for the content and accuracy of the final publication.


\bibliographystyle{cas-model2-names}
\bibliography{references}

\end{document}